\documentclass[12pt,a4wide]{article}
\usepackage{amsmath,amsthm,epsfig,graphicx}
\usepackage{amsmath,amsthm,amssymb}
\usepackage{amsfonts}
\usepackage{graphicx}
\usepackage{latexsym}
\usepackage[mathcal]{eucal}
\usepackage{epsfig}
\usepackage{verbatim,url}
\usepackage{float}
\newtheorem{theorem}{Theorem}[section]
\newtheorem{corollary}{Corollary}[section]
\newtheorem{definition}{Definition}[section]
\newtheorem{proposition}[theorem]{Proposition}
\newtheorem{remark}[theorem]{Remark}

\numberwithin{equation}{section}

\title{
{\bf \Large Static vs adapted optimal execution strategies in two benchmark trading models}
}
\author{Damiano Brigo\\ Dept. of Mathematics\\ Imperial College London\\  {\normalsize {\tt damiano.brigo@imperial.ac.uk}}  \and  Cl\'ement Piat \\ Dept. of Mathematics \\ Imperial College London \\  {\normalsize {\tt clement.piat15@imperial.ac.uk }}}

\date{\small{First Version: 7 Sept. 2016. This version: \today}}

\begin{document}

\maketitle

\thispagestyle{empty}

\begin{abstract}
We consider the optimal solutions to the trade execution problem in the two different classes of i) fully adapted or adaptive and ii) deterministic or static strategies, comparing them. We do this in two different benchmark models. The first model is a discrete time framework with an information flow process, dealing with both permanent and temporary impact, minimizing the expected cost of the trade. The second model is a continuous time framework where the objective function is the sum of the expected cost and a value at risk (or expected shortfall) type risk criterion. Optimal adapted solutions are known in both frameworks from the original works of Bertsimas and Lo (1998) and Gatheral and Schied (2011). In this paper we derive the optimal static strategies for both benchmark models and we study quantitatively the improvement in optimality when moving from static strategies to fully adapted ones. We conclude that, in the benchmark models we study, the difference is not relevant, except for extreme unrealistic cases for the model or impact parameters. This indirectly confirms that in the similar framework of Almgren and Chriss (2000) one is fine deriving a static optimal solution, as done by those authors, as opposed to a fully adapted one, since the static solution happens to be tractable and known in closed form. 
\end{abstract}

\bigskip

{\bf AMS Classiﬁcation Codes}: 60H10, 60J60, 91B70; 

{\bf JEL Classiﬁcation Codes:} C51, G12, G13

\bigskip



\medskip




\pagestyle{myheadings}
\markboth{}{{\footnotesize  D. Brigo, C. Piat: Static vs adapted optimal solutions in benchmark trade execution models}}

\section{Introduction}
A basic stylized fact of trade execution is that when a trader buys or sells a large amount of stock in a restricted amount of time, the market naturally tends to move in the opposite direction. If one assumes an unaffected price dynamics for the traded asset, trading activity will impact this price and lead to an affected price.  Supply and demand based analysis says that if a trader begins to buy large amounts, other traders will notice and the affected price will tend to increase. Similarly, if one begins to sell large amounts, the affected price will tend to decrease. This is particularly important when the market is highly illiquid, since in that case no trade goes unnoticed. The goal of optimal execution, or more properly optimal scheduling, is to find how to execute the order in a way such that the expected profit or cost is the best possible, taking into account the impact of the trade on the affected price.

As far as we are concerned in this paper, there are two main categories of trading strategies: deterministic, also called static in the execution literature jargon, and adapted, or adaptive. We will use static / deterministic and adapted / adaptive interchangeably. Deterministic strategies are set before the execution, so that they are independent of the actual path taken by the price. They only rely on information known initially. Adapted strategies are not known before the execution. The amount executed at each time depends on all information known up to this time. Clearly market operators, in reality, will monitor market prices and trade based on their evolution, so that the adapted strategy is the more natural one. However, in some models it is much harder to find an optimal trading strategy in the class of adapted strategies than in the class of static ones. 
\\

In 1998, Bertsimas and Lo \cite{bertsimas1998} have defined the best execution as the strategy that minimizes the expected cost of trading over a fixed period of time. 
They derive the optimal strategy by using dynamic programming, which means that they go backwards in time. The optimal solution is therefore sought in the class of adapted strategies, as is natural from backward induction, but is found to be deterministic anyway. However, once an information process is added, influencing the affected price, the optimal solutions are adapted and no longer static. This approach minimizes the expected trade cost only, without including any risk in the criterion to be optimized. In particular, the criterion does not take into account the variance of the cost function.

Two years later, Almgren and Chriss \cite{almgren2000} consider the minimization of an objective function that is the sum of the expected execution cost and of a cost-variance risk criterion. 
Unlike the previous model, this setting includes in the criterion the possibility to penalize large variability in the trading cost. To solve the resulting mean-variance optimization, Almgren and Chriss  assume the solution to be deterministic from the start. This allows them to obtain a closed-form solution. This solution, however, is only the best solution in the class of static strategies, and not in the broader and more natural class of adapted ones. 

Gatheral and Shied \cite{gatheral2011} later solve a similar problem, the main difference being that they assume a more realistic model for the unaffected price.
Gatheral and Schied derive an adapted solution by using an alternative risk criterion, the time-averaged value-at-risk function. They obtain closed-form expressions for the strategy and the optimal cost. The solution is not static. However, this does not seem to lead to a solution that is very different, qualitatively, from the static one. Indeed, Brigo and Di Graziano (2014), adding a displaced diffusion dynamics,  find that in many situations only the rough statistics of the signal matter in the class of simple regular diffusion models \cite{brigodigraziano}. In this paper we will compare the static and fully adapted solutions in detail. 

Since the solutions obtained in the setting of Almgren and Chriss \cite{almgren2000} are deterministic, they may be sub-optimal in the set of fully adapted solutions under a cost-variance risk criterion, so several papers have attempted to find adapted solutions by changing the framework slightly. This allows one to take the new price information into account during the execution, and to have more precise models. For example, in 2012 Almgren \cite{almgren2012} assumes that the volatility and liquidity are random. He numerically obtains adapted results under these assumptions. Almgren and Lorenz \cite{almgren2011} obtain adapted solutions by using an appropriate dynamic programming technique.
\\

Similarly, in this paper we will focus on what one gains from adopting a more general adapted strategy over a simple deterministic strategy in the classic discrete time setting of Bertsimas and Lo \cite{bertsimas1998} with information flow and in the continous time setting of Gatheral and Schied with time-averaged value-at-risk criterion \cite{gatheral2011}. 

The paper is structured as follows.
In Section \ref{sec:discrete} we will introduce the discrete time model by Bertsimas and Lo, looking at the cases of temporary and permanent market impact on the unaffected price, and including the solution for the case where the price is also affected by an information flow process. We will derive and study the optimal static and fully adapted solutions and compare them, quantifying in a few numerical examples how much one gains from going fully adapted.

In Section \ref{sec:cont} we will introduce the continuous time model as in Gatheral and Schied, allowing for both temporary and permanent impact and for a risk criterion based on value at risk. 
We will report the optimal fully adapted solution as derived in \cite{gatheral2011} and we will derive the optimal static solution using a calculus of variation technique, similar for example to the calculations in \cite{digraziano2016}. We will compare the two solutions and optimal criteria in a few numerical examples, to see again how much one gains from going fully adapted.

Section \ref{sec:conc} concludes the paper, summarizing its findings, and points to possible future research directions.

%

\section{Discrete time trading with information flow}\label{sec:discrete}

\subsection{Model formulation with cost based criterion}

Let $X_t$ be the number of units left to execute at time $t$, such that $X_0=X$ is the initial amount and $X_T=0$ at the final time $T$. 
In this section we consider a buy order, so that the purpose of the strategy is to buy an amount $X$ of asset by time $T$, minimizing the expected cost of the trade. 
The amount to be executed during the time interval $[t,t+1)$ is $\Delta V_t := X_t - X_{t+1}$. 
We expect $\Delta V_t$ to be non-negative, since we would like to implement a pure buy program. However we do not impose a constraint of positivity on $\Delta V$, so that the optimal solution, in principle, might consider a mixed buy/sell optimal strategy.

Since the problem is in discrete time, it is only updated every period so we will assume that the price does not change between two update times.

With that in mind, we assume that the unaffected mid-price process $\widetilde{S}$ is given by \\
\begin{align}\label{unaffPrice}
\widetilde{S}_t &=  \widetilde{S}_{t-1} + \gamma Y_t + \sigma \widetilde{S}_0 \Delta W_{t-1}, \\
Y_t &= \rho Y_{t-1} + \sigma_Y \Delta Z_{t-1},
\end{align}
where the information coefficient $\gamma$, and the volatilities $\sigma$ and $\sigma_Y$ are positive constants, $W$ and $Z$ are independent standard Brownian motions and the parameter $\rho$ is in $(-1,1)$.  We define $\Delta W_t = W_{t+1}-W_t$, $\Delta Z_t = Z_{t+1}-Z_t$. 

$\widetilde{S}$ would be the price if there were no impact from our executions. It follows an arithmetic Brownian motion (ABM) to which an information component $Y$ has been added. The information process $Y$ is an AR(1) process. It could be for example the return of the S\&P500 index, or some information specific to the security being traded. $\gamma$ represents the relevance of that information, that is how much it impacts the price.

There are two dynamics that we will consider for the real price $S$, depending on whether the market impact is assumed to be permanent or temporary. We will explain what those terms mean when defining the price dynamics below. We assume that the market impact is linear in both settings, which means that the market reacts proportionally to the amount executed.

In the case of permanent market impact the mid-price dynamics are changed by each execution. This  means that when we compute the trade cost, the unaffected price $\widetilde{S}$ is replaced, during the execution, by the impacted or affected price $S$: 
\begin{equation}\label{permDyn}
S_t = S_{t-1} + \theta \Delta V_{t-1} + \gamma Y_t + \sigma S_0 \Delta W_{t-1},\quad S_0 = \widetilde{S}_0,
\end{equation}
where the permanent impact parameter $\theta$ is a positive constant.

In the case of temporary market impact each execution only changes the price for the current time period. The mid-price $\widetilde{S}$ is still given by \eqref{unaffPrice}, and the effective price $S$ is derived from $\widetilde{S}$ each period. $S$ has the following dynamics: \\
\begin{equation}\label{tempDyn}
S_t = \widetilde{S}_{t} + \eta \Delta V_{t-1},\quad S_0 = \widetilde{S}_0,
\end{equation}
where the temporary impact parameter $\eta$ is a positive constant.

\begin{remark}
Since one case assumes that the impact lasts for the whole trade, and the other assumes that the impact is instantaneous and affects only an order at the time it is done, both are limit cases of a more general impact pattern that is more progressive, see for example Obizhaeva and Wang \cite{obi2006}.
\end{remark}

We will keep the two more stylized impact cases and analyze them separately. The problem in both cases is to minimize the expected cost of execution. Since we are considering a buy order, $X_t$ is the number of units left to buy. Hence the optimal expected execution cost at time $0$ is \\
\begin{equation}\label{DisCost}
C^*(X_0,S_0) := \min_{\{\Delta V\}} C(X_0,S_0,\{\Delta V\}) = \min_{\{\Delta V\}} \mathbb{E}_0 \left[ \sum_{t=0}^{T-1} S_{t+1} \Delta V_t\right],
\end{equation}
subject to $X_0 = X$, $X_T = 0$.

\begin{remark}
As we mentioned earlier, we do not enforce any constraint on the sign of $\Delta V$, which means that we are allowed to sell in our buy order. 
\end{remark}

%

We now present some calculations deriving the optimal solution of problem \eqref{DisCost} in the cases of permanent and temporary impact. Our calculations in the general setting follow essentially Bertsimas and Lo \cite{bertsimas1998} but with a slightly different notation, as done initially in Bonart, Brigo and Di Graziano \cite{bonart2014} and Kulak \cite{kulak2015}. We further derive the optimal solution in the static class, using a more straightforward method.

\subsection{Permanent market impact: optimal adapted solution}

In this section, we solve problem \eqref{DisCost} reproducing the solution of Bertsimas and Lo \cite{bertsimas1998}, assuming that the market impact is permanent, which means that the affected price follows \eqref{permDyn}. In the adapted setting, the problem is solved recursively. At any time $t$, we consider the problem as if $t$ was the initial time, and the execution was optimal from time $t+1$. We only have to make a decision for the period $t$, ignoring the past and having already solved the future.

For any $t$, the execution cost from time $t$ onward is the sum of  the cost at time $t$ and the cost from time $t+1$ onward.
Taking the minimum of the expectation, this can be written as the Bellman equation: \\
\begin{equation}\label{bell}
C^*(X_t,S_t)=\min_{\Delta V} \mathbb{E}_t \left[S_{t+1}\Delta V_t + C^*(X_{t+1},S_{t+1}) \right].
\end{equation} \\
Since the execution should be finished by time T ($X_T=0$), all the remaining shares must be executed during the last period : \\
\begin{equation*}
\Delta V_{T-1}^* = X_{T-1}.
\end{equation*}
Substituting this value into the Bellman equation \eqref{bell} taken at $t=T-1$ gives us the optimal expected cost at time $T-1$: \\
\begin{equation*}
\begin{split}
C^*(X_{T-1}, S_{T-1}) &= \min_{\Delta V} \mathbb{E}_{T-1}[S_T \Delta V_{T-1}]\\
&= \mathbb{E}_{T-1}[S_T X_{T-1}]\\ 
&= \mathbb{E}_{T-1}[(S_{T-1} + \theta X_{T-1} + \gamma Y_T) X_{T-1}]\\
&= S_{T-1} X_{T-1} + \theta X_{T-1}^2 + \rho \gamma X_{T-1} Y_{T-1},
\end{split}
\end{equation*}
where we used the fact that $Y_{T-1}$, $X_{T-1}$ and $S_{T-1}$ are known at time $T-1$, as well as the null expectation of standard Brownian motion increments.

We now move one step backward to obtain the optimal strategy at time $T-2$, plugging the expression above in \eqref{bell} taken at $t=T-2$ and noting that $X_{T-1} = X_{T-2} - \Delta V_{T-2}$. \\
\begin{equation*}
\begin{split}
C^*(X_{T-2}, S_{T-2}) &= \min_{\Delta V} \mathbb{E}_{T-2}[S_{T-1} \Delta V_{T-2} + C^*(X_{T-1}, S_{T-1})]\\
&= \min_{\Delta V} \mathbb{E}_{T-2}[(S_{T-2} + \theta \Delta V_{T-2} + \gamma \rho Y_{T-2}) \Delta V_{T-2} + S_{T-1} X_{T-1} + \theta X_{T-1}^2 + \rho \gamma X_{T-1} Y_{T-1} ]\\
&= \min_{\Delta V} [(S_{T-2} + \theta \Delta V_{T-2} + \gamma \rho Y_{T-2}) \Delta V_{T-2}\\
&+ (S_{T-2} + \theta \Delta V_{T-2} + \gamma \rho Y_{T-2} + \gamma \rho^2 Y_{T-2}) (X_{T-2} - \Delta V_{T-2}) + \theta (X_{T-2} - \Delta V_{T-2})^2]\\
&= \min_{\Delta V} [S_{T-2} X_{T-2} + \gamma \rho Y_{T-2}  X_{T-2} (1+\rho) -(\gamma \rho^2  Y_{T-2} +\theta X_{T-2})\Delta V_{T-2} + \theta X_{T-2}^2 + \theta \Delta V_{T-2}^2 ].\\
\end{split}
\end{equation*}
In order to find the minimum of this expression, we set to zero its derivative with respect to $\Delta V_{T-2}$: \\
\begin{equation*}
\frac{\partial C(X_{T-2}, S_{T-2})}{\partial \Delta V_{T-2}} = - \theta X_{T-2} - \gamma \rho^2 Y_{T-2} +2 \theta \Delta V_{T-2} = 0.
\end{equation*}
The solution of this equation is the optimal amount to execute at time $T-2$: \\
\begin{equation*}
\Delta V_{T-2}^* = \frac{X_{T-2}}{2} + \frac{\gamma \rho^2 Y_{T-2}}{2\theta}.
\end{equation*}
The optimal expected cost at time $T-2$ is \\
\begin{equation*}
\begin{split}
C^*(X_{T-2}, S_{T-2}) &= S_{T-2} X_{T-2} + \gamma \rho Y_{T-2}  X_{T-2} (1+\rho) -(\gamma \rho^2  Y_{T-2} +\theta X_{T-2})\Delta V_{T-2}^* + \theta X_{T-2}^2 + \theta (\Delta V_{T-2}^*)^2 \\
&= S_{T-2} X_{T-2} + \gamma \rho Y_{T-2}  X_{T-2} (1+\rho) -(\gamma \rho^2  Y_{T-2} +\theta X_{T-2})\left(\frac{X_{T-2}}{2} + \frac{\gamma \rho^2 Y_{T-2}}{2\theta}\right)\\ 
&+ \theta X_{T-2}^2 + \theta \left(\frac{X_{T-2}}{2} + \frac{\gamma \rho^2 Y_{T-2}}{2\theta}\right)^2 \\
&= S_{T-2} X_{T-2} + \frac{3 \theta}{4}X_{T-2}^2 + \gamma \rho (1 + \frac{\rho}{2}) X_{T-2} Y_{T-2} - \frac{\gamma^2 \rho^4}{4 \theta}Y_{T-2}^2.\\
\end{split}
\end{equation*}
We resume the recursion using the expression above, and obtain the optimal strategy at time $T-3$. \\
\begin{equation*}
\begin{split}
C^*(X_{T-3}, S_{T-3}) &= \min_{\Delta V} \mathbb{E}_{T-3}[S_{T-2} \Delta V_{T-3} + C^*(X_{T-2}, S_{T-2})]\\
&= \min_{\Delta V} \mathbb{E}_{T-3}[S_{T-2} \Delta V_{T-3} + S_{T-2} X_{T-2} + \frac{3 \theta}{4}X_{T-2}^2 + \gamma \rho (1 + \frac{\rho}{2}) X_{T-2} Y_{T-2} - \frac{\gamma^2 \rho^4}{4 \theta}Y_{T-2}^2]\\
&= \min_{\Delta V} [(S_{T-3} + \theta \Delta V_{T-3} + \gamma \rho Y_{T-3}) X_{T-3} + \frac{3 \theta}{4}(\Delta V_{T-3}-X_{T-3})^2 \\
&+ \gamma \rho^2 (1 + \frac{\rho}{2}) (\Delta V_{T-3}-X_{T-3}) Y_{T-3} - \frac{\gamma^2 \rho^4}{4 \theta}(\rho^2Y_{T-3}^2+\sigma_Y^2)]\\
&= \min_{\Delta V} [\frac{3\theta}{4} \Delta V_{T-3}^2 - \left(\frac{\theta X_{T-3}}{2}+\gamma \rho^2(1+\frac{\rho}{2}) Y_{T-3}\right)\Delta V_{T-3}+S_{T-3}X_{T-3}+\frac{3\theta}{4} X_{T-3}^2 \\ &+\gamma\rho Y_{T-3}X_{T-3}(1+\rho+\frac{\rho^2}{2})-\frac{\gamma^2\rho^4}{4\theta}(\rho^2Y_{T-3}^2+\sigma_Y^2)].\\
\end{split}
\end{equation*}
In order to find the minimum of this expression, we set to zero its derivative with respect to $\Delta V_{T-3}$: \\
\begin{equation*}
\frac{\partial C(X_{T-3}, S_{T-3})}{\partial \Delta V_{T-3}} = \frac{3\theta}{2} \Delta V_{T-3} - \frac{\theta}{2} X_{T-3} - \gamma \rho^2 (1 + \frac{\rho}{2}) Y_{T-3} = 0.
\end{equation*}
The solution of this equation is the optimal amount to execute at time $T-3$: \\
\begin{equation*}
\begin{split}
\Delta V_{T-3}^* &= \frac{X_{T-3}}{3} + \frac{\gamma \rho^2 (\rho + 2)}{3 \theta} Y_{T-3}.
\end{split}
\end{equation*}
We then compute the optimal expected cost at time $T-3$: \\
\begin{equation*}
\begin{split}
C^*(X_{T-3}, S_{T-3}) &= \frac{3\theta}{4} (\Delta V_{T-3}^*)^2 - \left(\frac{\theta X_{T-3}}{2}+\gamma \rho^2(1+\frac{\rho}{2}) Y_{T-3}\right)\Delta V_{T-3}^*+S_{T-3}X_{T-3}+\frac{3\theta}{4} X_{T-3}^2 \\ &+\gamma\rho Y_{T-3}X_{T-3}(1+\rho+\frac{\rho^2}{2})-\frac{\gamma^2\rho^4}{4\theta}(\rho^2Y_{T-3}^2+\sigma_Y^2) \\
&= \frac{3\theta}{4} \left(\frac{X_{T-3}}{3} + \frac{\gamma \rho^2 (\rho + 2)}{3 \theta} Y_{T-3}\right)^2 +S_{T-3}X_{T-3}+\frac{3\theta}{4} X_{T-3}^2-\frac{\gamma^2\rho^4}{4\theta}(\rho^2Y_{T-3}^2+\sigma_Y^2)\\
&- \left(\frac{\theta X_{T-3}}{2}+\gamma \rho^2(1+\frac{\rho}{2}) Y_{T-3}\right)\left(\frac{X_{T-3}}{3} + \frac{\gamma \rho^2 (\rho + 2)}{3 \theta} Y_{T-3}\right) +\gamma\rho Y_{T-3}X_{T-3}(1+\rho+\frac{\rho^2}{2}) \\
&= S_{T-3} X_{T-3} + \frac{2 \theta}{3} X_{T-3}^2 + \frac{\rho^2 + 2 \rho + 3 }{3} \gamma \rho X_{T-3} Y_{T-3}- \frac{\gamma^2\rho^4}{4\theta}\left((\frac{(\rho + 2)^2}{3} + \rho^2) Y_{T-3}^2 + \sigma_Y^2\right).\\
\end{split}
\end{equation*}
More generally, we can see a pattern emerging from the three previous optimal strategies and expected costs results, which can be proven formally by induction. 
\begin{proposition}[Optimal execution strategy]
For any $i \geq 1$ the optimal execution strategy at time $T-i$ is \\
\begin{equation*}
\begin{split}
&\Delta V_{T-i}^* = \frac{X_{T-i}}{i} + a_i Y_{T-i}\\
&\text{with }a_i = \frac{\gamma \sum_{k=1}^{i-1}{(i-k) \rho^{k+1}}}{i \theta} \quad \text{for i }\geq 2, \text{and }a_1 = 0.
\end{split}
\end{equation*}
\end{proposition}
\begin{remark}
$a_i$ can be simplified to
\begin{equation*}
a_i = \frac{\gamma \rho^2}{i\theta(1-\rho)^2}(\rho^i-i\rho+i-1) \quad \text{for i }\geq 1.
\end{equation*}
\end{remark}
\begin{proof}
Let $i \geq 2$.
\begin{align*}
a_i &= \frac{\gamma}{i\theta}\sum_{k=1}^{i-1}{(i-k) \rho^{k+1}} \\
&= \frac{\gamma}{\theta}\rho^2\sum_{k=0}^{i-2}\rho^{k}-\frac{\gamma}{i\theta}\rho^2\sum_{k=1}^{i-1}k\rho^{k-1} \\
&= \frac{\gamma}{\theta}\rho^2\frac{1-\rho^{i-1}}{1-\rho}-\frac{\gamma}{i\theta}\rho^2\frac{(i-1)\rho^i-i\rho^{i-1}+1}{(1-\rho)^2} \\
&= \frac{\gamma\rho^2}{i\theta(1-\rho)^2}\left(i-i\rho^{i-1}-i\rho+i\rho^i-(i-1)\rho^i+i\rho^{i-1}-1\right)
\end{align*}
and $a_1=0$.
\end{proof}
\begin{proposition}[Optimal expected cost]
For any $i \geq 1$, the minimum expected cost at time $T-i$ is \\
\begin{equation*}
\begin{split}
&C^*_{ad}(X_{T-i}, S_{T-i}) = S_{T-i} X_{T-i} + \frac{i+1}{2i} \theta X_{T-i}^2 + \frac{(i+1) \theta a_{i+1}}{i \rho} X_{T-i} Y_{T-i} - b_i Y_{T-i}^2 - (\sum_{k=2}^{i-1} b_k) \sigma_Y^2\\
&\text{with }b_i = \sum_{k=2}^{i} {\theta \rho^{2(i-k)} \frac{k}{2(k-1)} a_k^2} \text{ for }i \geq 2.
\end{split}
\end{equation*}
\\
\end{proposition}
\begin{remark}\label{bPerm}
$b_i$ can be simplified to
\begin{equation*}
b_i = \frac{\gamma^2\rho^4}{2\theta(1-\rho)^3} \left( \frac{1-\rho^{2i}}{1+\rho} -\frac{(1-\rho^i)^2}{i(1-\rho)}\right) \quad \text{for i }\geq 2.
\end{equation*}
\end{remark}
\begin{proof}
Let $i \geq 2$.
\begin{align*}
b_i &= \sum_{k=2}^{i} {\rho^{2(i-k)} \frac{\gamma^2 \rho^4}{2(k-1)k\theta(1-\rho)^4}(\rho^k-k\rho+k-1)^2} \\
&= \frac{\gamma^2 \rho^{2i+4}}{2\theta(1-\rho)^4} \sum_{k=2}^{i} {\rho^{-2k} \frac{\rho^{2k}-2k\rho^{k+1}+2(k-1)\rho^k+k^2\rho^2-2k(k-1)\rho+k^2-2k+1}{(k-1)k}} \\
&= \frac{\gamma^2 \rho^{2i+4}}{2\theta(1-\rho)^4} \sum_{k=2}^{i} \left[ \frac{1}{(k-1)k}-\frac{2\rho^{1-k}}{k-1}+\frac{2\rho^{-k}}{k}+\frac{k\rho^{2(1-k)}}{k-1}-2\rho^{1-2k}+\frac{(k-1)\rho^{-2k}}{k}\right] \\
&= \frac{\gamma^2 \rho^{2i+4}}{2\theta(1-\rho)^4} \left[ \sum_{k=2}^{i} \left(\frac{1}{k-1}-\frac{1}{k}\right)-\sum_{k=1}^{i-1}\frac{2\rho^{-k}}{k}+\sum_{k=2}^{i}\frac{2\rho^{-k}}{k}+\sum_{k=1}^{i-1}\frac{(k+1)\rho^{-2k}}{k}-\sum_{k=2}^{i}2\rho^{1-2k}+\sum_{k=2}^{i}\frac{(k-1)\rho^{-2k}}{k}\right] \\
&= \frac{\gamma^2 \rho^{2i+4}}{2\theta(1-\rho)^4} \left[ \left(1-\frac{1}{i}\right)+2\left(-\rho^{-1}+\frac{\rho^{-i}}{i}\right)+ (\rho^{-2}+\rho^{-4}-2\rho^{-3})\frac{\rho^{-2(i-1)}-1}{\rho^{-2}-1} + \left(\rho^{-2}-\frac{\rho^{-2i}}{i}\right)\right] \\
&= \frac{\gamma^2}{2\theta(1-\rho)^4} \left[ \left(1-\frac{1}{i}\right)\rho^{2i+4}-2\rho^{2i+3}+2\frac{\rho^{i+4}}{i}+ (1-\rho)^2\frac{\rho^{2}-\rho^{2i}}{\rho^{-2}(1-\rho^2)} + \rho^{2i+2}-\frac{\rho^{4}}{i}\right] \\
&= \frac{\gamma^2}{2\theta(1-\rho)^4} \left[ (\rho^{2}-2\rho+ 1)\rho^{2i+2}+ (1-\rho)\frac{\rho^{4}-\rho^{2i+2}}{1+\rho} +\frac{-\rho^{2i+4}+2\rho^{i+4}-\rho^{4}}{i}\right] \\
&= \frac{\gamma^2}{2\theta(1-\rho)^4} \left[ \frac{(1-\rho)(\rho^{2i+2}-\rho^{2i+4})+ (1-\rho)(\rho^{4}-\rho^{2i+2})}{1+\rho} -\frac{(1-\rho^i)^2}{i}\rho^4\right].
\end{align*}
\end{proof}
\begin{corollary}
In particular, the optimal expected cost at time $0$ is \\
\begin{equation}
C^*_{ad}(X_{0}, S_{0}) = S_{0} X + \frac{T+1}{2T} \theta X^2 + \frac{(T+1) \theta a_{T+1}}{T \rho} X Y_{0} - b_T Y_{0}^2 - (\sum_{k=2}^{T-1} b_k) \sigma_Y^2.
\end{equation}
\\
\end{corollary}
\begin{remark}
Although this strategy is adapted, it does not take into account the price, but only the information process. This makes sense because if there was no information, the optimal strategy would be deterministic as shown in \cite{bertsimas1998}.
\end{remark}

\subsection{Permanent market impact: optimal deterministic solution}

We will now constrain the solutions of \eqref{DisCost} to be deterministic, so that the strategy is known at time $0$ and can be executed with no further calculations, independently of the path taken by the price. 

\begin{theorem}[Optimal deterministic execution strategy]
When we restrict the solutions to the subset of deterministic strategies, the optimal strategy is 
\begin{equation}
X_t^* = \frac{T-t}{T}X + \frac{\gamma Y_0\rho^2}{\theta(1-\rho)^2}\left[\rho^t - 1 + (1-\rho^T)\frac{t}{T}\right].
\end{equation}
\end{theorem}

\begin{proof}
To solve \eqref{DisCost}, we will simply assume that every $X_t$ is known at time $0$ and compute the expected cost at time $0$: 
\begin{align*}
C(X_0,S_0,\{\Delta V\}) &= \mathbb{E}_0 \left[ \sum_{t=0}^{T-1} S_{t+1} \Delta V_t\right] \\
&=  \sum_{t=0}^{T-1} \Delta V_t \mathbb{E}_0[S_{t+1}] \quad \text{since $\Delta V_t = X_t - X_{t+1}$ is deterministic}\\
&= \sum_{t=0}^{T-1} \Delta V_{t}\left( \mathbb{E}_0[S_{t}] + \theta \Delta V_{t} + \gamma \mathbb{E}_0[Y_{t+1}] \right)\\
&= \sum_{t=0}^{T-1} \Delta V_{t}\left( S_0 + \theta \sum_{i=0}^{t} \Delta V_i +  \gamma Y_0 \sum_{i=1}^{t+1} \rho^i \right) \quad \text{by induction}\\
&= S_0 X_0 + \sum_{t=0}^{T-1} (X_t - X_{t+1})\left( \theta \sum_{i=0}^{t} (X_i - X_{i+1}) + \gamma Y_0 \rho \frac{1 - \rho^{t+1}}{1-\rho} \right) \\
&= S_0 X_0 + \theta \sum_{t=0}^{T-1} (X_t - X_{t+1}) (X_0 - X_{t+1}) + \gamma Y_0 \rho \sum_{t=0}^{T-1}\frac{1 - \rho^{t+1}}{1-\rho}(X_t - X_{t+1}).
\end{align*}
Problem \eqref{DisCost} can be rewritten as
\begin{align*}
C^*(X_0,S_0)&=\min_{x} C(x).
\end{align*}
To find the minimum, we set to zero the partial derivatives of the expected cost with respect to $X_1$, ..., $X_{T-1}$. For $t = 1,...,T-1$ it gives us 
\begin{equation*}
\frac{\partial C}{\partial X_t} = \theta(X_0 - X_{t+1}) - \theta(X_0 - X_{t}) - \theta(X_{t-1} - X_{t}) + \gamma \rho Y_0 \left(\frac{1 - \rho^{t+1}}{1-\rho} - \frac{1 - \rho^{t}}{1-\rho}\right) = 0.
\end{equation*}
We obtain the difference equation 
\begin{equation}\label{diffPerm}
X_{t+1} - 2X_t + X_{t-1} = \frac{\gamma Y_0 }{\theta} \rho^{t+1},
\end{equation}
with boundary conditions $X_0 = X$ and $X_T = 0$.

The solution of \eqref{diffPerm} is of the form $A + Bt + C\rho^t$ for some constants $A$, $B$ and $C$.
Plugging this expression back in the equation yields 
\begin{eqnarray*}
&&A+B(t+1)+C\rho^{t+1}-2(A + Bt + C\rho^t)+A + B(t-1) + C\rho^{t-1}= \frac{\gamma Y_0}{\theta} \rho^{t+1}\\
&&C\rho^t(\rho-2+\rho^{-1}) = \frac{\gamma Y_0}{\theta} \rho^{t+1}\\
&&C = \frac{\gamma Y_0\rho^2}{\theta (1-\rho)^2}. 
\end{eqnarray*}
From the boundary conditions we have 
\[
X_0 = A + C = X,  \ \ 
A = X - \frac{\gamma Y_0\rho^2}{\theta (1-\rho)^2},\]
and 
\[
X_T = A + BT + C\rho^T = 0, \ \
B = -\frac{X}{T} + \frac{\gamma Y_0\rho^2(1-\rho^T)}{\theta (1-\rho)^2 T}. \]
Combining those, we obtain the closed-form formula of the optimal deterministic strategy.
\end{proof}

\begin{remark}
If $Y_0=0$ (no initial information), $\rho=0$ (information is just noise) or $\gamma=0$ (information is irrelevant), the strategy consists in splitting the execution in orders of equal amounts over the period $T$. This is a particular case of the strategy more generally known as VWAP (volume-weighted average price), and is the strategy obtained when there is no information.
\end{remark}

\begin{theorem}[Optimal expected cost associated with the deterministic strategy]
The expected cost at time $0$ associated with the optimal deterministic strategy is 
\begin{equation}
C^*_{det}(X_0,S_0) = S_0 X + \frac{T+1}{2T}\theta X^2 + \frac{\gamma Y_0 \rho X}{T(1-\rho)}\left(T-\rho\frac{1-\rho^{T}}{1-\rho}\right) + \frac{\gamma^2 Y_0^2\rho^4}{2\theta (1-\rho)^3}\left(\frac{(1-\rho^T)^2}{T (1-\rho)}- \frac{1-\rho^{2T}}{1+\rho} \right).
\end{equation}
\end{theorem}
\begin{proof}
For lighter calculations, we set \begin{equation*}
C = \frac{\gamma Y_0\rho^2}{\theta (1-\rho)^2}.
\end{equation*}
The optimal expected cost at time $0$ is 
\begin{equation}\label{calcCostperm}
C^*(X_0,S_0) = S_0 X_0 + \theta \sum_{t=0}^{T-1} (X^*_t - X^*_{t+1}) (X_0 - X^*_{t+1}) + \gamma Y_0 \rho \sum_{t=0}^{T-1}\frac{1 - \rho^{t+1}}{1-\rho}(X^*_t - X^*_{t+1}).
\end{equation}
We compute the two sums in (\ref{calcCostperm}) separately for clarity:
\begin{equation*}
C^*(X_0,S_0) = S_0 X_0 + \theta S_1 + \gamma Y_0 \rho S_2.
\end{equation*}
The second sum is 
\begin{align*}
S_2
&= \sum_{t=0}^{T-1}\frac{1 - \rho^{t+1}}{1-\rho}\left(\frac{X}{T}+C(\rho^t(1-\rho)+\frac{\rho^T - 1}{T})\right) \\
&= \frac{X}{T(1-\rho)}\left(T-\rho\frac{1-\rho^{T}}{1-\rho}\right)+\frac{C}{1-\rho}\left(\frac{1-\rho^T}{1-\rho}(1-\rho)+\rho^T-1-\rho\frac{1-\rho^{2T}}{1-\rho^2}(1-\rho)+\rho\frac{(1-\rho^T)^2}{T(1-\rho)}\right) \\
&=\frac{X}{T(1-\rho)}\left(T-\rho\frac{1-\rho^{T}}{1-\rho}\right)-\frac{C\rho(1-\rho^{2T})}{1-\rho^2}+\frac{C\rho(1-\rho^T)^2}{T(1-\rho)^2}.
\end{align*}
The first sum is
\begin{align*}
S_1
&=\sum_{t=0}^{T-1} \left(\frac{X}{T}+C(\rho^t(1-\rho)+\frac{\rho^T - 1}{T})\right) \left(\frac{t+1}{T}X-C(\rho^{t+1}-1+(1-\rho^T)\frac{t+1}{T})\right) \\ 
&=\sum_{t=0}^{T-1} \frac{t+1}{T^2}X^2+\sum_{t=0}^{T-1}\frac{CX}{T}\left(1-\rho^{t+1}+(\rho^T-1)\frac{t+1}{T}\right)
+\sum_{t=0}^{T-1}\frac{t+1}{T}CX\left(\rho^t(1-\rho)+\frac{\rho^T - 1}{T}\right) \\
&+\sum_{t=0}^{T-1}C^2\left(\rho^t(1-\rho)+\frac{\rho^T - 1}{T}\right)\left(1-\rho^{t+1}+(\rho^T-1)\frac{t+1}{T}\right)\\
&= \frac{T(T+1)}{2T^2}X^2+\sum_{t=0}^{T-1}\frac{CX}{T}\left(-(t+2)\rho^{t+1}+(t+1)\rho^t+2(t+1) \frac{\rho^T-1}{T}+1\right) \\
&+C^2\sum_{t=0}^{T-1} \left(\rho^{t}(1-\rho-\rho^{t+1}+\rho^{t+2})+\left((t+1)\rho^{t}-(t+2)\rho^{t+1}+1\right) \frac{\rho^{T}-1}{T}+(t+1)\frac{(\rho^{T}-1)^2}{T^2}\right) \\
&=\frac{CX}{T}\left(\frac{-(T-1)\rho^{T+2}+T\rho^{T+1}-\rho^2}{(1-\rho)^2}-2\frac{\rho-\rho^{T+1}}{1-\rho}+\frac{T\rho^{T+1}-(T+1)\rho^{T}+1}{(1-\rho)^2}+(T+1)\rho^T-1\right) \\
&+C^2\sum_{t=0}^{T-1}\left(\rho^{2t+1}(\rho-1)+\frac{(1-\rho)(\rho^T-1)}{T}t\rho^{t}+\left(1-\rho+\frac{\rho^T-1}{T}(-2\rho+1)\right)\rho^t+\frac{(\rho^T-1)^2}{T^2}(t+1)\right) \\&+C^2\sum_{t=0}^{T-1}\frac{\rho^T-1}{T}+\frac{T+1}{2T}X^2\\
&= \frac{CX}{T}\frac{(-T-1)\rho^{T+2}+2(T+1)\rho^{T+1}+\rho^2-(T+1)\rho^{T}+1-2\rho+((T+1)\rho^T-1)(1+\rho^2-2\rho)}{(1-\rho)^2}\\
& +C^2\left(\frac{\rho^{2T}-1}{1+\rho}\rho+\frac{\rho^{T}-1}{T}\rho\frac{(T-1)\rho^T-T\rho^{T-1}+1}{1-\rho}+\frac{T-T\rho+(1-2\rho)(\rho^T-1)}{T}\frac{1-\rho^{T}}{1-\rho}\right) \\
&+C^2\left(\frac{T+1}{2T}(\rho^T-1)^2+\rho^T-1\right) + \frac{T+1}{2T}X^2\\
&=
\frac{C^2}{T(1-\rho^2)}\left(T\rho^{2T+1}-T\rho-T\rho^{2T+2}+T\rho^2 + (\rho^{T+1}-\rho+\rho^{T+2}-\rho^{2})((T-1)\rho^{T}-T\rho^{T-1}+1)\right)\\
&+\frac{C^2}{T(1-\rho^2)}\left((T-T\rho+\rho^{T}-1-2\rho^{T+1}+2\rho)(1-\rho^{T}+\rho-\rho^{T+1})\right) \\
&+\frac{C^2}{T(1-\rho^2)}\left(\frac{T+1}{2}(1+\rho^{2T}-2\rho^{T}-\rho^{2}-\rho^{2T+2}+2\rho^{T+2}) + T\rho^{T}-T-T\rho^{T+2}+T\rho^{2}\right) + \frac{T+1}{2T}X^2\\ 
&=\frac{C^2}{T(1-\rho^2)}\left( \left(\frac{1-T}{2}\rho^{2} +T\rho +\frac{-T-1}{2}\right)\rho^{2T} -\rho^{T+2} +\rho^{T} +\frac{T+1}{2}\rho^{2} -T\rho + \frac{T-1}{2}\right) + \frac{T+1}{2T}X^2 \\
&=\frac{C^2(1-\rho)(1-\rho^{2T})}{2(1+\rho)} -\frac{C^2(1-\rho^T)^2}{2T} + \frac{T+1}{2T}X^2.
\end{align*}
Substituting those results in (\ref{calcCostperm}), we obtain 
\begin{align*}
C^*(X_0,S_0) &= S_0 X + \theta \left(\frac{C^2(1-\rho)(1-\rho^{2T})}{2(1+\rho)} -\frac{C^2(1-\rho^T)^2}{2T} + \frac{T+1}{2T}X^2\right) \\
&+ \gamma Y_0 \rho \left(\frac{X}{T(1-\rho)}\left(T-\rho\frac{1-\rho^{T}}{1-\rho}\right)-\frac{C\rho(1-\rho^{2T})}{1-\rho^2}+\frac{C\rho(1-\rho^T)^2}{T(1-\rho)^2}\right) \\
&= S_0 X + \frac{\gamma^2 Y_0^2\rho^4(1-\rho^{2T})}{2(1+\rho)\theta (1-\rho)^3} -\frac{\gamma^2 Y_0^2\rho^4(1-\rho^T)^2}{2T\theta (1-\rho)^4} + \frac{T+1}{2T}\theta X^2 \\
&+  \frac{\gamma Y_0 \rho X}{T(1-\rho)}\left(T-\rho\frac{1-\rho^{T}}{1-\rho}\right)-\frac{\gamma^2 Y_0^2\rho^4(1-\rho^{2T})}{(1-\rho^2)\theta (1-\rho)^2}+\frac{\gamma^2 Y_0^2\rho^4(1-\rho^T)^2}{T(1-\rho)^4\theta}.
\end{align*}
\end{proof}

\subsection{Permanent market impact: adapted vs deterministic solution}

We will now quantify the difference between the two strategies obtained above. First, we define the difference.
\begin{definition}[Absolute difference]
The absolute difference between the deterministic and the adapted optimal expected cost at time $0$ is 
\begin{equation*}
\epsilon_{abs} := C_{det}^*(X_0,S_0)-C_{ad}^*(X_0,S_0).
\end{equation*}
\end{definition}
\begin{proposition}[Value of the absolute difference]
The value of the absolute difference is 
\begin{equation}
\epsilon_{abs} = (\sum_{k=2}^{T-1} b_k) \sigma_Y^2.
\end{equation}
\end{proposition}
\begin{proof}
By definition, 
\begin{equation*}
\epsilon_{abs} = \frac{\gamma^2 Y_0^2\rho^4}{2\theta (1-\rho)^3}\left(\frac{(1-\rho^T)^2}{T (1-\rho)}- \frac{1-\rho^{2T}}{1+\rho} \right) + b_T Y_{0}^2 + (\sum_{k=2}^{T-1} b_k) \sigma_Y^2.
\end{equation*}
Substituting the value of $b_T$ obtained in Remark \ref{bPerm} in this expression yields the result.
\end{proof}
\begin{corollary}
The two strategies have the same expected cost when the information process is not random ($\sigma_Y=0$).
\end{corollary}
\begin{corollary}
As expected, the adapted strategy is always better than the deterministic one.
\end{corollary}
\begin{proof}
For any $k \geq 2$, $b_k$ is a sum of products of non-negative terms by definition, so it is non-negative. Hence their sum is non-negative. Multiplying this sum by the non-negative term $\sigma_Y^2$, we conclude that $\epsilon_{abs}$ is non-negative.
\end{proof}
\begin{definition}[Relative difference]
The relative difference between the deterministic and the adapted optimal expected cost at time $0$ is \\
\begin{equation*}
\epsilon_{rel} :=\frac{\epsilon_{abs}}{C_{det}^*(X_0,S_0)}.
\end{equation*}
\end{definition}

We now quantify the difference between the deterministic and the adapted strategies through a few numerical examples. The amount of shares to execute $X$ is set at $10^6$, big enough to have an impact on the market. The initial price of the stock is $S_0=\$100$, making it intuitive to take the percentage volatility. The number of periods is $T=14$ so that there is around one execution every 30 minutes over a trading day for example. The market impact $\theta=10^{-5}$ is chosen to increase the expected price by a total of 10\% over the execution, as done in Bertsimas and Lo \cite{bertsimas1998}: \\
\begin{equation*}
X(S_0+\theta X)=1.1S_0X.
\end{equation*}
The percentage standard deviation of the price over a time period $\sigma=0.51\%$ is chosen such that the annual volatility is around $30\%$, or equivalently the daily volatility is around $1.89\%$: \\
\begin{equation*}
\sigma \sqrt{14}=1.89\%.
\end{equation*}
The information process is positively auto-correlated $\rho=0.5$. Its importance $\gamma=1$ is chosen arbitrarily. Its volatility $\sigma_Y=0.44$ is chosen such that the standard deviation of the information component is of the same order as that of the stock price: \\
\begin{equation*}
\sqrt{\mathbb{E}[(\gamma Y_t)^2]}\simeq\frac{\gamma \sigma_Y}{\sqrt{1-\rho^2}}=0.51 \quad \text{for t large enough}.
\end{equation*}
By default we assume that there is no initial information $Y_0=0$.

The values described above are summarized in Table \ref{benchPerm}.
\begin{table}[!h]
\centering
    \begin{tabular}{| l | l |}
    \hline
   $X$ & $10^6$ \\ \hline
   $S_0$ & $100$ \\ \hline
   $T$ & $14$ \\ \hline
   $\theta$& $10^{-5}$   \\ \hline
   $\sigma$ & $0.51\%$  \\ \hline
   $\rho$ & $0.5$   \\ \hline
   $\gamma$ & $1$   \\ \hline
   $\sigma_Y$ & $0.44$  \\ \hline
   $Y_0$ & $0$ \\ \hline
    \end{tabular}
\caption{Benchmark parameter values}
\label{benchPerm}
\end{table}

\begin{remark}
In order to obtain an order of magnitude for the expected cost, note that the best we can do is the cost of an instantaneous execution, which is the cost without market impact, and this would be 
\begin{equation*}
S_0X = 10^8.
\end{equation*}
\end{remark}

To get an idea of the influence of the initial information on both strategies, we give a few examples of unaffected and affected price paths obtained with different values of $Y_0$, and their associated strategies in Figures \ref{permSim}, \ref{permSimY5} and \ref{permSimY_5}.
\begin{figure}[h]
\centering
\includegraphics[scale=0.9]{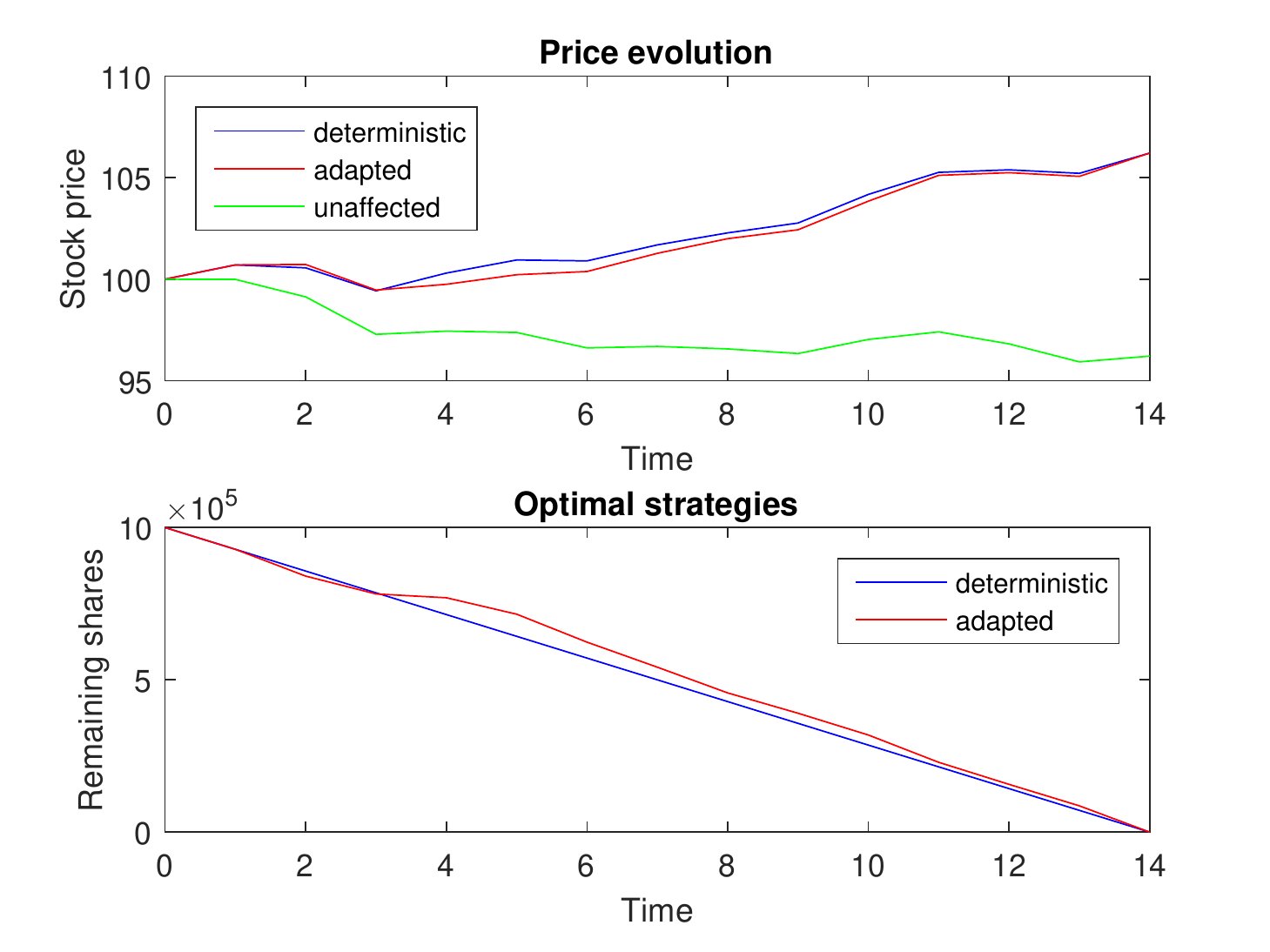}
\caption{One path of a simulated strategy with benchmark parameters ($Y_0=0$)}
\label{permSim}
\end{figure}

\begin{figure}[h]
\centering
\includegraphics[scale=0.9]{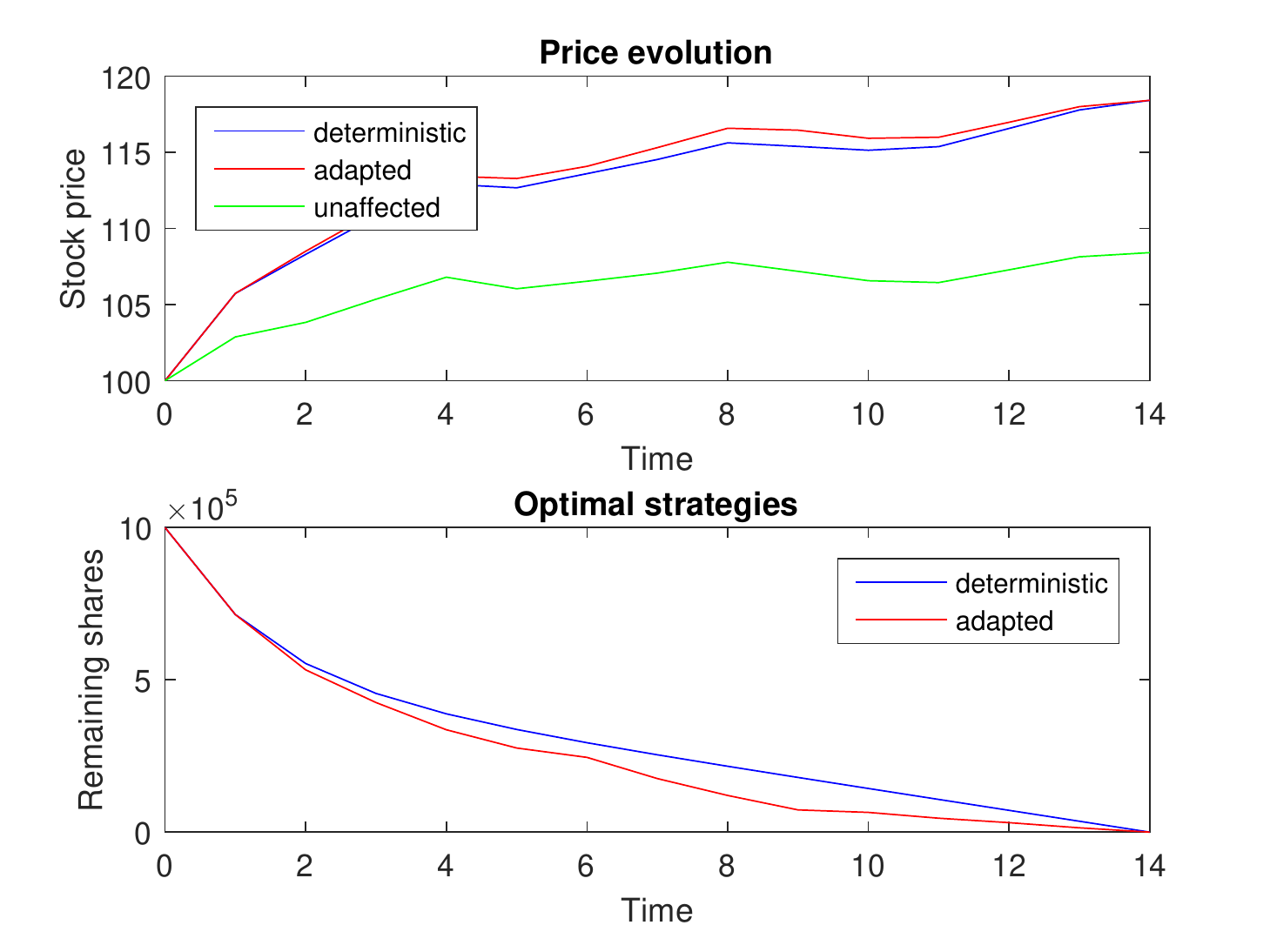}
\caption{One path of a simulated strategy with positive initial information ($Y_0=5$)}
\label{permSimY5}
\end{figure}

\begin{figure}[h]
\centering
\includegraphics[scale=0.9]{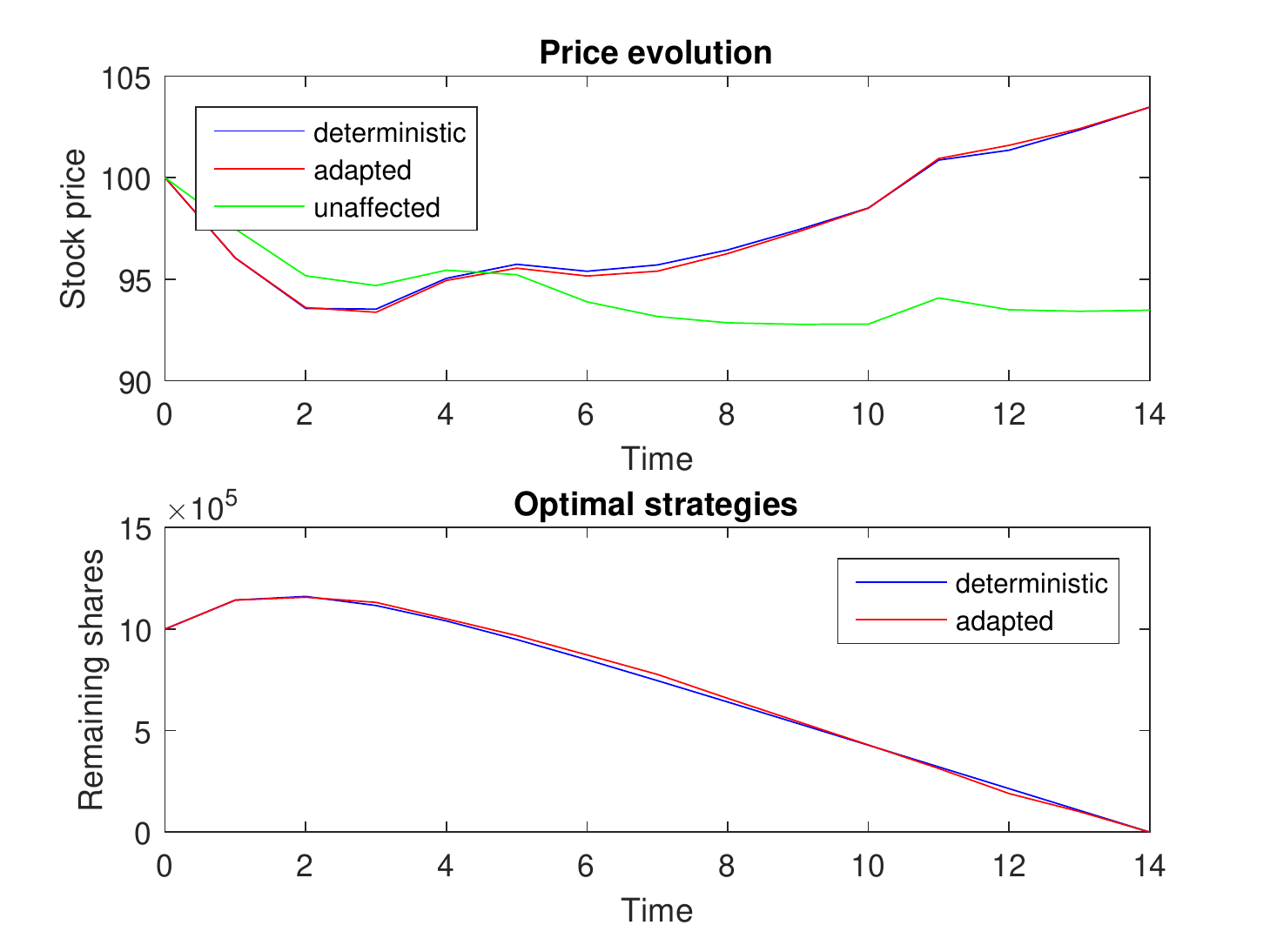}
\caption{One path of a simulated strategy with negative information $Y_0=-5$.}
\label{permSimY_5}
\end{figure}

The upper plot in Figure \ref{permSim} represents the evolution of the price throughout the execution. 
As we can see, the affected price $S$ would be higher than the unaffected price $\tilde{S}$ with both strategies since the market is reacting against a buy order. The lower plot in Figure~\ref{permSim} represents the amount of shares $X_t$ left to be executed throughout the execution. The red curve is the optimal fully adapted strategy. The blue curve is the optimal static or deterministic strategy. Since $Y_0=0$, the deterministic strategy is simply a straight line going from the initial value $X$ at time $0$ to the final value $0$ at time $T$: the execution is done evenly over the time horizon and this is the well known VWAP strategy. The adapted strategy is roughly the same, but it is less smooth since the strategy changes according to the path taken by the price during the execution.

With the benchmark parameters, we find that $C_{det}^*(X_0,S_0)=1.053 \times 10^8$, $C_{ad}^*(X_0,S_0)=1.0534 \times 10^8$ and $\epsilon_{rel}=1.97 \times 10^{-4}$. In particular, the costs obtained with the path shown in Figure \ref{permSim} are $C_{det}(X_0,S_0)=1.0256 \times 10^8$ and $C_{ad}(X_0,S_0)=1.0257 \times 10^8$ so the deterministic strategy would have been better than the adapted one in retrospect.
\begin{remark}
The first step is always the same for both strategies since it relies purely on information known at time $0$. 
\end{remark}

Since the information process is cumulative and positively auto-correlated, a positive initial information suggests that the information term will be increasing throughout the trade. To minimize the impact of the information, the trade is shifted towards the beginning of the time horizon: we increase the rate at which we buy in a first part. 

With $Y_0=5$, we find that the optimal costs for the static and adapted cases are, respectively, $C_{det}^*(X_0,S_0)=1.0967 \times 10^8$, $C_{ad}^*(X_0,S_0)=1.0965 \times 10^8$ and $\epsilon_{rel}=1.89 \times 10^{-4}$. In particular, the costs obtained in the single path shown in Figure \ref{permSimY5} are $C_{det}(X_0,S_0)=1.1086 \times 10^8$ and $C_{ad}(X_0,S_0)=1.1082 \times 10^8$.

On the other hand, a negative initial information suggests that the information term will be more and more negative throughout the term, so its impact on the price will be to reduce it more and more. Hence we want to begin buying as late as we can, even selling shares in a first part to maximize the benefits from the price decrease. Indeed, with $Y_0=-5$, we find that $C_{det}^*(X_0,S_0)=1.0039 \times 10^8$, $C_{ad}^*(X_0,S_0)=1.0037 \times 10^8$ and $\epsilon_{rel}=2.06 \times 10^{-4}$. In particular, the costs obtained in the single path shown in Figure \ref{permSimY_5} are $C_{det}(X_0,S_0)=9.8777 \times 10^7$ and $C_{ad}(X_0,S_0)=9.8758 \times 10^7$. Note that since we begin by selling shares, the effective price goes below the unaffected price at first.
\begin{remark}
In some situations it might be natural  to impose a constraint on the sign of $\Delta V$, since one may not wish to sell during a buy order.
\end{remark}

Now that we have in mind the path taken by the price and by the strategies for a few examples, we will study the influence of each parameter separately, analyzing in a few numerical examples the impact of the parameters and inputs
\[  X, T, \theta, \rho, \gamma, \sigma_Y \ .\]
In each numerical example, the parameters will be those of Table \ref{benchPerm} except for the one whose influence we study. This allows us to study one parameter at a time.
\begin{remark}
 Since $\sigma$ does not appear in the formulas in either case, it has no influence on the optimal expected cost.
\end{remark}


We begin by studing the influence of $X$. Figure \ref{permInfX} shows the evolution of the expected costs and the relative difference when $X$ varies from $10^5$ to $10^7$. The absolute difference does not depend on the amount of shares to execute $X$, while the expected cost grows with $X$, so the relative error decreases when $X$ increases. This can be explained by the fact that the market impact parameter $\theta$ has been calibrated for a certain $X$, and its total permanent influence becomes considerable when $X$ is very large. For example, when $X=10^7$ the permanent impact doubles the price over the execution: the affected price at time $T$ is roughly twice the unaffected price.  This is not really representative of the impact of $X$ since $\theta$ should be a function of $X$: the impact we have on the market should not grow linearly with the amount executed, as opposed to our assumption. 
\begin{figure}[H]
	\begin{minipage}[c]{.54\linewidth}
		\includegraphics[scale=0.6]{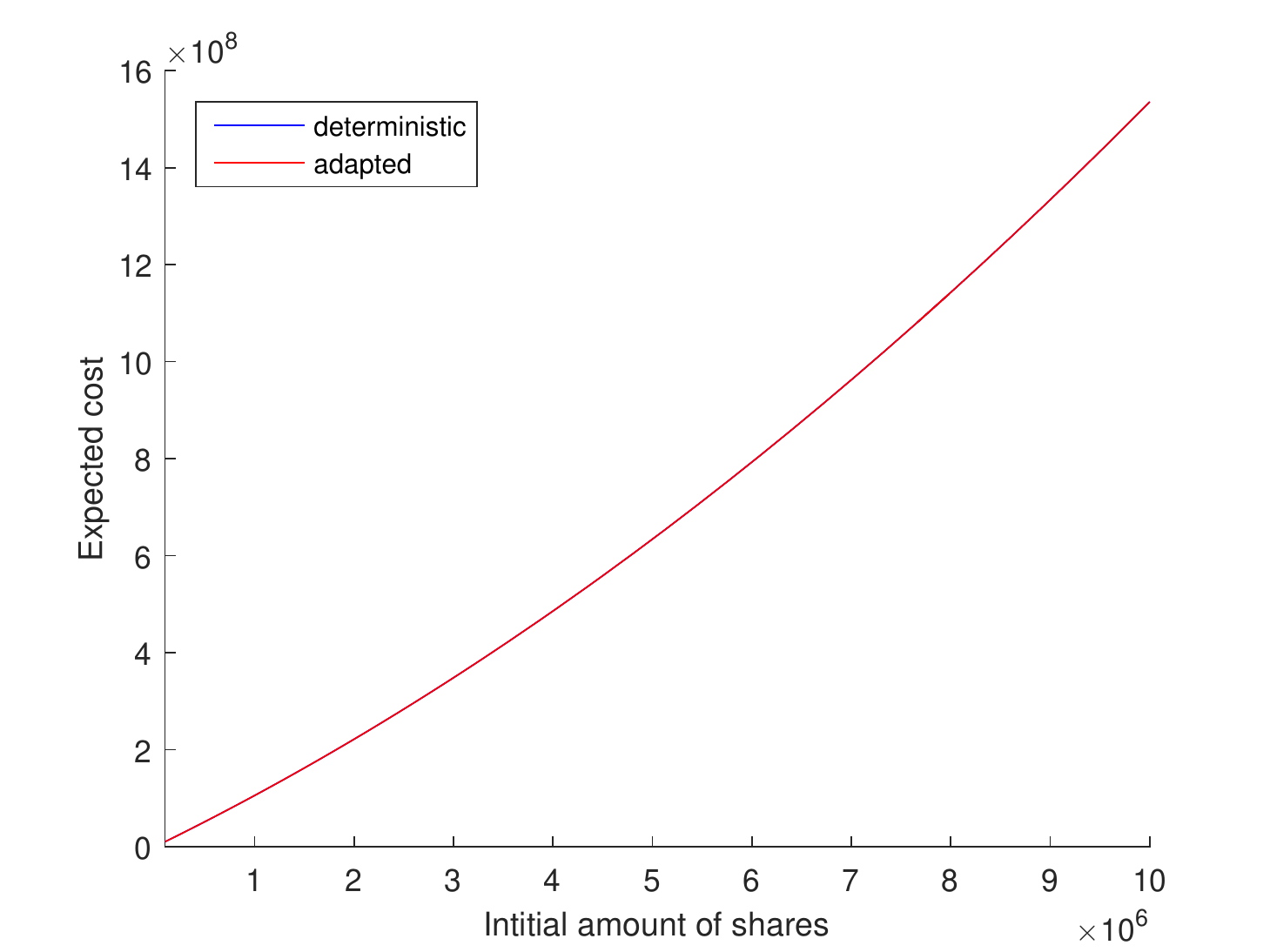}
	\end{minipage}
	\begin{minipage}[c]{.46\linewidth}
		\includegraphics[scale=0.6]{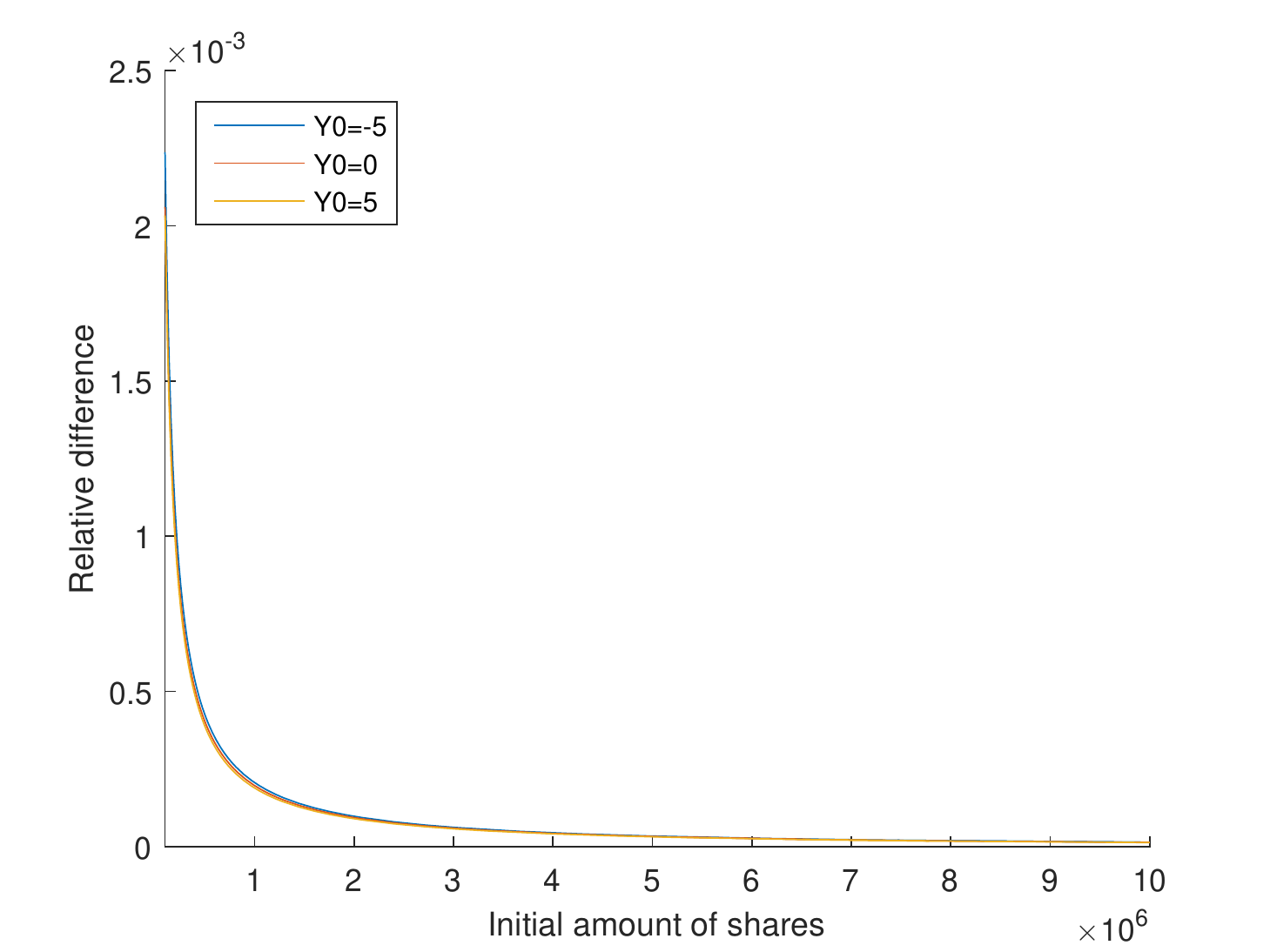}
	\end{minipage}
\caption{Influence of $X$ on the expected costs and relative difference}
\label{permInfX}
\end{figure}

We now consider the influence of $T$. 
\begin{figure}[H]
	\begin{minipage}[c]{.54\linewidth}
		\includegraphics[scale=0.6]{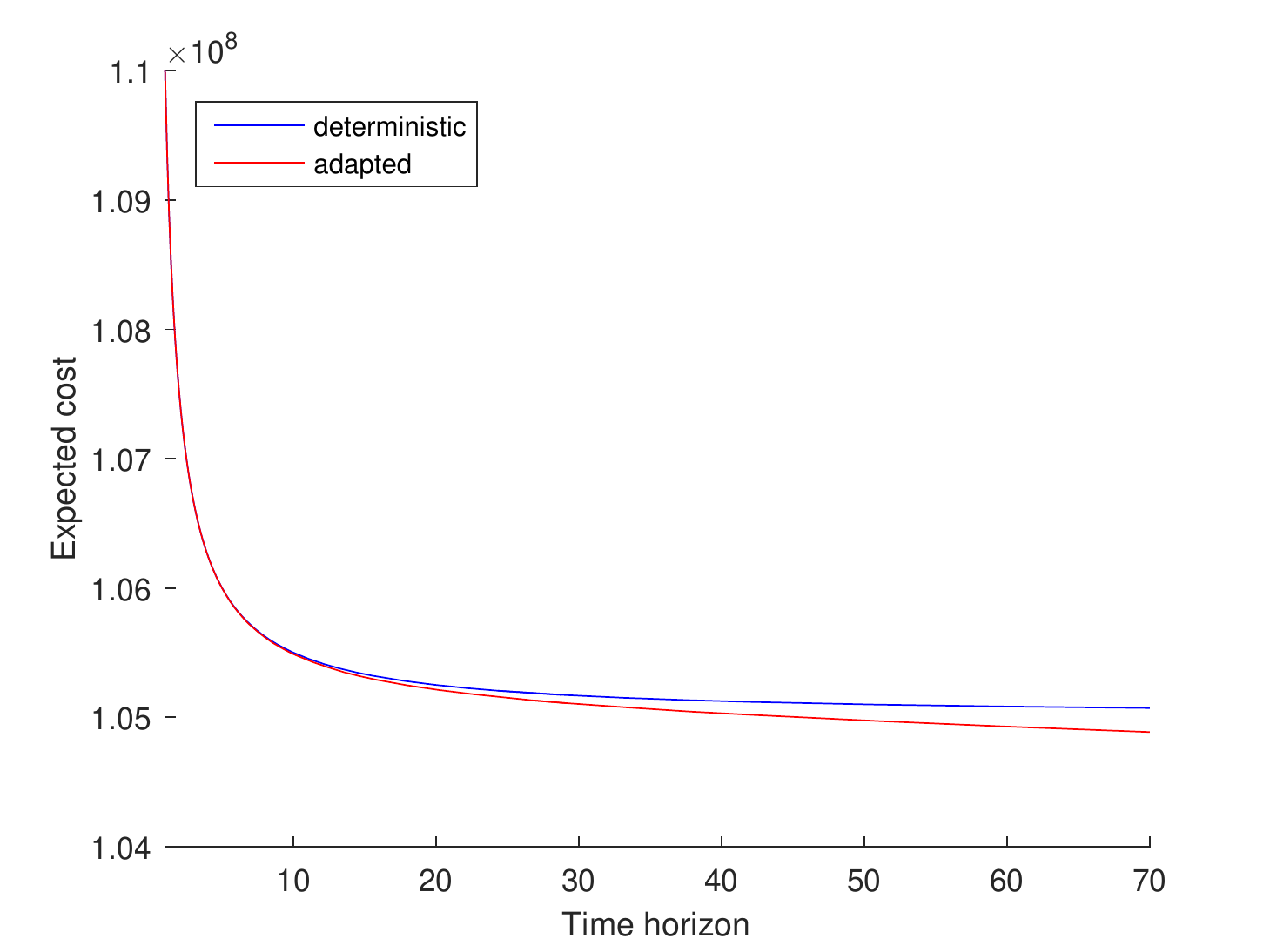}
	\end{minipage}
	\begin{minipage}[c]{.46\linewidth}
		\includegraphics[scale=0.6]{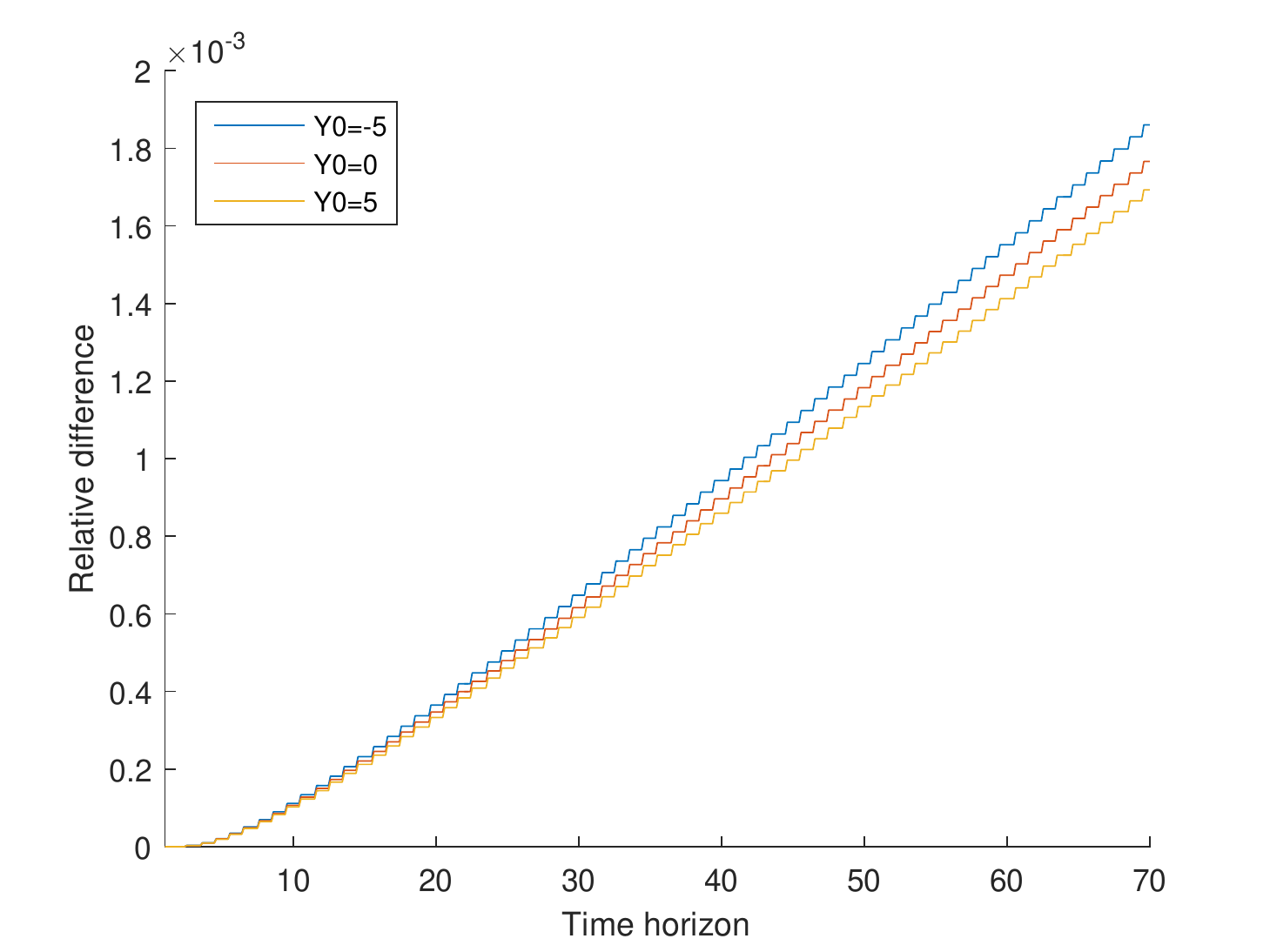}
	\end{minipage}
\caption{Influence of $T$ on the expected costs and relative difference}
\label{permInfT}
\end{figure}
Figure \ref{permInfT} shows the evolution of the expected costs and the relative difference when the time horizon varies from half an hour ($T=1$) to a trading week of $5$ days ($T=70$). The relative difference between the two strategies increases linearly with the time horizon for $T$ large enough. This stems from the fact that the deterministic strategy is set at time $0$, and does not benefit from the information that arrives after, while the adapted strategy will do the best of what is given. Given a full trading week to execute the order, the adapted strategy is almost $0.2\%$ better than the deterministic one.

We now turn to the influence of $\theta$.
\begin{figure}[H]
	\begin{minipage}[c]{.54\linewidth}
		\includegraphics[scale=0.6]{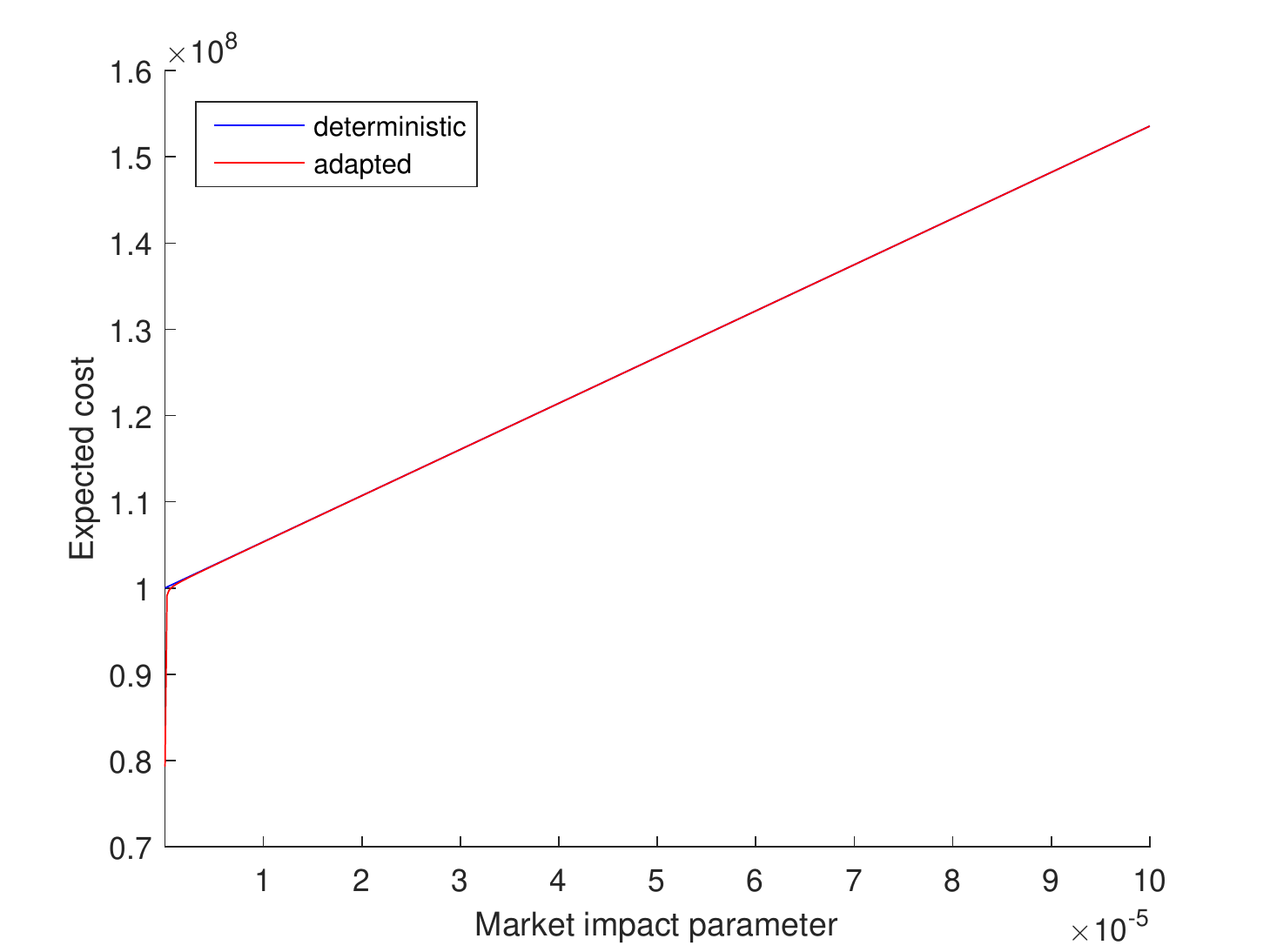}
	\end{minipage}
	\begin{minipage}[c]{.46\linewidth}
		\includegraphics[scale=0.6]{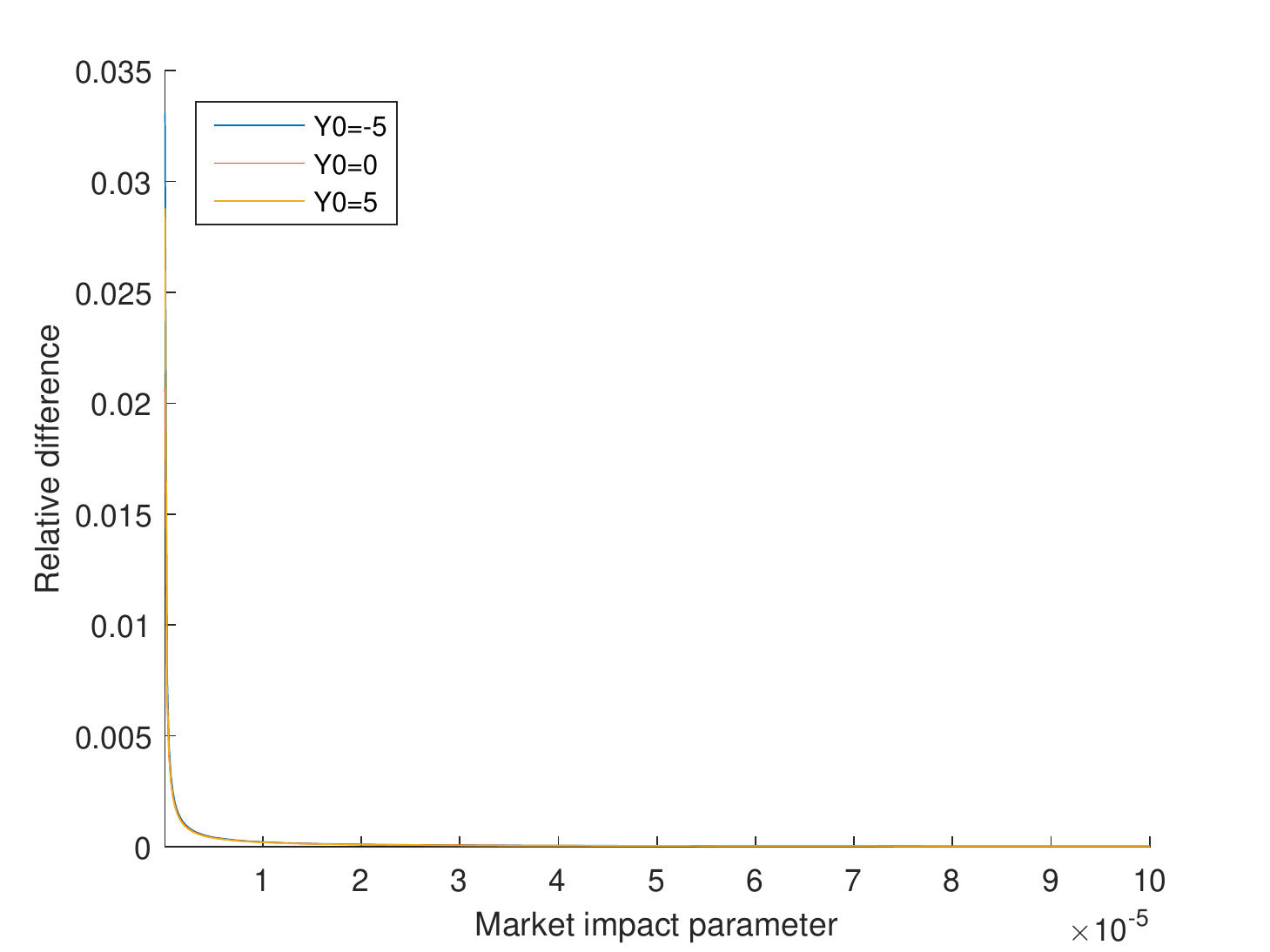}
	\end{minipage}
\caption{Influence of $\theta$ on the expected costs and relative difference}
\label{permInfTH}
\end{figure}
As said in the study of the influence of $X$, when $\theta$ increases, the impact we have on the market increases. More and more of the expected cost is unavoidable so it becomes more and more difficult to reduce the expected cost. Hence the relative difference decreases as $\theta$ increases. Figure \ref{permInfTH} shows the evolution of the expected costs and the relative difference when $\theta$ varies from $10^{-8}$ to $10^{-4}$. For a total increase of $1\%$ of the price over the execution ($\theta=10^{-6}$), the relative difference is $0.21\%$.
\begin{remark}
It would be interesting to study the joint influence of $X$ and $\theta$, as they depend strongly on each other financially. For example, $\theta$ could be taken as a function of $X$ (one could start with a linear function).
\end{remark}

We have an interesting pattern on the optimal expected cost when $\theta \downarrow 0$.
\begin{proposition}
As long as $\sigma_Y \neq 0$, the optimal expected cost tends to $-\infty$ when $\theta$ tends to $0$. When there is initial information ($Y_0 \neq 0$), the expected cost associated with the best deterministic strategy tends to $-\infty$ when $\theta$ tends to $0$.
\end{proposition}
To understand the intuition behind this, we will look at a few examples of strategies used for a small value of $\theta$, and initial information.
As we can see in Figure \ref{permSimTH8Y5}, the strategies are extremely aggressive when the market impact parameter is small, since we accelerate the execution when the price goes against us. There are strategies related to idealized round trips: due to the cumulative effect of information on the trading price, we quickly buy way more than needed, and sell back later, with a higher information-increased price, until we reach our goal. Without market impact, it seems there is no foreseeable punishment for massively leveraging the information benefit.  Note that it is impossible to do this in reality since there is a finite number of shares and this would be prohibited as market manipulation.

\begin{figure}[H]
\centering
\includegraphics[scale=0.8]{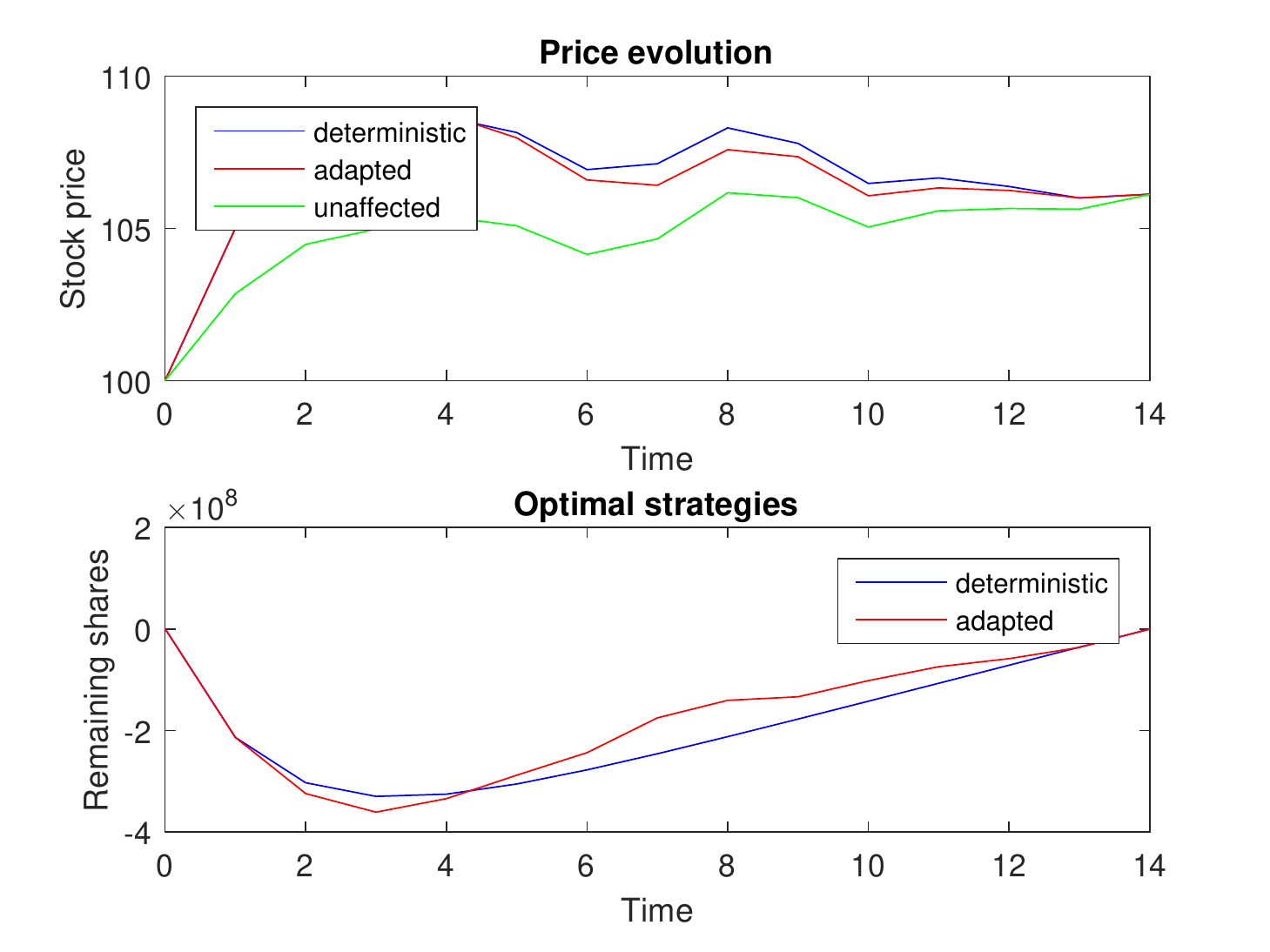}
\caption{One path of a simulated strategy with positive initial information ($Y_0=5$) and small market impact ($\theta=10^{-8}$)}
\label{permSimTH8Y5}
\end{figure}
\begin{figure}[H]
\centering
\includegraphics[scale=0.8]{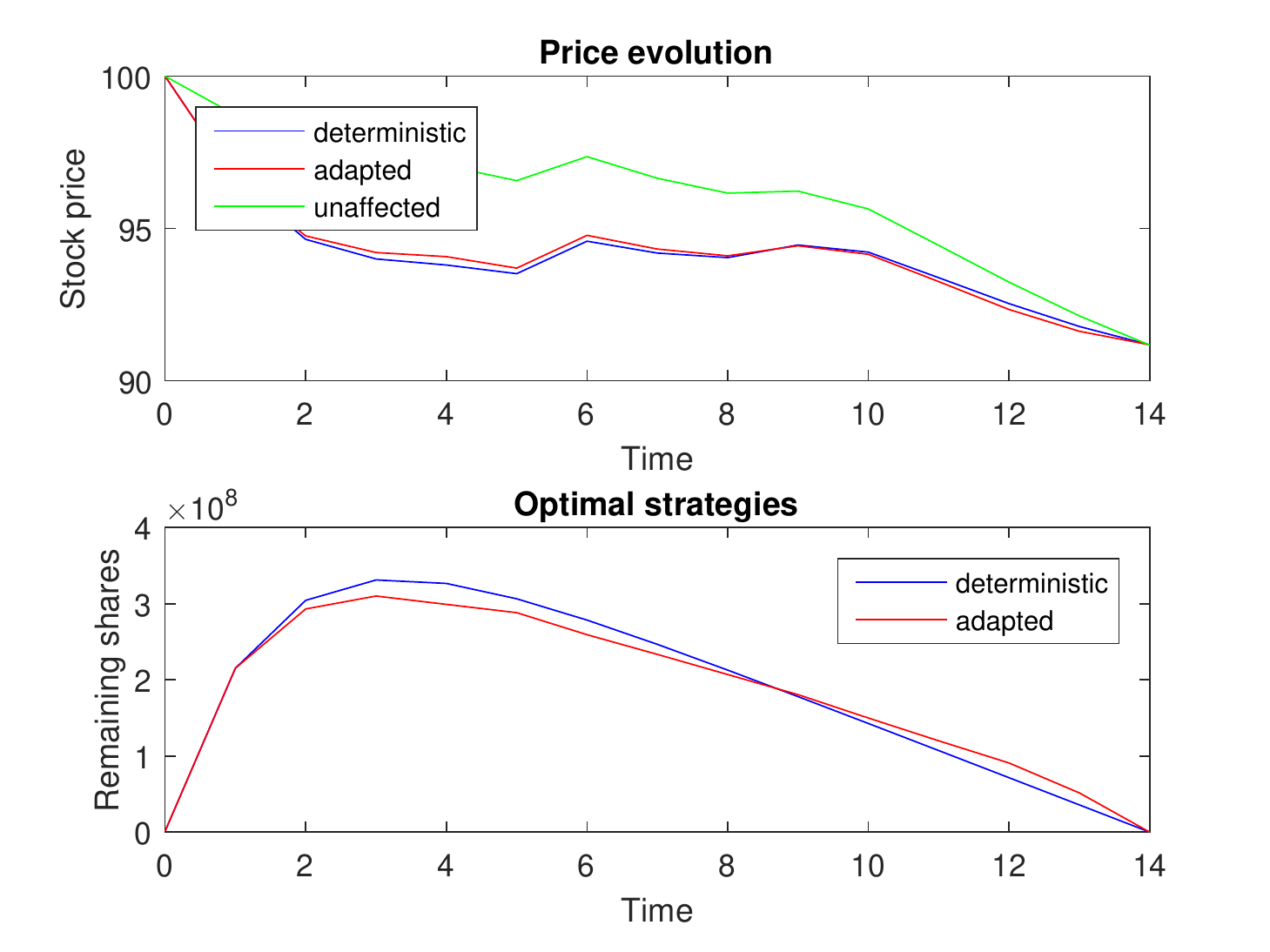}
\caption{One path of a simulated strategy with negative initial information ($Y_0=-5$) and small market impact ($\theta=10^{-8}$)}
\label{permSimTH8Y_5}
\end{figure}
As we can see in Figure \ref{permSimTH8Y_5}, when there is negative initial information the strategies are the opposite of the case of positive initial information, since now information will tend to decrease the price cumulatively in time. We sell a lot of shares initially, since we know that the price will go down later due to information, when we will be able to buy back at a much reduced price.

We consider now the influence of $\rho$.
Figure \ref{permInfR} shows the evolution of the expected costs and the relative difference when $\rho$ varies from $-0.9$ to $0.9$.
Although there is some noticeable difference in the expected costs for large negative auto-correlations ($\rho<-0.8$), the relative difference is particularly relevant when the information process is strongly positively auto-correlated ($\rho>0.8$). It then explodes, up to $8.7\%$ when $Y_0=-5$ and $\rho=0.9$, but such a huge value does not seem realistic for~$\rho$.
\begin{figure}[H]
	\begin{minipage}[c]{.54\linewidth}
		\includegraphics[scale=0.6]{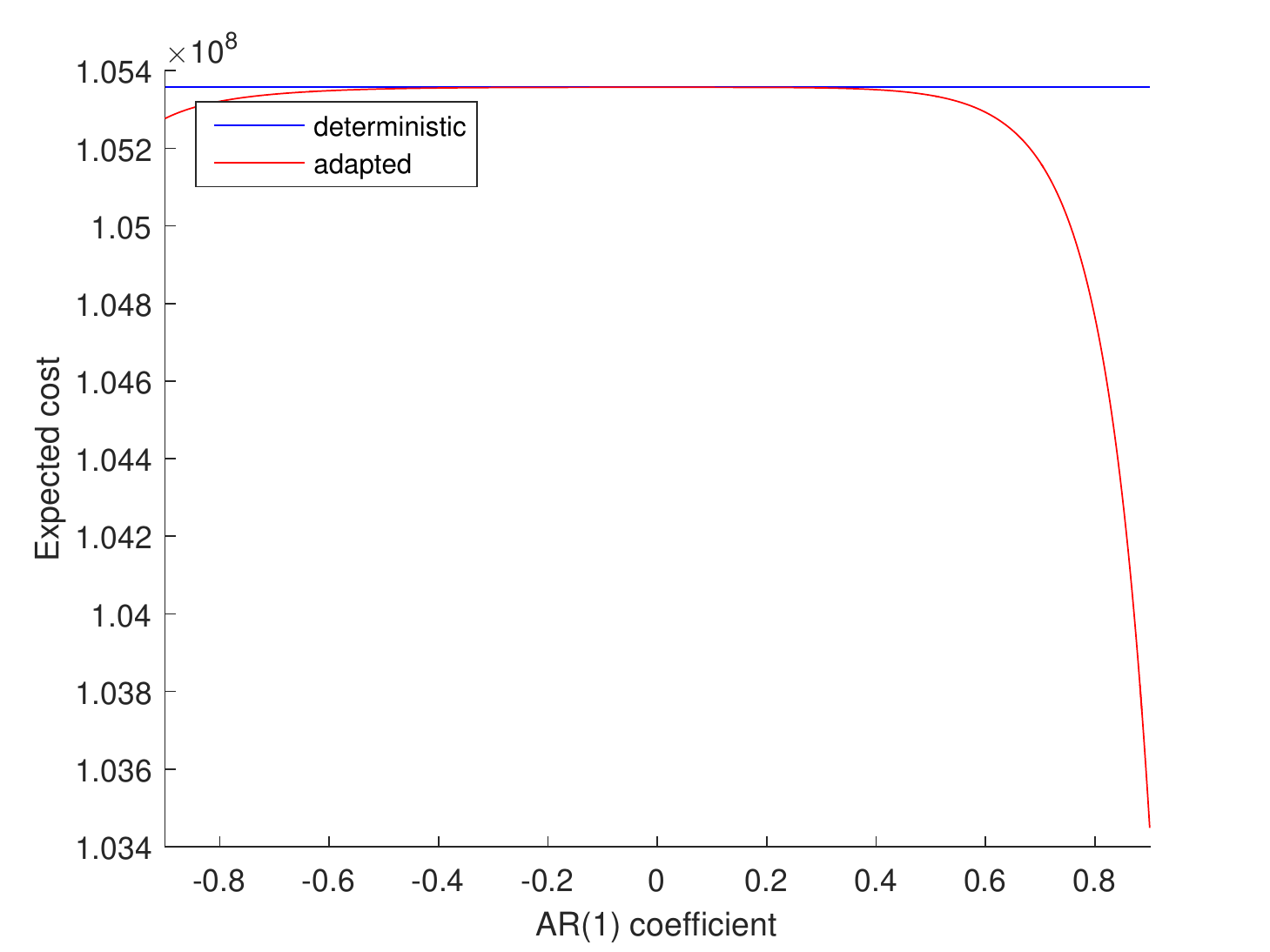}
	\end{minipage}
	\begin{minipage}[c]{.46\linewidth}
		\includegraphics[scale=0.6]{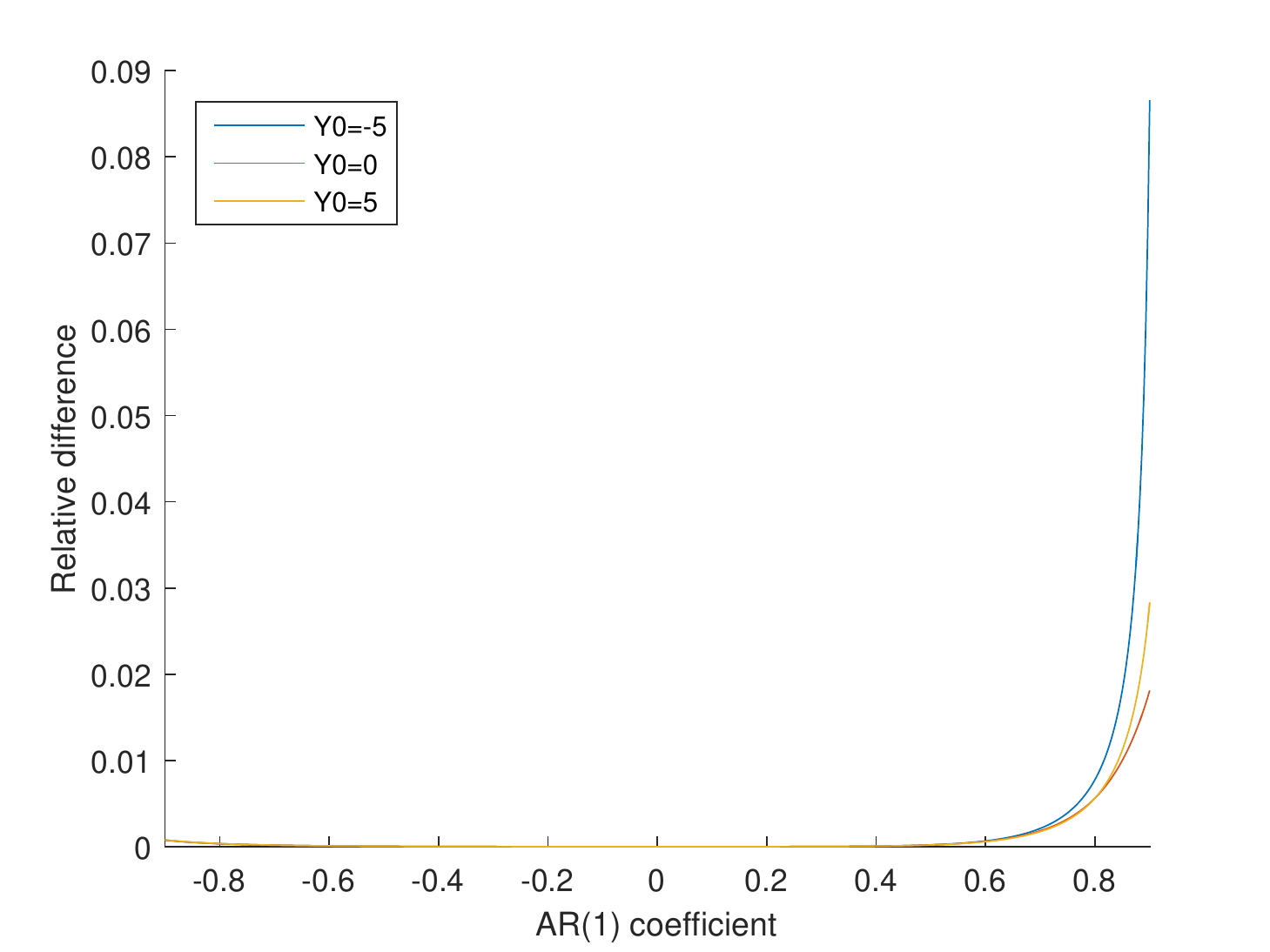}
	\end{minipage}
\caption{Influence of $\rho$ on the expected costs and relative difference}
\label{permInfR}
\end{figure}

As concerns the influence of $\gamma$, we have the following results.
The relative difference grows with $\gamma$, which is intuitive since the more relevant the information is, the more important it is to update our strategy when we receive new information. This  seems especially true when the initial information is negative.
Figure \ref{permInfG} shows the evolution of the expected costs and the relative difference when $\gamma$ varies from $1$ to $10$.

We finally study the influence of $\sigma_Y$.
Figure \ref{permInfS} shows the evolution of the expected costs and the relative difference when $\sigma_Y$ varies from $0$ to $4$.
The volatility of the information process has no influence on the deterministic expected cost, while the adapted expected cost decreases with $\sigma_Y$. Hence the relative difference increases with $\sigma_Y$.

\begin{figure}[H]
	\begin{minipage}[c]{.54\linewidth}
		\includegraphics[scale=0.6]{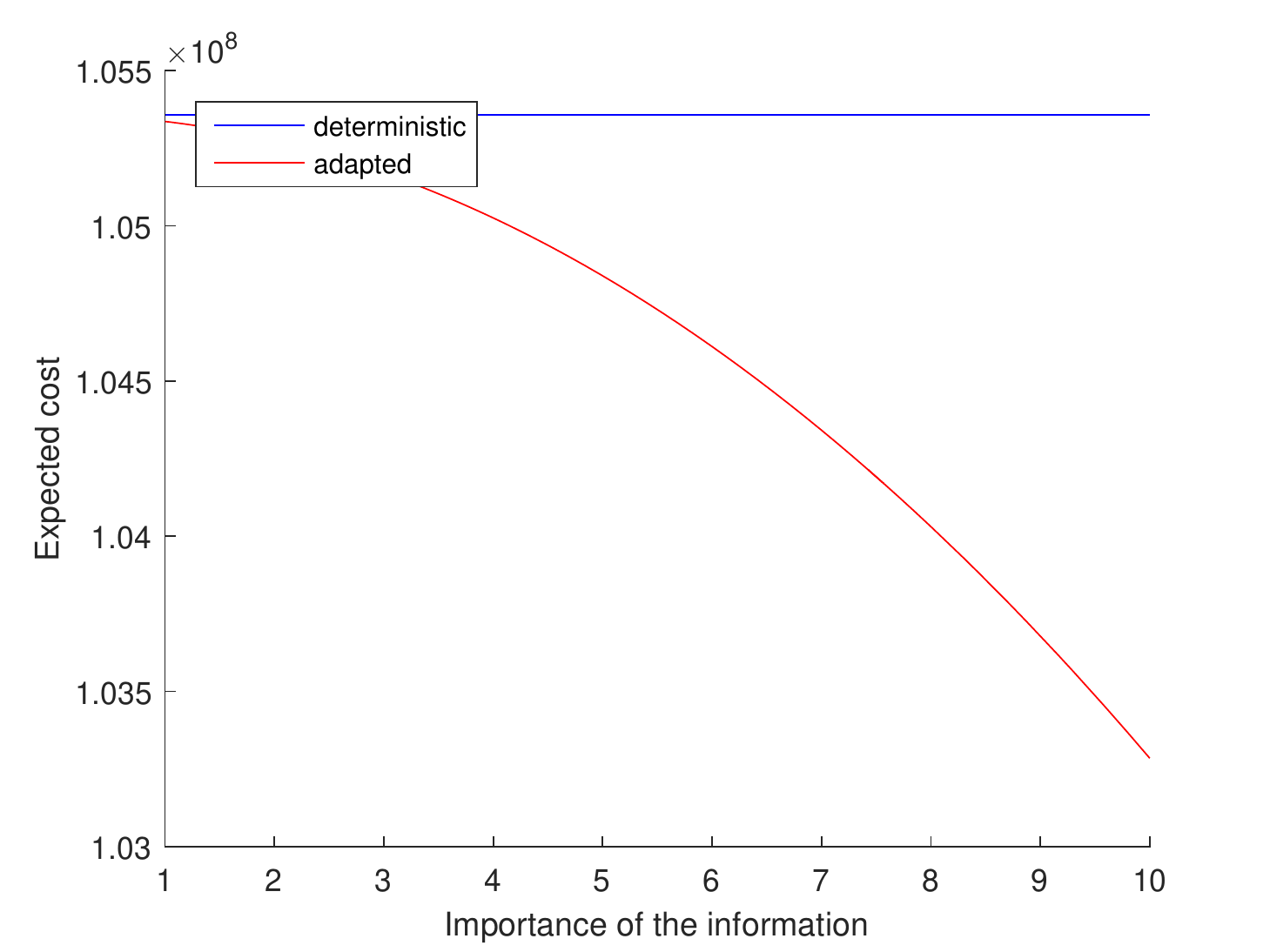}
	\end{minipage}
	\begin{minipage}[c]{.46\linewidth}
		\includegraphics[scale=0.6]{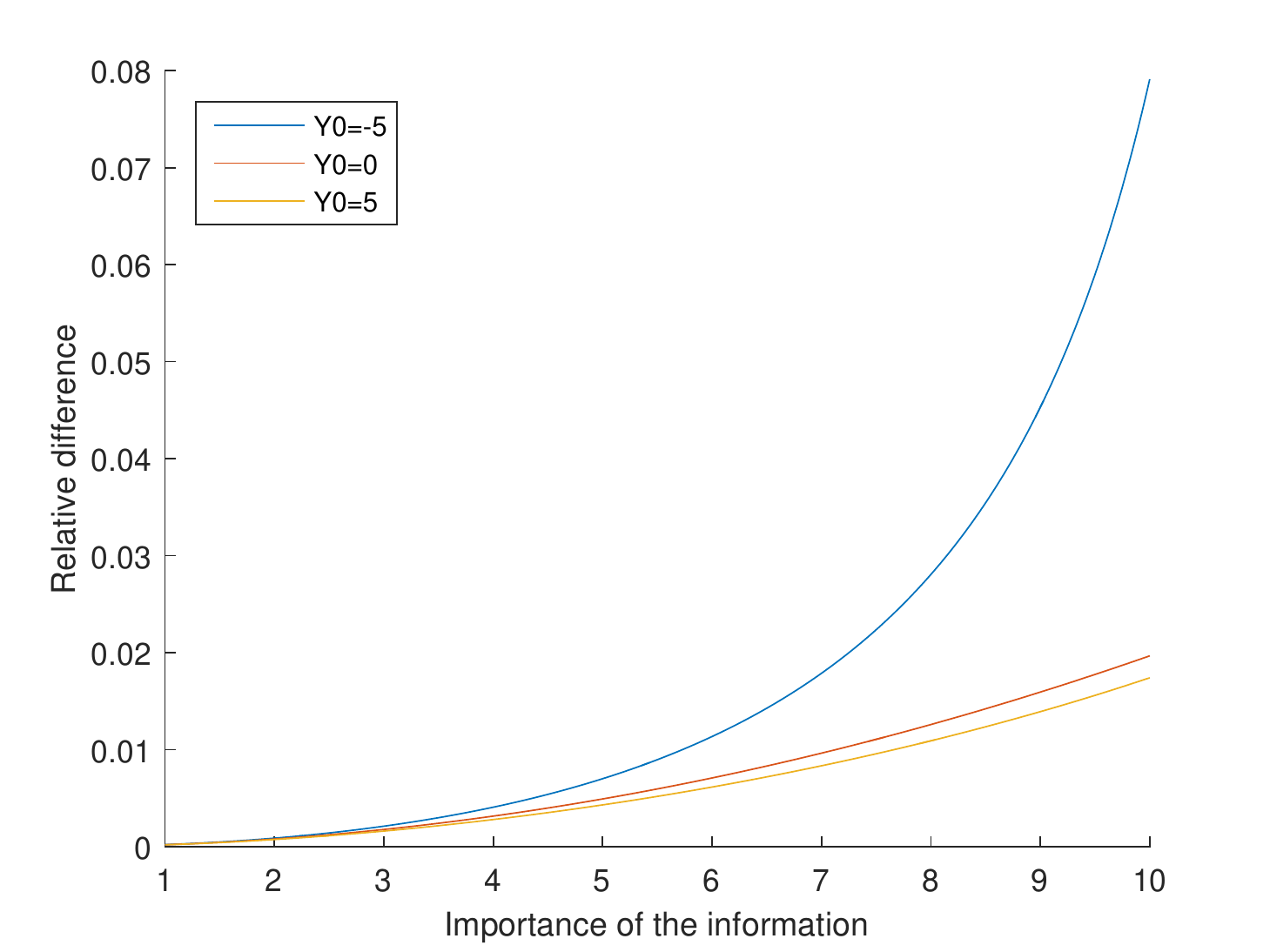}
	\end{minipage}
\caption{Influence of $\gamma$ on the expected costs and relative difference}
\label{permInfG}
\end{figure}
\begin{figure}[h]
	\begin{minipage}[c]{.54\linewidth}
		\includegraphics[scale=0.6]{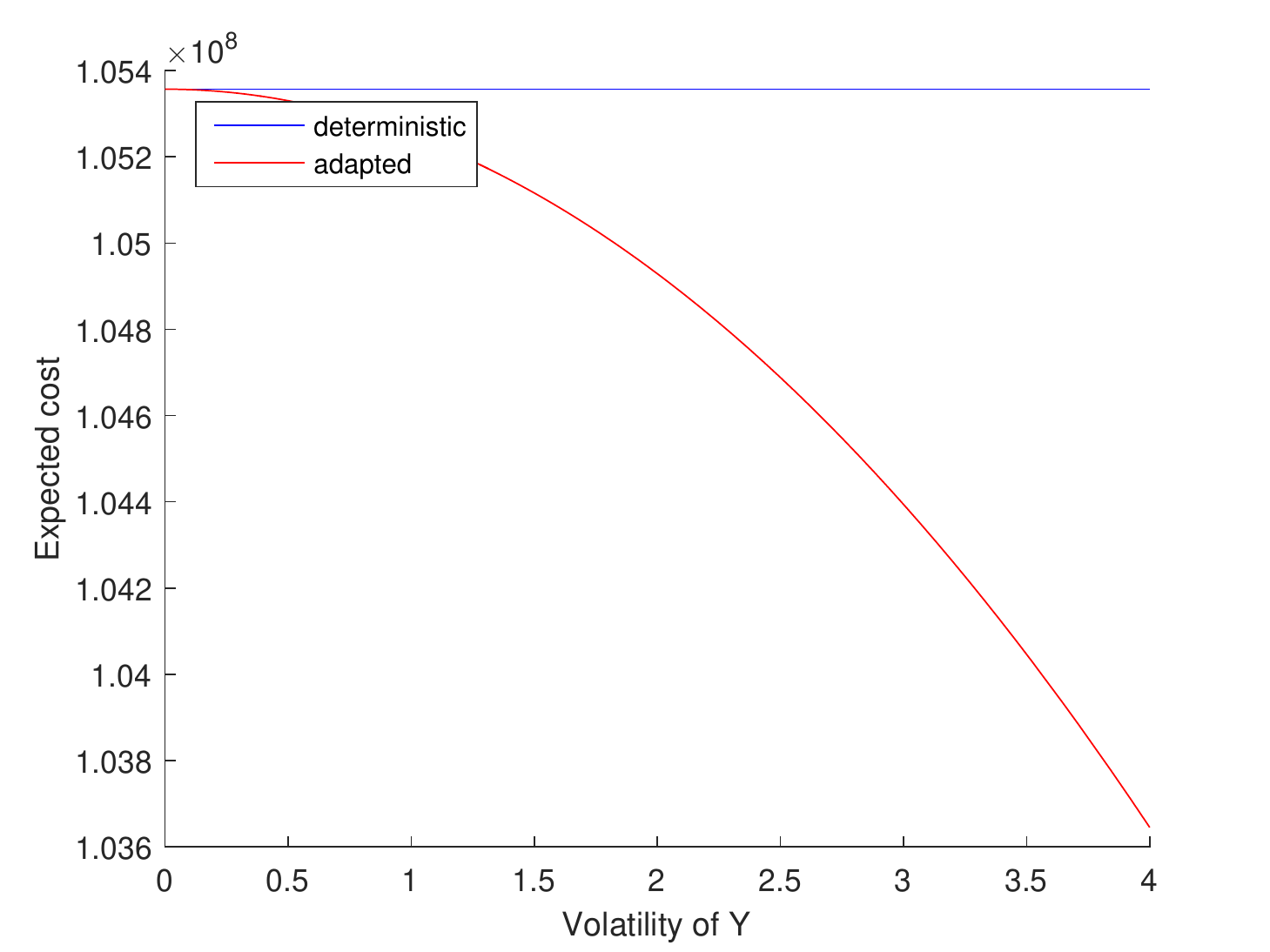}
	\end{minipage}
	\begin{minipage}[c]{.46\linewidth}
		\includegraphics[scale=0.6]{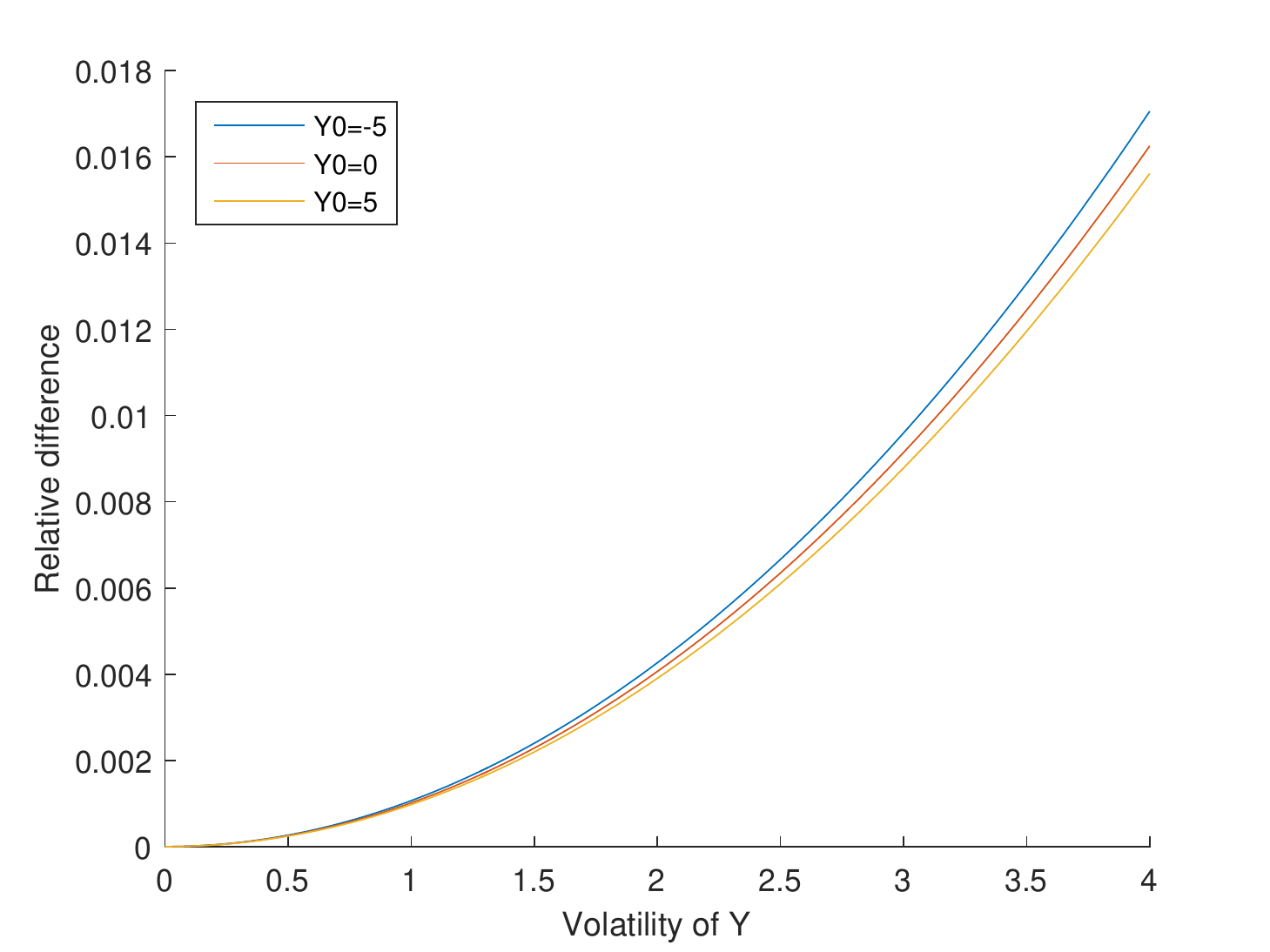}
	\end{minipage}
\caption{Influence of $\sigma_Y$ on the expected costs and relative difference}
\label{permInfS}
\end{figure}

%
%
%
%

\subsection{Temporary market impact: optimal adapted solution}
In this section, we solve problem \eqref{DisCost} reproducing the solution of Bertsimas and Lo \cite{bertsimas1998}, assuming that the market impact is temporary, which means that the affected price follows \eqref{tempDyn}.

In the adapted setting, the problem is again solved by dynamic programming. The steps followed are the same as in the permanent impact case. The Bellman equation (\ref{bell}) holds, and the final condition too: \\
\begin{equation*}
\Delta V_{T-1}^* = X_{T-1}.
\end{equation*}
From this we can calculate the expected cost at time $T-1$ \\
\begin{equation*}
\begin{split}
C^*(X_{T-1}, S_{T-1}) &= \min_{\Delta V} \mathbb{E}_{T-1}[S_T \Delta V_{T-1}]\\
&= \mathbb{E}_{T-1}[S_T X_{T-1}]\\ 
&= \mathbb{E}_{T-1}[(\widetilde{S}_{T-1} + \eta X_{T-1} + \gamma Y_T) X_{T-1}]\\
&= \widetilde{S}_{T-1} X_{T-1} + \gamma \rho X_{T-1} Y_{T-1} + \eta X_{T-1}^2.
\end{split}
\end{equation*}
We then iterate the procedure to the previous time step: \\
\begin{equation*}
\begin{split}
C^*(X_{T-2}, S_{T-2}) &= \min_{\Delta V} \mathbb{E}_{T-2}[S_{T-1} \Delta V_{T-2} + C^*(X_{T-1}, S_{T-1})]\\
&= \min_{\Delta V} [(\widetilde{S}_{T-2} + \gamma \rho Y_{T-2} + \eta \Delta V_{T-2}) \Delta V_{T-2} + (\widetilde{S}_{T-2} + \gamma \rho Y_{T-2}) (X_{T-2} - \Delta V_{T-2})\\
&+ \gamma\rho^2 (X_{T-2} - \Delta V_{T-2}) Y_{T-2} + \eta (X_{T-2} - \Delta V_{T-2})^2]\\
&= \min_{\Delta V} [2\eta \Delta V_{T-2}^2 - (\gamma \rho^2 Y_{T-2} +2 \eta X_{T-2} )\Delta V_{T-2}+(\widetilde{S}_{T-2}+ \gamma \rho (1+\rho)Y_{T-2})X_{T-2} +\eta X_{T-2}^2].\\
\end{split}
\end{equation*}
In order to find the minimum of this expression, we set to zero its derivative with respect to $\Delta V_{T-2}$: \\
\begin{equation*}
\frac{\partial C(X_{T-2}, S_{T-2})}{\partial \Delta V_{T-2}} = 4 \eta \Delta V_{T-2} - \gamma \rho^2 Y_{T-2} - 2 \eta X_{T-2} = 0.
\end{equation*}
The solution of this equation is the optimal amount to execute at time $T-2$: \\
\begin{equation*}
\Delta V_{T-2}^* = \frac{X_{T-2}}{2} + \frac{\gamma \rho^2}{4\eta} Y_{T-2}.
\end{equation*}
Substituting this expression in the expected cost yields: \\
\begin{equation*}
\begin{split}
C^*(X_{T-2}, S_{T-2}) &= 2\eta (\Delta V_{T-2}^*)^2 - (\gamma \rho^2 Y_{T-2} +2 \eta X_{T-2} )\Delta V_{T-2}^*+(\widetilde{S}_{T-2}+ \gamma \rho (1+\rho)Y_{T-2})X_{T-2} +\eta X_{T-2}^2 \\
&= 2\eta \left(\frac{X_{T-2}}{2} + \frac{\gamma \rho^2}{4\eta} Y_{T-2}\right)^2 - (\gamma \rho^2 Y_{T-2} +2 \eta X_{T-2} )\left(\frac{X_{T-2}}{2} + \frac{\gamma \rho^2}{4\eta} Y_{T-2}\right)\\ &+(\widetilde{S}_{T-2}+ \gamma \rho (1+\rho)Y_{T-2})X_{T-2} +\eta X_{T-2}^2\\
&= \widetilde{S}_{T-2} X_{T-2} + \eta \frac{X_{T-2}^2}{2} + \gamma \rho (1 + \frac{\rho}{2}) X_{T-2} Y_{T-2} - \frac{\gamma^2 \rho^4}{8 \eta} Y_{T-2}^2.\\
\end{split}
\end{equation*}
Using this expression, we can now compute the optimal strategy one step backward. \\
\begin{equation*}
\begin{split}
C^*(X_{T-3}, S_{T-3}) &= \min_{\Delta V} \mathbb{E}_{T-3}[S_{T-2} \Delta V_{T-3} + C^*(X_{T-2}, S_{T-2})]\\
&= \min_{\Delta V} \mathbb{E}_{T-3}[S_{T-2} \Delta V_{T-3} + \widetilde{S}_{T-2} X_{T-2} + \eta \frac{X_{T-2}^2}{2} + \gamma \rho (1 + \frac{\rho}{2}) X_{T-2} Y_{T-2} - \frac{\gamma^2 \rho^4}{8 \eta} Y_{T-2}^2]\\
&= \min_{\Delta V} [(\widetilde{S}_{T-3} + \gamma \rho Y_{T-3} + \eta \Delta V_{T-3}) \Delta V_{T-3} + (\widetilde{S}_{T-3} + \gamma \rho Y_{T-3})(X_{T-3} - \Delta V_{T-3})\\
&+ \frac{\eta}{2} (X_{T-3} - \Delta V_{T-3})^2 + \gamma \rho^2 (1 + \frac{\rho}{2}) Y_{T-3} (X_{T-3} - \Delta V_{T-3}) - \frac{\gamma^2 \rho^4}{8 \eta} (\rho^2 Y_{T-3}^2+\sigma_Y^2)]\\
&= \min_{\Delta V} [\frac{3\eta}{2} \Delta V_{T-3}^2 - (\eta X_{T-3}+\gamma \rho^2 (1 + \frac{\rho}{2}) Y_{T-3}) \Delta V_{T-3} + \widetilde{S}_{T-3}X_{T-3}+\frac{\eta}{2}X_{T-3}^2\\&+\gamma \rho (1 + \rho + \frac{\rho^2}{2}) Y_{T-3}X_{T-3}- \frac{\gamma^2 \rho^4}{8 \eta} (\rho^2 Y_{T-3}^2+\sigma_Y^2)].\\
\end{split}
\end{equation*}
In order to find the minimum of this expression, we set to zero its derivative with respect to $\Delta V_{T-3}$: \\
\begin{equation*}
\frac{\partial C(X_{T-3}, S_{T-3})}{\partial \Delta V_{T-3}} = 3 \eta \Delta V_{T-3} - \eta X_{T-3} - \gamma \rho^2 (1 + \frac{\rho}{2}) Y_{T-3} = 0.
\end{equation*}
The solution of this equation is the optimal amount to execute at time $T-3$: \\
\begin{equation*}
\Delta V_{T-3}^* = \frac{X_{T-3}}{3} + \frac{\gamma \rho^2 (\rho+2)}{6 \eta} Y_{T-3}.
\end{equation*}
We can compute the expected cost at time $T-3$. \\
\begin{equation*}
\begin{split}
C^*(X_{T-3}, S_{T-3})&= \frac{3\eta}{2} (\Delta V_{T-3}^*)^2 - (\eta X_{T-3}+\gamma \rho^2 (1 + \frac{\rho}{2}) Y_{T-3}) \Delta V_{T-3}^* + \widetilde{S}_{T-3}X_{T-3}+\frac{\eta}{2}X_{T-3}^2\\&+\gamma \rho (1 + \rho + \frac{\rho^2}{2}) Y_{T-3}X_{T-3}- \frac{\gamma^2 \rho^4}{8 \eta} (\rho^2 Y_{T-3}^2+\sigma_Y^2) \\
&= \frac{3\eta}{2} \left(\frac{X_{T-3}}{3} + \frac{\gamma \rho^2 (\rho+2)}{6 \eta} Y_{T-3}\right)^2+ \widetilde{S}_{T-3}X_{T-3}+\frac{\eta}{2}X_{T-3}^2+\gamma \rho (1 + \rho + \frac{\rho^2}{2}) Y_{T-3}X_{T-3}\\& - (\eta X_{T-3}+\gamma \rho^2 (1 + \frac{\rho}{2}) Y_{T-3}) \left(\frac{X_{T-3}}{3} + \frac{\gamma \rho^2 (\rho+2)}{6 \eta} Y_{T-3}\right) - \frac{\gamma^2 \rho^4}{8 \eta} (\rho^2 Y_{T-3}^2+\sigma_Y^2) \\
&= \widetilde{S}_{T-3} X_{T-3} + \frac{\eta}{3} X_{T-3}^2 + \frac{ \rho^2 + 2 \rho + 3 }{3} \gamma \rho X_{T-3} Y_{T-3}- \frac{\gamma^2 \rho^4}{8 \eta}\left((\frac{(\rho+2)^2}{3} + \rho^2) Y_{T-3}^2 +  \sigma_Y^2\right).\\
\end{split}
\end{equation*}
In a similar way as in the case of a permanent impact, we deduce from these results a formula for the optimal execution strategy: \\
\begin{proposition}[Optimal execution strategy]
For any $i \geq 1$ the optimal execution strategy at time $T-i$ is
\begin{equation*}
\begin{split}
&\Delta V_{T-i} = \frac{X_{T-i}}{i} + a'_i Y_{T-i}\\
&\text{with }a'_i = \frac{\gamma \sum_{k=1}^{i-1}{(i-k) \rho^{k+1}}}{2 i \eta} \text{for i } \geq 2, \text{and } a_1 = 0.
\end{split}
\end{equation*}
\\
\end{proposition}

\begin{remark}
$a'_i$ can be simplified to
\begin{equation*}
a'_i = \frac{\gamma \rho^2}{2i\eta(1-\rho)^2}(\rho^i-i\rho+i-1) \quad \text{for i }\geq 1.
\end{equation*}
\end{remark}
\begin{proof}
The proof is the same as in the permanent impact case.
\end{proof}
\begin{proposition}
For any $i \geq 1$ the optimal expected cost at time $T-i$ is \\
\begin{equation*}
\begin{split}
&C^*_{ad}(X_{T-i}, S_{T-i}) = \widetilde{S}_{T-i} X_{T-i} + \eta \frac{X_{T-i}^2}{i} + \frac{2(i+1) \eta a'_{i+1}}{i \rho} X_{T-i} Y_{T-i} - b'_i Y_{T-i}^2 - (\sum_{k=2}^{i-1} b'_k) \sigma_Y^2\\
&\text{with } b'_i = \sum_{k=2}^{i} {\frac{\eta i \rho^{2(i-k)}}{i-1} (a'_k)^2} \text{ for i } \geq 2.
\end{split}
\end{equation*}
\\
\end{proposition}
\begin{remark}
$b'_i$ can be simplified to
\begin{equation*}
b'_i = \frac{\gamma^2\rho^4}{4\eta(1-\rho)^3} \left( \frac{1-\rho^{2i}}{1+\rho} -\frac{(1-\rho^i)^2}{i(1-\rho)}\right) \quad \text{for i }\geq 2.
\end{equation*}
\end{remark}
\begin{proof}
The proof is the same as in the permanent impact case.
\end{proof}
\begin{corollary}
In particular, the optimal expected cost at time $0$ is \\
\begin{equation}
C^*_{ad}(X_{0}, S_{0}) = S_{0} X + \eta \frac{X^2}{T} + \frac{2(T+1) \eta a'_{T+1}}{T \rho} X Y_{0} - b'_T Y_{0}^2 - (\sum_{k=2}^{T-1} b'_k) \sigma_Y^2.
\end{equation}
\end{corollary}

\subsection{Temporary market impact: optimal deterministic solution}

\begin{theorem}[Optimal deterministic execution strategy]
When we restrict the solutions to the subset of deterministic strategies, the optimal strategy is \\
\begin{equation}
X_t^* = \frac{T-t}{T}X + \frac{\gamma Y_0\rho^2}{2\eta (1-\rho)^2}\left[\rho^t - 1 + (1-\rho^T)\frac{t}{T}\right].
\end{equation}
\end{theorem}
\begin{proof}
We constrain the solution of \eqref{DisCost} to be deterministic, using the same method as in the case of a permanent impact.
The expected cost at time 0 is \\
\begin{align*}
C(X_0,S_0,\{\Delta V\}) &= \mathbb{E}_0 \left[ \sum_{t=0}^{T-1} (\widetilde{S}_{t+1} + \eta \Delta V_{t}) \Delta V_t\right] \\
&= \sum_{t=0}^{T-1} \left( \mathbb{E}_0[\widetilde{S}_{t+1}] + \eta \Delta V_{t}\right) \Delta V_{t} \text{ since the strategy is deterministic}\\
&= \sum_{t=0}^{T-1} \left( \widetilde{S}_0 + \gamma \sum_{i=1}^{t+1} \rho^{i} Y_0 + \eta \Delta V_{t}\right) \Delta V_{t}\\
&= S_0 \sum_{t=0}^{T-1} \Delta V_{t} + \gamma Y_0 \sum_{t=0}^{T-1} \Delta V_{t}\frac{\rho - \rho^{t+2}}{1-\rho} + \eta \sum_{t=0}^{T-1} (\Delta V_{t})^2 \\
&= S_0 X_0 + \gamma \rho Y_0 \sum_{t=0}^{T-1} \frac{1-\rho^{t+1}}{1-\rho} (X_t - X_{t+1})+ \eta \sum_{t=0}^{T-1} (X_t - X_{t+1})^2.
\end{align*}
Problem \eqref{DisCost} can be rewritten as
\begin{align}
C^*(X_0,S_0)&=\min_{x} C(x).
\end{align}
To find the minimum, we set to zero the partial derivatives of the expected cost with respect to $X_1$, ..., $X_{T-1}$. For $t=1,...,T-1$ it gives us \\
\begin{align}
\frac{\partial C}{\partial X_t} = \gamma \rho Y_0 \left(\frac{1-\rho^{t+1}}{1-\rho} - \frac{1-\rho^{t}}{1-\rho}\right) + 2\eta (2X_t - X_{t+1} - X_{t-1}) = 0.
\end{align}
We obtain the difference equation \\
\begin{equation}\label{diffTemp}
X_{t+1} - 2X_t + X_{t-1} = \frac{\gamma Y_0 }{2\eta} \rho^{t+1},
\end{equation}
with boundary conditions $X_0 = X$ and $X_T = 0$.

The solution of \eqref{diffTemp} is of the form $A + Bt + C\rho^t$ for some constants $A$, $B$ and $C$.
Substituting this expression back in the equation yields 
\begin{eqnarray*}
&& A+B(t+1)+C\rho^{t+1}-2(A + Bt + C\rho^t)+A + B(t-1) + C\rho^{t-1}= \frac{\gamma Y_0}{2\eta} \rho^{t+1}\\
&& C\rho^t(\rho-2+\rho^{-1}) = \frac{\gamma Y_0}{2\eta} \rho^{t+1}\\
&& C = \frac{\gamma Y_0\rho^2}{2\eta (1-\rho)^2}. \\
\end{eqnarray*}
From the boundary conditions we have 
\[
X_0 = A + C = X, \  \
A = X - \frac{\gamma Y_0\rho^2}{2\eta (\rho-1)^2},\]
and 
\[
X_T = A + BT + C\rho^T = 0, \ \ 
B = -\frac{X}{T} + \frac{\gamma Y_0\rho^2(1-\rho^T)}{2\eta (\rho-1)^2 T}. \]
Combining those, we obtain the closed-form formula of the optimal optimal deterministic solution.
\end{proof}
\begin{remark}
As in the case of a permanent impact, the strategy is a VWAP when there is no relevant initial information.
\end{remark}

\begin{theorem}[Optimal expected cost associated with the deterministic strategy]
The expected cost at time $0$ obtained when using the optimal deterministic strategy is \\
\begin{equation}
C^*_{det}(X_0,S_0)= S_0 X + \eta \frac{X^2}{T} + \frac{\gamma \rho Y_0 X}{T(1-\rho)}\left(T-\rho\frac{1-\rho^{T}}{1-\rho}\right) + \frac{\gamma^2 Y_0^2\rho^4}{4\eta (1-\rho)^3}\left(\frac{(1-\rho^{T})^2}{T(1-\rho)} -\frac{1-\rho^{2T}}{1+\rho} \right).
\end{equation}
\end{theorem}
\begin{proof}
Replacing $X_t$ with $X_t^*$ in the expression of the expected cost at time $0$ gives \\
\begin{align*}
&C^*(X_0,S_0)= S_0 X_0 + \gamma \rho Y_0 \sum_{t=0}^{T-1} \frac{1-\rho^{t+1}}{1-\rho}(X_t^* - X_{t+1}^*) + \eta \sum_{t=0}^{T-1} (X_t^* - X_{t+1}^*)^2 \\
&= S_0 X + \gamma \rho Y_0 \sum_{t=0}^{T-1} \frac{1-\rho^{t+1}}{1-\rho} \left(\frac{X}{T} + \frac{\gamma Y_0\rho^2}{2\eta (1-\rho)^2}\left(\rho^t (1 - \rho) - \frac{1-\rho^T}{T} \right)\right) \\
&+ \eta \sum_{t=0}^{T-1} \left(\frac{X}{T} + \frac{\gamma Y_0\rho^2}{2\eta (1-\rho)^2}\left(\rho^t (1 - \rho) - \frac{1-\rho^T}{T} \right)\right)^2 \\
&= S_0 X + \frac{\gamma \rho Y_0 X}{T}\sum_{t=0}^{T-1} \frac{1-\rho^{t+1}}{1-\rho} + \frac{\gamma^2 Y_0^2\rho^3}{2\eta (1-\rho)^2}\sum_{t=0}^{T-1} \frac{1-\rho^{t+1}}{1-\rho}\left(\rho^t (1 - \rho) - \frac{1-\rho^T}{T} \right) \\
&+ \eta \sum_{t=0}^{T-1} \left(\frac{X^2}{T^2} + \frac{\gamma^2 Y_0^2\rho^4}{4\eta^2 (1-\rho)^4}\left(\rho^t (1 - \rho) - \frac{1-\rho^T}{T} \right)^2+2\frac{X}{T}\frac{\gamma Y_0\rho^2}{2\eta (1-\rho)^2}\left(\rho^t (1 - \rho) - \frac{1-\rho^T}{T} \right)\right) \\
&= S_0 X + \frac{\gamma \rho Y_0 X}{T(1-\rho)}\left(T-\rho\frac{1-\rho^{T}}{1-\rho}\right) + \frac{\gamma^2 Y_0^2\rho^3}{2\eta (1-\rho)^3}\sum_{t=0}^{T-1}\left( (1 - \rho)(\rho^t-\rho^{2t+1}) - (1-\rho^{t+1})\frac{1-\rho^T}{T}\right) + \eta \frac{X^2}{T} \\
&+ \sum_{t=0}^{T-1}\frac{\gamma^2 Y_0^2\rho^4}{4\eta (1-\rho)^4}\left(\rho^{2t} (1 - \rho)^2 + \frac{(1-\rho^T)^2}{T^2} - 2\rho^t (1 - \rho)\frac{1-\rho^T}{T} \right) +\frac{X\gamma Y_0\rho^2}{ T(1-\rho)^2}\sum_{t=0}^{T-1}\left(\rho^t (1 - \rho) - \frac{1-\rho^T}{T} \right) \\
&= S_0 X  + \frac{\gamma \rho Y_0 X}{T(1-\rho)}\left(T-\rho\frac{1-\rho^{T}}{1-\rho}\right) + \frac{\gamma^2 Y_0^2\rho^3}{2\eta (1-\rho)^3}\left( 1-\rho^T-\rho\frac{1-\rho^{2T}}{1+\rho} - \left(T-\rho\frac{1-\rho^{T}}{1-\rho}\right)\frac{1-\rho^T}{T}\right)  \\
&+ \eta \frac{X^2}{T}+ \frac{\gamma^2 Y_0^2\rho^4}{4\eta (1-\rho)^4}\left(\frac{1-\rho^{2T}}{1-\rho^2} (1 - \rho)^2 + \frac{(1-\rho^T)^2}{T} - 2\frac{(1-\rho^T)^2}{T} \right) +\frac{X\gamma Y_0\rho^2}{ T(1-\rho)^2}\left(1-\rho^T - 1+\rho^T \right) \\
&= S_0 X  + \frac{\gamma \rho Y_0 X}{T(1-\rho)}\left(T-\rho\frac{1-\rho^{T}}{1-\rho}\right) + \frac{\gamma^2 Y_0^2\rho^4}{2\eta (1-\rho)^3}\left( -\frac{1-\rho^{2T}}{1+\rho} +\frac{(1-\rho^{T})^2}{T(1-\rho)}\right)+ \eta \frac{X^2}{T} \\
&+ \frac{\gamma^2 Y_0^2\rho^4}{4\eta (1-\rho)^3}\left(\frac{1-\rho^{2T}}{1+\rho}  - \frac{(1-\rho^T)^2}{T(1-\rho)} \right).
\end{align*}
\\
\end{proof}

\subsection{Temporary market impact: adapted vs deterministic solution}

We define the absolute and relative differences the same way as in the case of a permanent impact: \\
\begin{definition}[Absolute difference]
\begin{align*}
\epsilon_{abs}&:= C_{det}^*(X_0,S_0)-C_{ad}^*(X_0,S_0) \\
&=\frac{\gamma^2 Y_0^2\rho^4}{4\eta (1-\rho)^3}\left(\frac{(1-\rho^{T})^2}{T(1-\rho)} -\frac{1-\rho^{2T}}{1+\rho} \right) + b'_T Y_{0}^2 + (\sum_{k=2}^{T-1} b'_k) \sigma_Y^2.
\end{align*}
\end{definition}
\begin{proposition}[Value of the absolute difference]
The value of the absolute difference is 
\begin{equation}
\epsilon_{abs} = (\sum_{k=2}^{T-1} b'_k) \sigma_Y^2.
\end{equation}
\end{proposition}
\begin{proof}
The proof is the same as in the case of a permanent market impact.
\end{proof}
\begin{corollary}
The two strategies have the same expected cost when the information process is not random ($\sigma_Y=0$). In this case the optimal adapted solution turns out to be static.
\end{corollary}
\begin{corollary}
As expected, the fully adapted strategy is always better (or equal) than the deterministic one, in that it results in a smaller or equal criterion value. 
\end{corollary}
\begin{definition}[Relative difference]
\begin{equation*}
\epsilon_{rel} :=\frac{\epsilon_{abs}}{C_{det}^*(X_0,S_0)}.
\end{equation*}
\end{definition}

We now quantify the difference between the deterministic and the adapted strategies through a few numerical examples. As in the permanent impact case, we set $X=10^6$, $S_0=\$100$, $T=14$, $\sigma=0.51\%$, $\rho=0.5$, $\gamma=1$, $\sigma_Y=0.44$. The market impact $\eta=10^{-5}$ is chosen to increase the expected price by $10\%$ if the execution is made entirely in the first period, assuming no initial information:
\begin{equation*}
 (S_0+\eta X)X = 1.1 S_0X
\end{equation*}
In a first part we assume that there is no initial information $Y_0=0$.


The values described above are summarized in Table \ref{benchTemp}.
\begin{table}[!h]
\centering
    \begin{tabular}{| l | l |}
    \hline
   $X$ & $10^6$ \\ \hline
   $S_0$ & $100$ \\ \hline
   $T$ & $14$ \\ \hline
   $\eta$& $10^{-5}$   \\ \hline
   $\sigma$ & $0.51\%$  \\ \hline
   $\rho$ & $0.5$   \\ \hline
   $\gamma$ & $1$   \\ \hline
   $\sigma_Y$ & $0.44$  \\ \hline
   $Y_0$ & $0$ \\ \hline
    \end{tabular}
\caption{Benchmark parameter values}
\label{benchTemp}
\end{table}

To get an idea of the influence of the initial information on the strategies, we give a few examples of paths for different values of $Y_0$ in Figures \ref{tempSim}, \ref{tempSimY5} and \ref{tempSimY_5}.
\begin{figure}[H]
\centering
\includegraphics[scale=0.8]{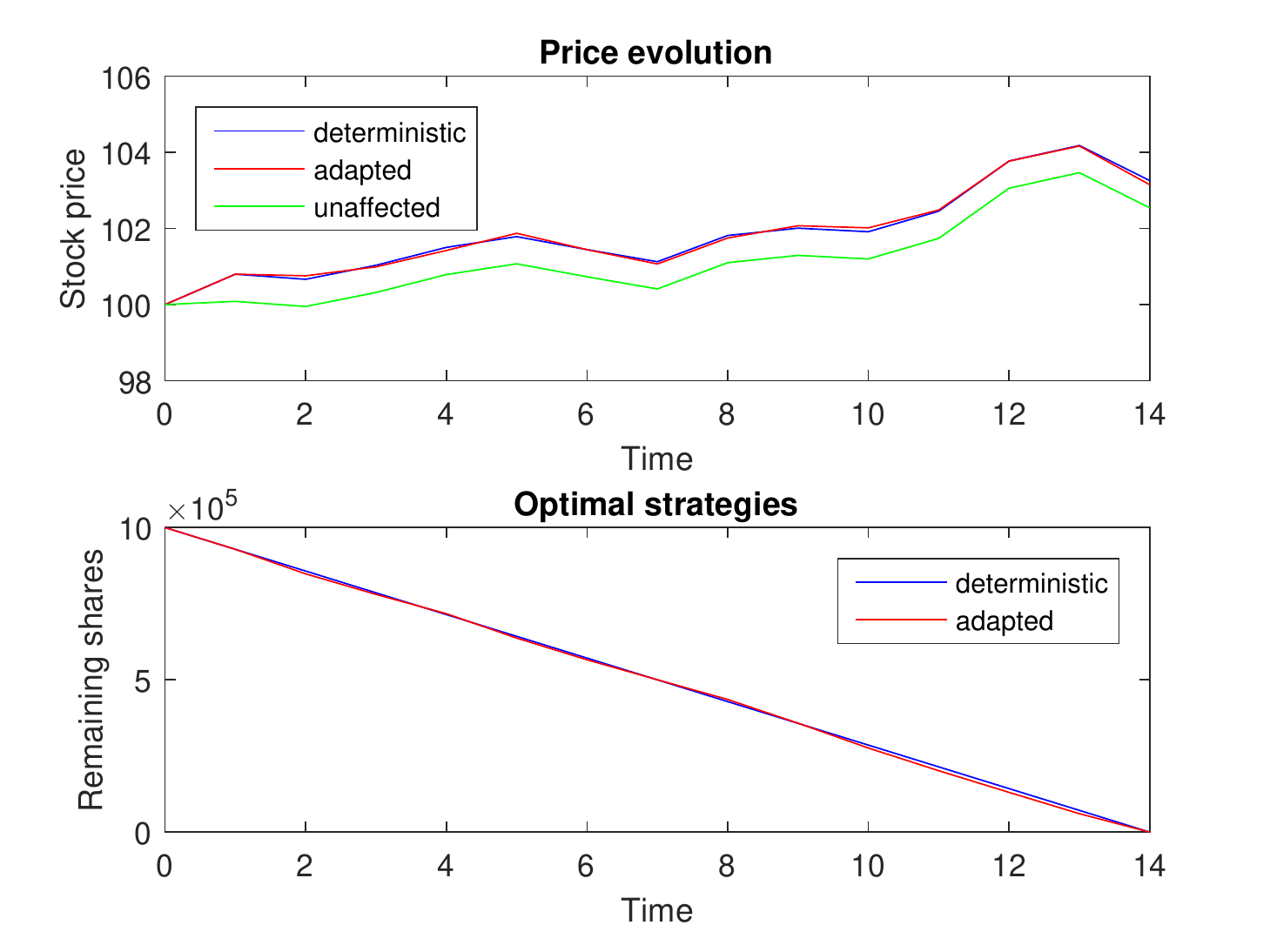}
\caption{One path of a simulated strategy with benchmark parameters ($Y_0=0$)}
\label{tempSim}
\end{figure}
\begin{figure}[H]
\centering
\includegraphics[scale=0.8]{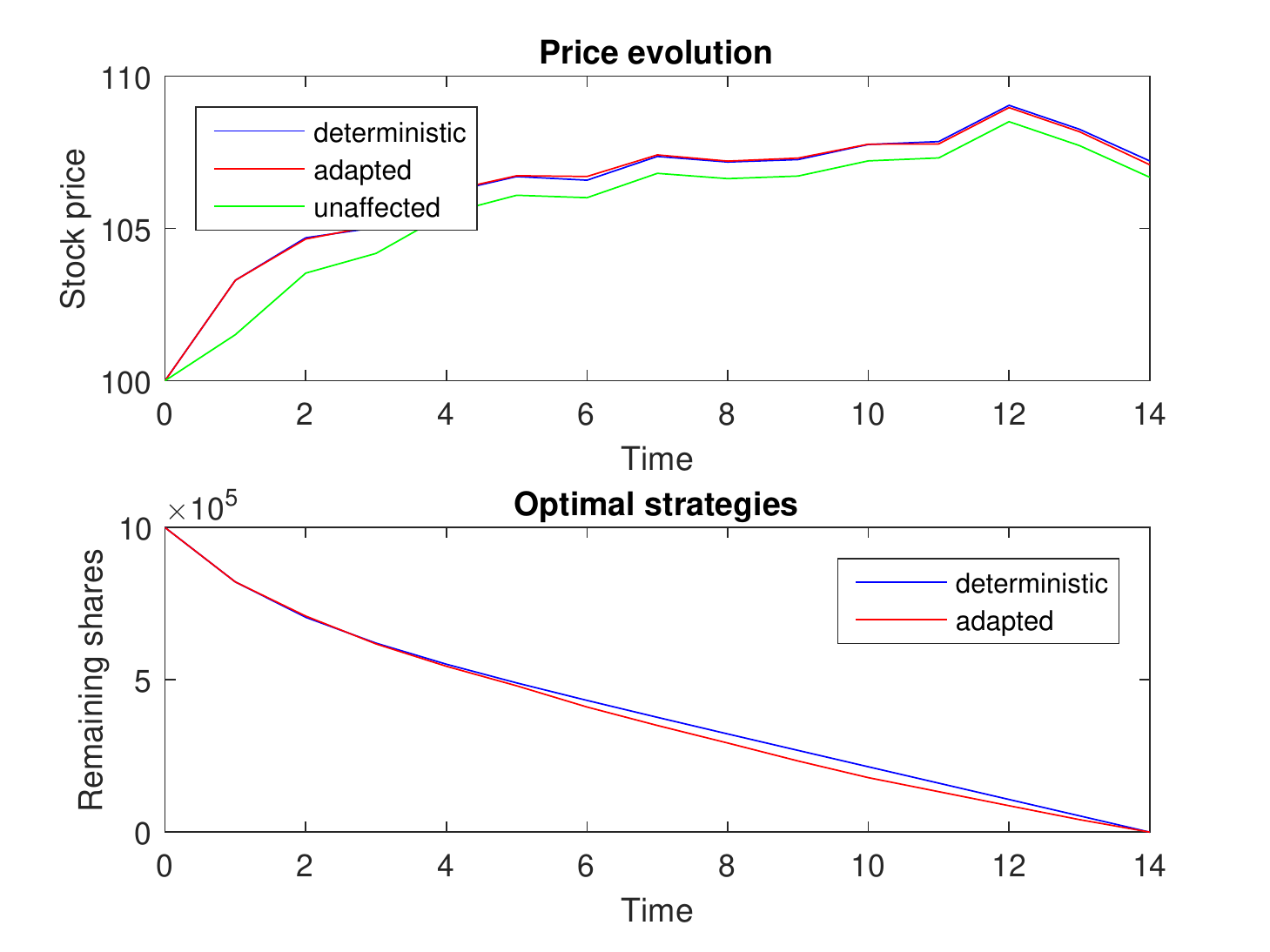}
\caption{One path of a simulated strategy with positive initial information ($Y_0=5$)}
\label{tempSimY5}
\end{figure}
\begin{figure}[h]
\centering
\includegraphics[scale=0.8]{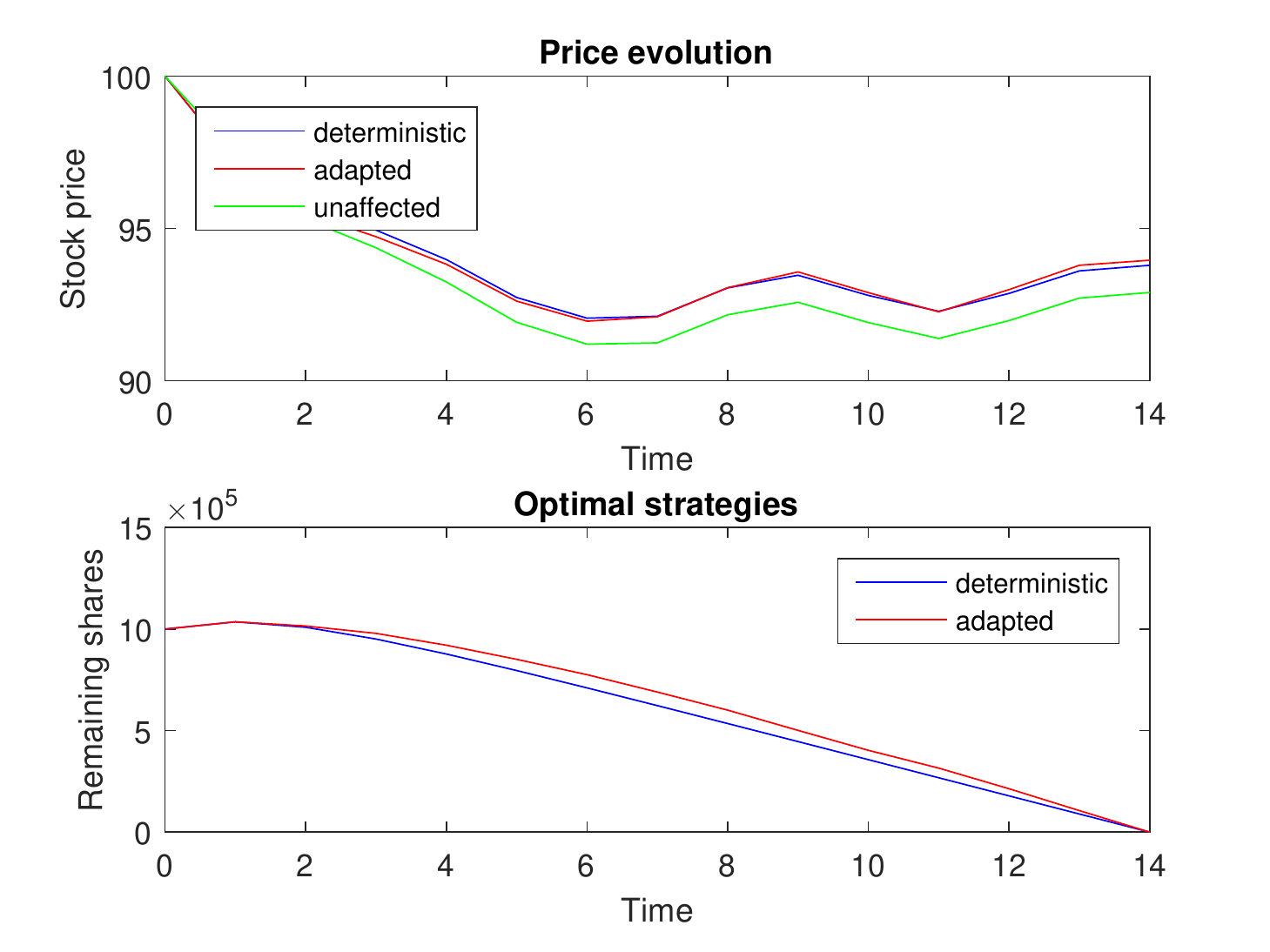}
\caption{One path of a simulated strategy with negative initial information ($Y_0=-5$)}
\label{tempSimY_5}
\end{figure}

With the benchmark parameters, we find that $C_{det}^*(X_0,S_0)=1.0071 \times 1	0^8$, $C_{ad}^*(X_0,S_0)=1.0070 \times 10^8$ and $\epsilon_{rel}=1.03 \times 10^{-4}$. In particular, the costs obtained with the path shown in Figure \ref{tempSim} are $C_{det}(X_0,S_0)=1.0198 \times 10^8$ and $C_{ad}(X_0,S_0)=1.0197 \times 10^8$.
With $Y_0=5$, we find that $C_{det}^*(X_0,S_0)=1.0519 \times 10^8$, $C_{ad}^*(X_0,S_0)=1.0518 \times 10^8$ and $\epsilon_{rel}=9.85 \times 10^{-5}$. In particular, the costs obtained with the path shown in Figure \ref{tempSimY5} are $C_{det}(X_0,S_0)=1.0612 \times 10^8$ and $C_{ad}(X_0,S_0)=1.0609 \times 10^8$.
With $Y_0=-5$, we find that $C_{det}^*(X_0,S_0)=9.5908 \times 10^7$, $C_{ad}^*(X_0,S_0)=9.5897 \times 10^7$ and $\epsilon_{rel}=1.08 \times 10^{-4}$. In particular, the costs obtained with the path shown in Figure \ref{tempSimY_5} are $C_{det}(X_0,S_0)=9.3002 \times 10^7$ and $C_{ad}(X_0,S_0)=9.3004 \times 10^7$. For this path, the adapted strategy is less effective than the deterministic one.

As in the case of a permanent impact, both strategies are aggressive.
\\

Now that we have a feel for the paths obtained in a few examples, we will study the influence of each parameter separately, as we did for the permanent impact case, considering parameters
\[  X, T, \eta, \rho, \gamma, \sigma_Y \ .\]
In each numerical example, the parameters will be those of Table \ref{benchTemp} except for the one whose influence we study. This allows us to study one parameter at a time.
\begin{remark}
 Since $\sigma$ does not appear in the formulas in either case, it has no influence on the optimal expected cost.
\end{remark}

We begin by considering the influence of $X$.
\begin{figure}[h]
	\begin{minipage}[c]{.54\linewidth}
		\includegraphics[scale=0.6]{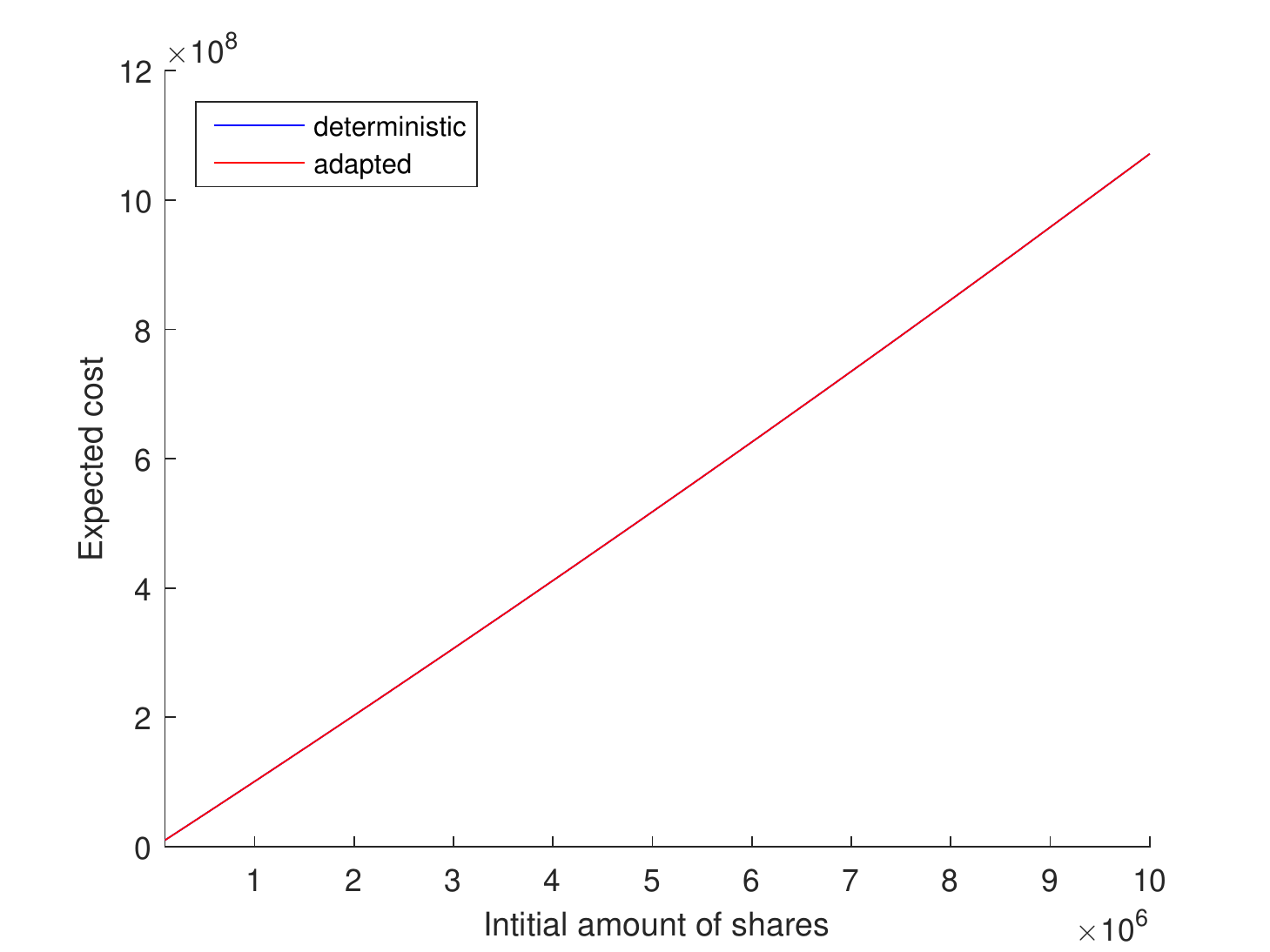}
	\end{minipage}
	\begin{minipage}[c]{.46\linewidth}
		\includegraphics[scale=0.6]{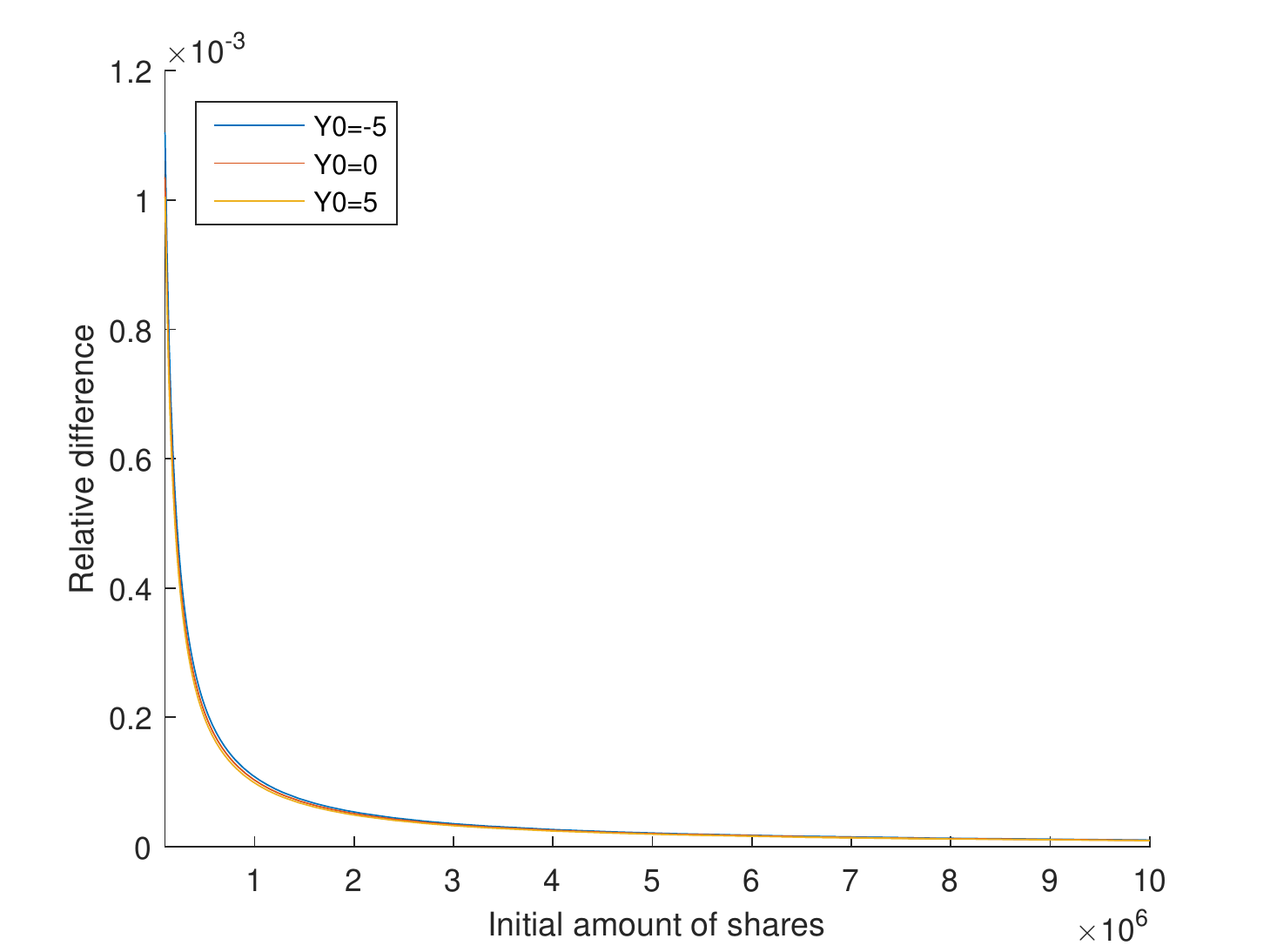}
	\end{minipage}
\caption{Influence of $X$ on the expected costs and relative difference}
\label{tempInfX}
\end{figure}
Figure \ref{tempInfX} shows the evolution of the expected costs and the relative difference when $X$ varies from $10^{5}$ to $10^{7}$. The influence of $X$ is similar as in the case of a permanent impact. This is due to the fact that the market impact parameter $\eta$ has been calibrated for a certain $X$, and its influence becomes overwhelming when $X$ is too big. It is not really representative of the impact of $X$ since $\eta$ should be a function of $X$.

As regards the influence of $T$, figure \ref{tempInfT} shows the evolution of the expected costs and the relative difference when $T$ varies from $1$ to $70$.
The relative difference between the two strategies increases linearly with the time horizon for $T$ large enough, for the same reason as with a permanent impact. For a time horizon of $5$ days, the adapted strategy is $0.1\%$ better than the deterministic one.
\begin{figure}[H]
	\begin{minipage}[c]{.54\linewidth}
		\includegraphics[scale=0.6]{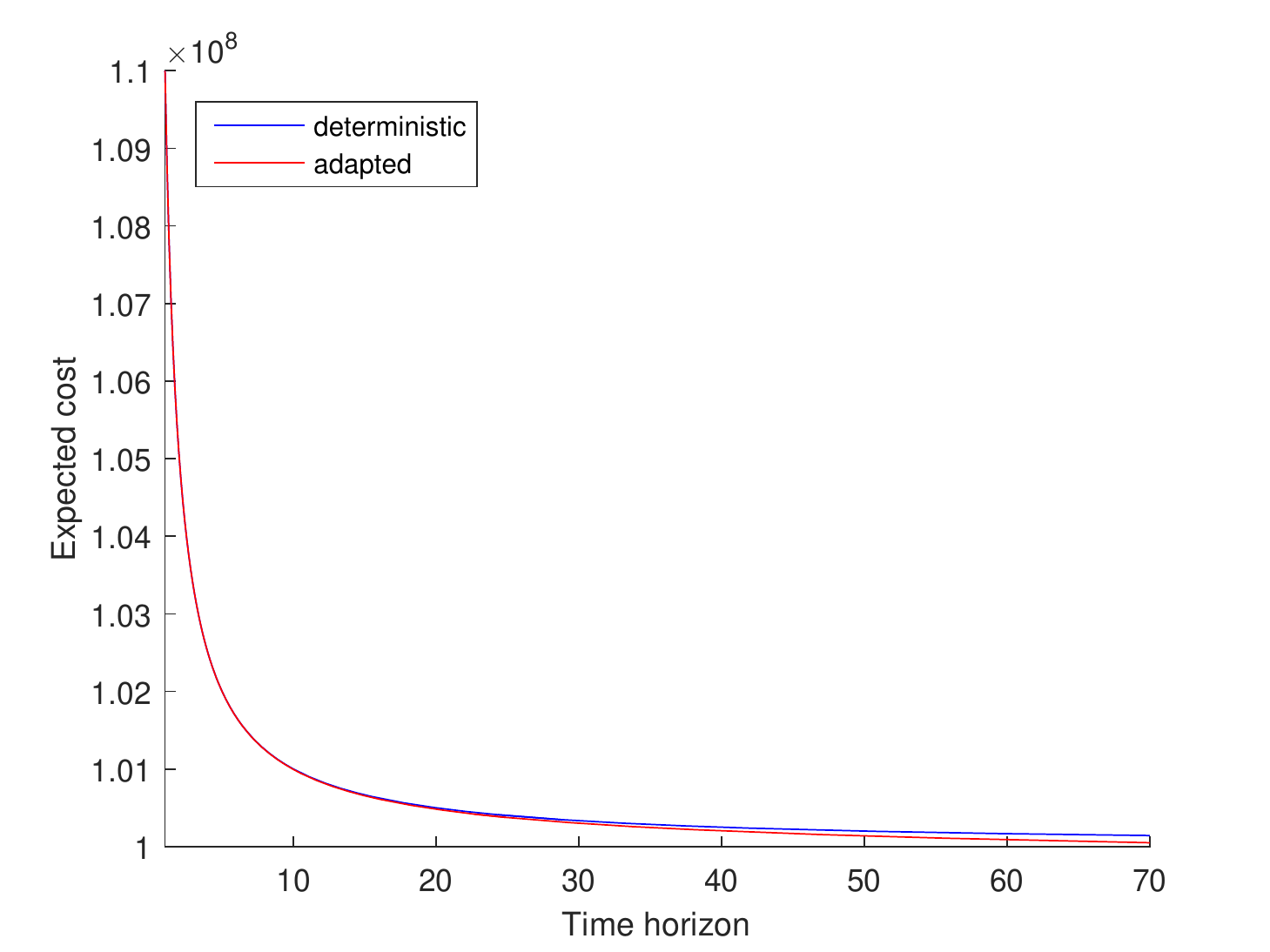}
	\end{minipage}
	\begin{minipage}[c]{.46\linewidth}
		\includegraphics[scale=0.6]{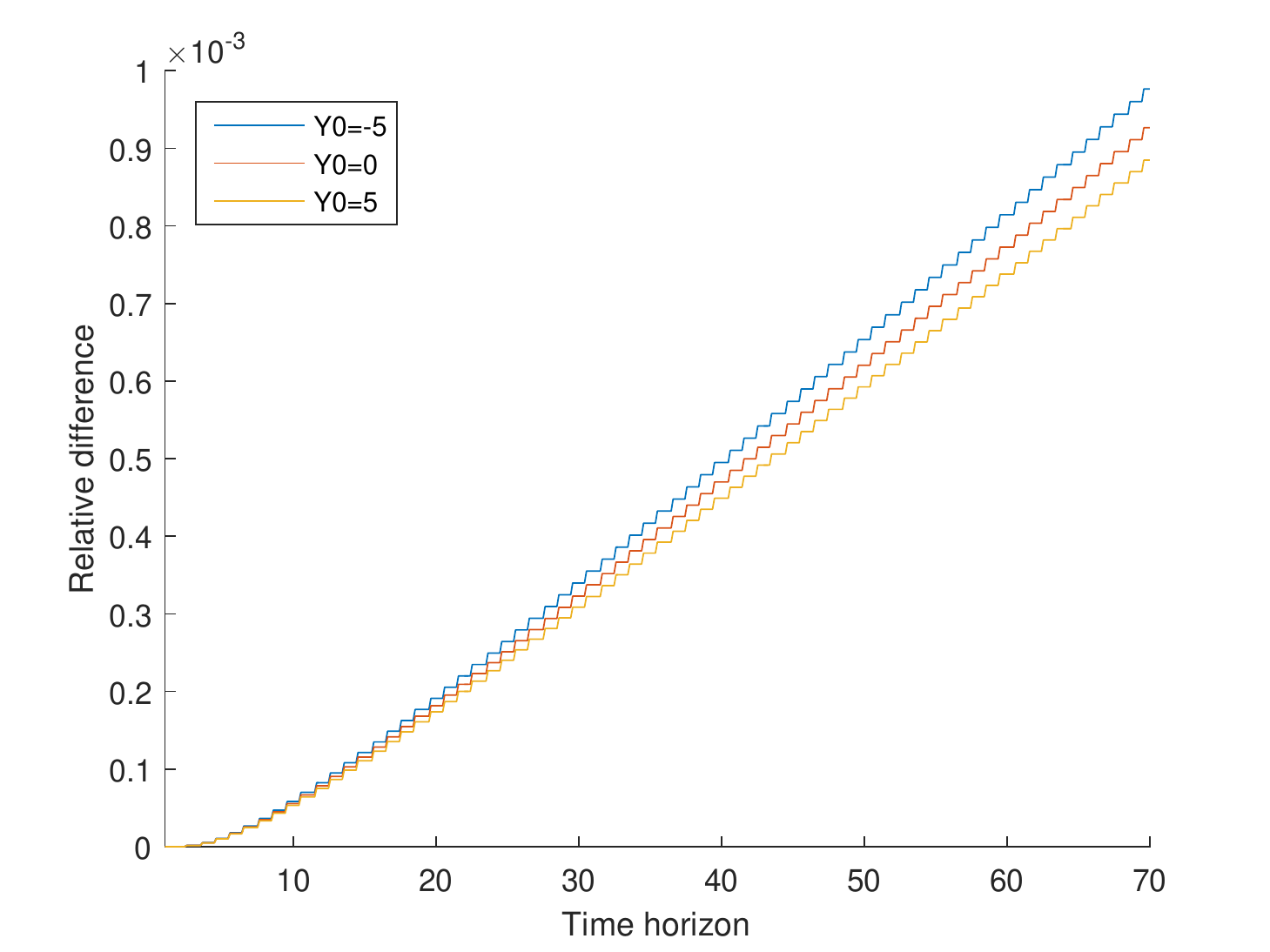}
	\end{minipage}
\caption{Influence of $T$ on the expected costs and relative difference}
\label{tempInfT}
\end{figure}

We now move to the Influence of $\eta$.
\begin{figure}[h]
	\begin{minipage}[c]{.54\linewidth}
		\includegraphics[scale=0.6]{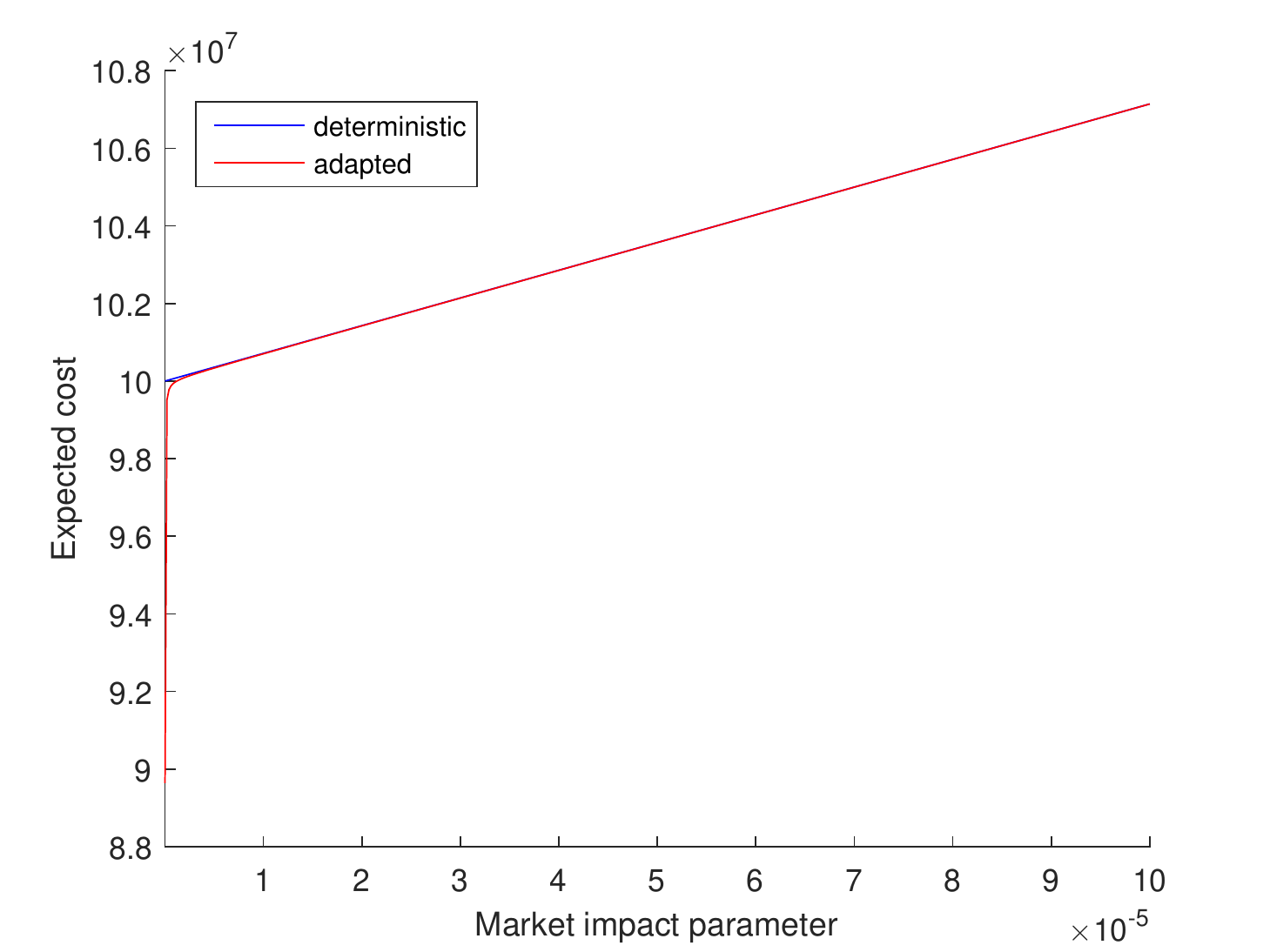}
	\end{minipage}
	\begin{minipage}[c]{.46\linewidth}
		\includegraphics[scale=0.6]{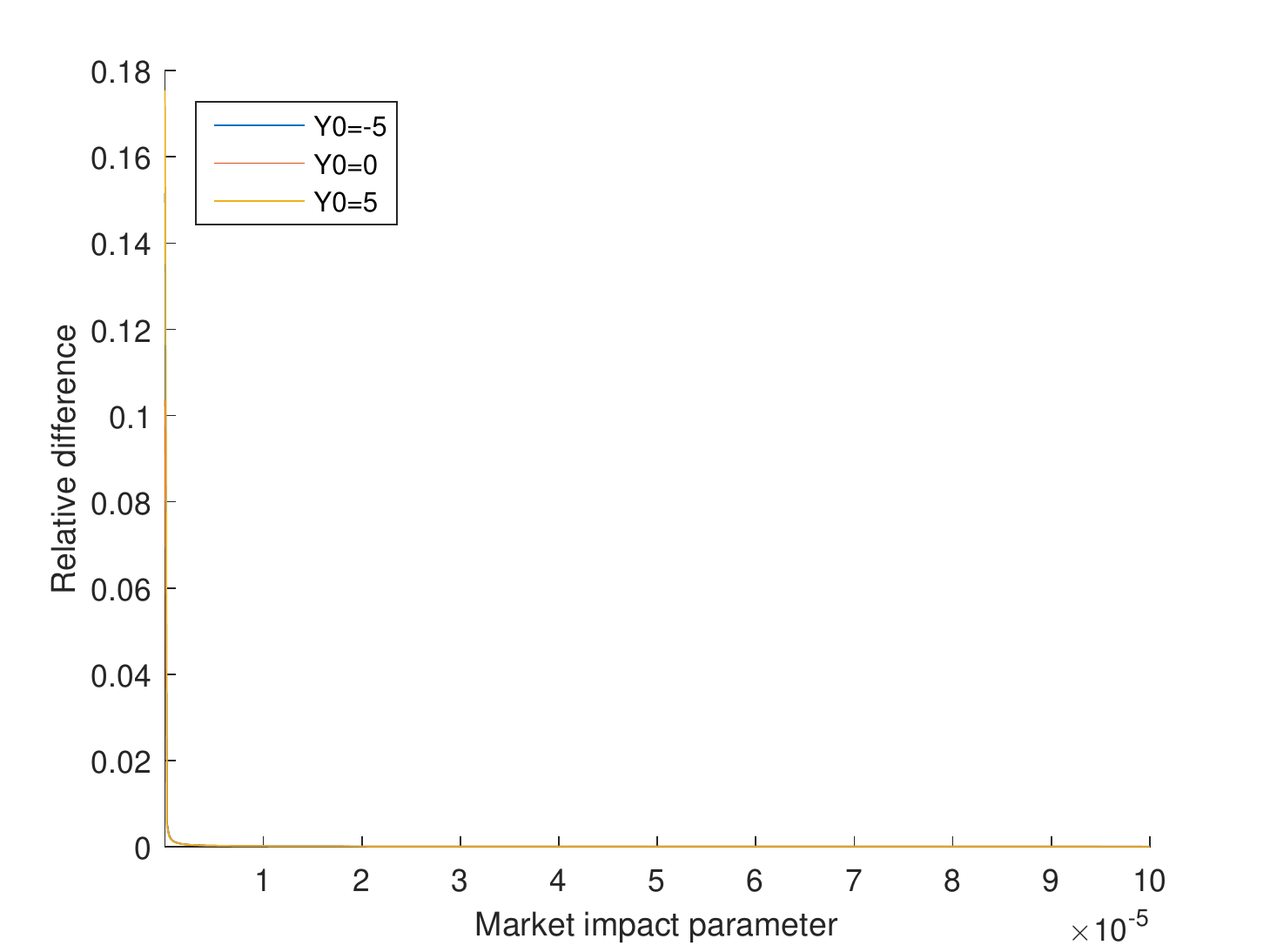}
	\end{minipage}
\caption{Influence of $\eta$ on the expected costs and relative difference}
\label{tempInfTH}
\end{figure}
Figure \ref{tempInfTH} shows the evolution of the expected costs and the relative difference when $\eta$ varies from $10^{-8}$ to $10^{-4}$.
Similarly to the permanent market impact, when the temporary impact parameter $\eta$ increases it becomes difficult to reduce its impact on the cost. Hence the difference between the two strategies is crushed by the total expected cost. For a small market impact, for example an increase of only $1\%$ if the execution is fully done in one period ($\eta=10^{-6}$), the adapted strategy is $0.1\%$ better than the deterministic one.
\begin{remark}
As long as $\sigma_Y \neq 0$, the optimal expected cost tends to $-\infty$ when $\eta$ tends to $0$. When there is initial information ($Y_0 \neq 0$), the expected cost associated with the best deterministic strategy tends to $-\infty$ when $\eta$ tends to $0$. The reasons why this is the case are the same as with a permanent impact, which is intuitive since when there is no impact, it does not matter whether it would be permanent or temporary.
\end{remark}

We now analyze the influence of $\rho$. Figure \ref{tempInfR} shows the evolution of the expected costs and the relative difference when $\rho$ varies from $-0.9$ to $0.9$.
As in the permanent impact case, the relative difference is particularly relevant when the information process is strongly positively auto-correlated ($\rho>0.8$), where it explodes. But once again, such huge auto-correlation doesn't seem very realistic.
\begin{figure}[h]
	\begin{minipage}[c]{.54\linewidth}
		\includegraphics[scale=0.6]{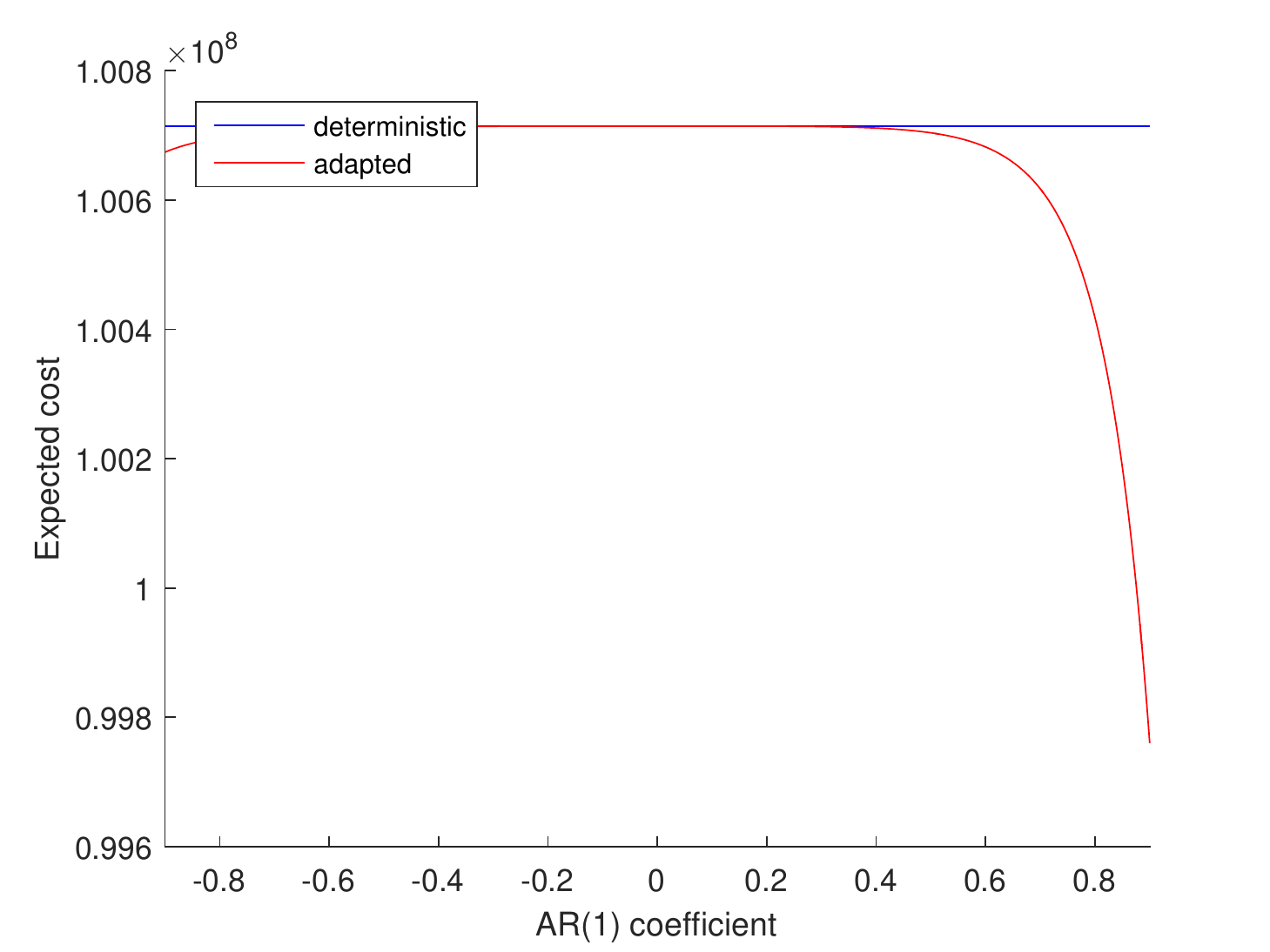}
	\end{minipage}
	\begin{minipage}[c]{.46\linewidth}
		\includegraphics[scale=0.6]{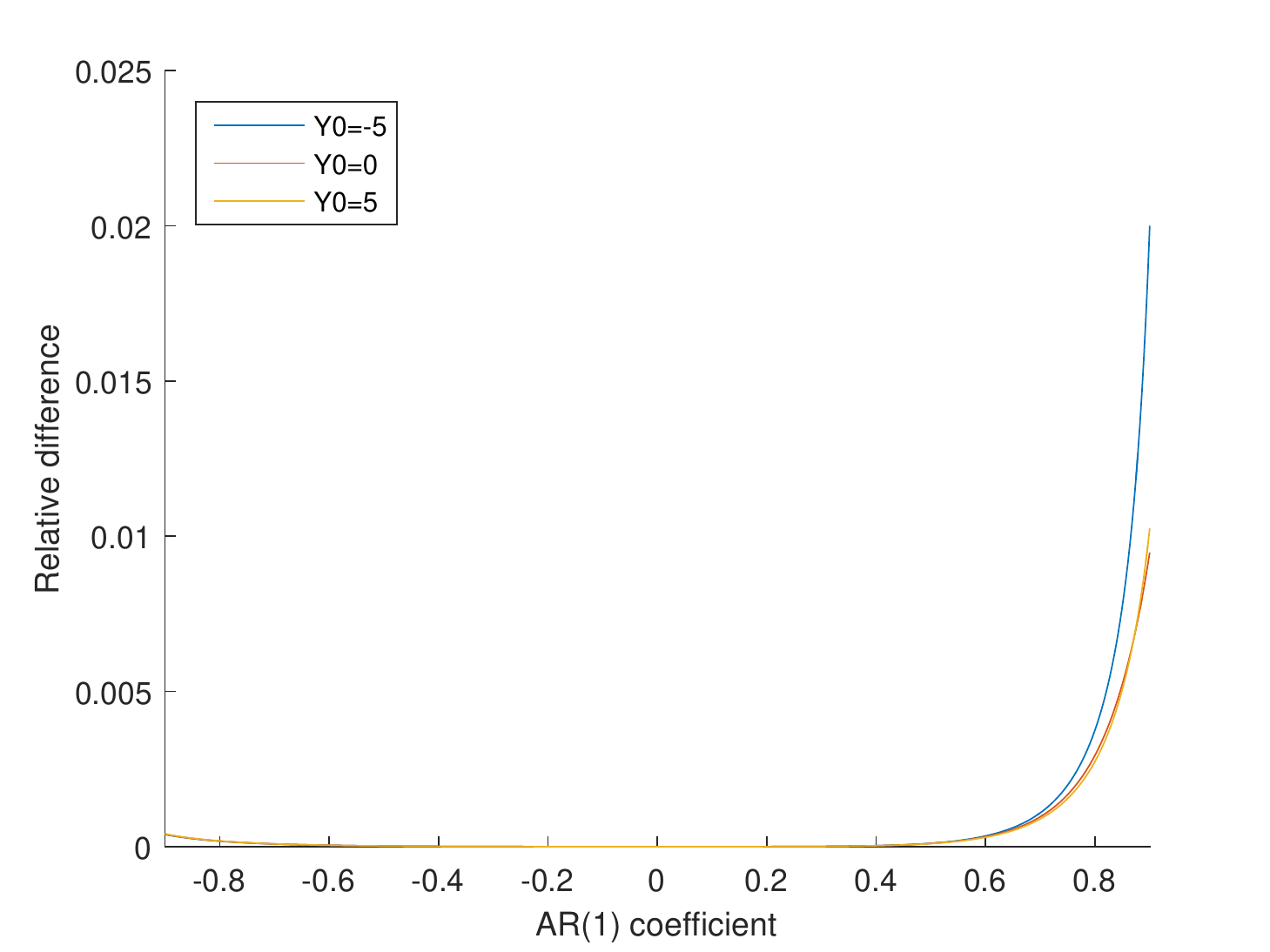}
	\end{minipage}
\caption{Influence of $\rho$ on the expected costs and relative difference}
\label{tempInfR}
\end{figure}

As concerns the influence of $\gamma$, we have the following. 
\begin{figure}[H]
	\begin{minipage}[c]{.54\linewidth}
		\includegraphics[scale=0.6]{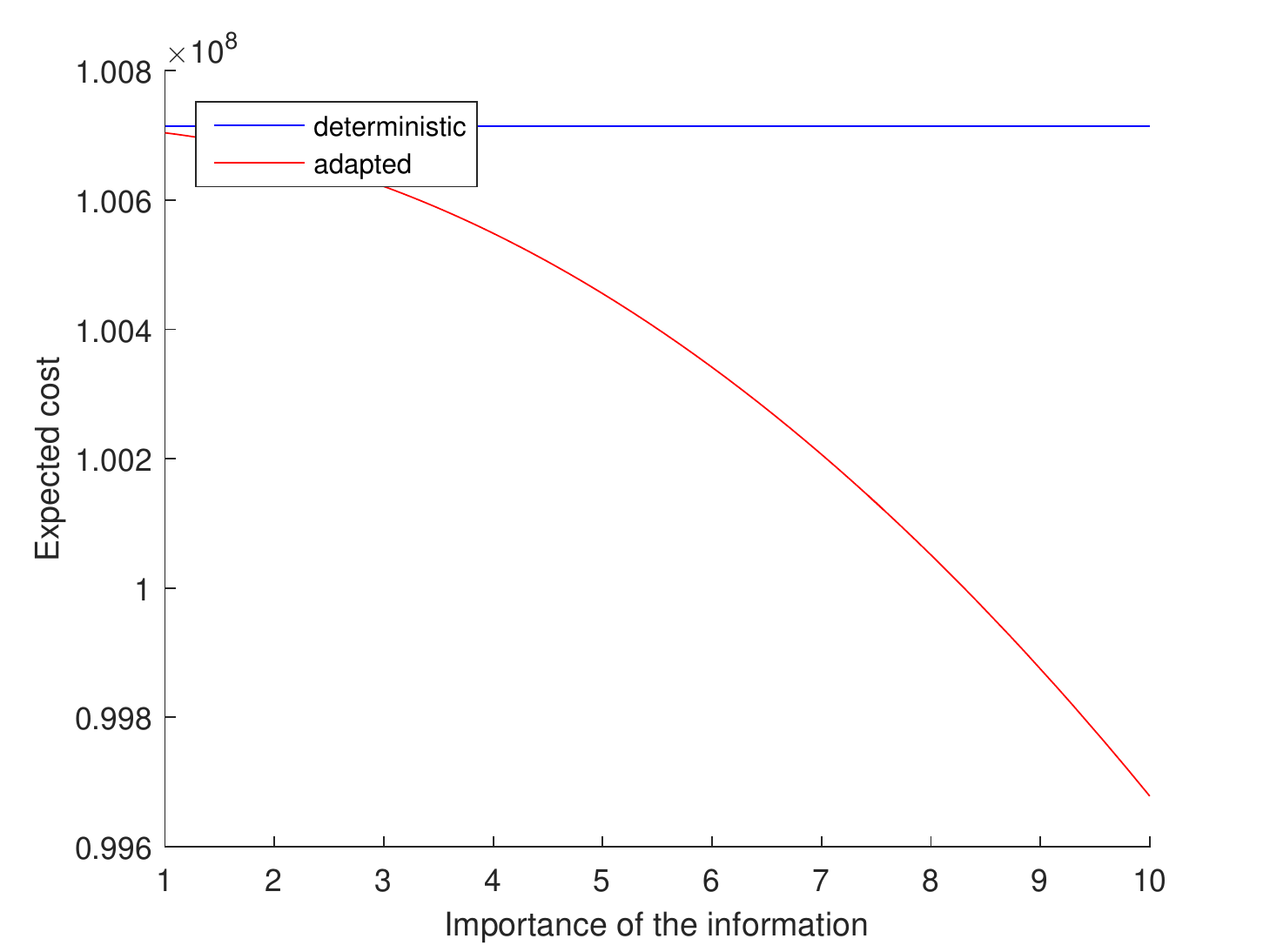}
	\end{minipage}
	\begin{minipage}[c]{.46\linewidth}
		\includegraphics[scale=0.6]{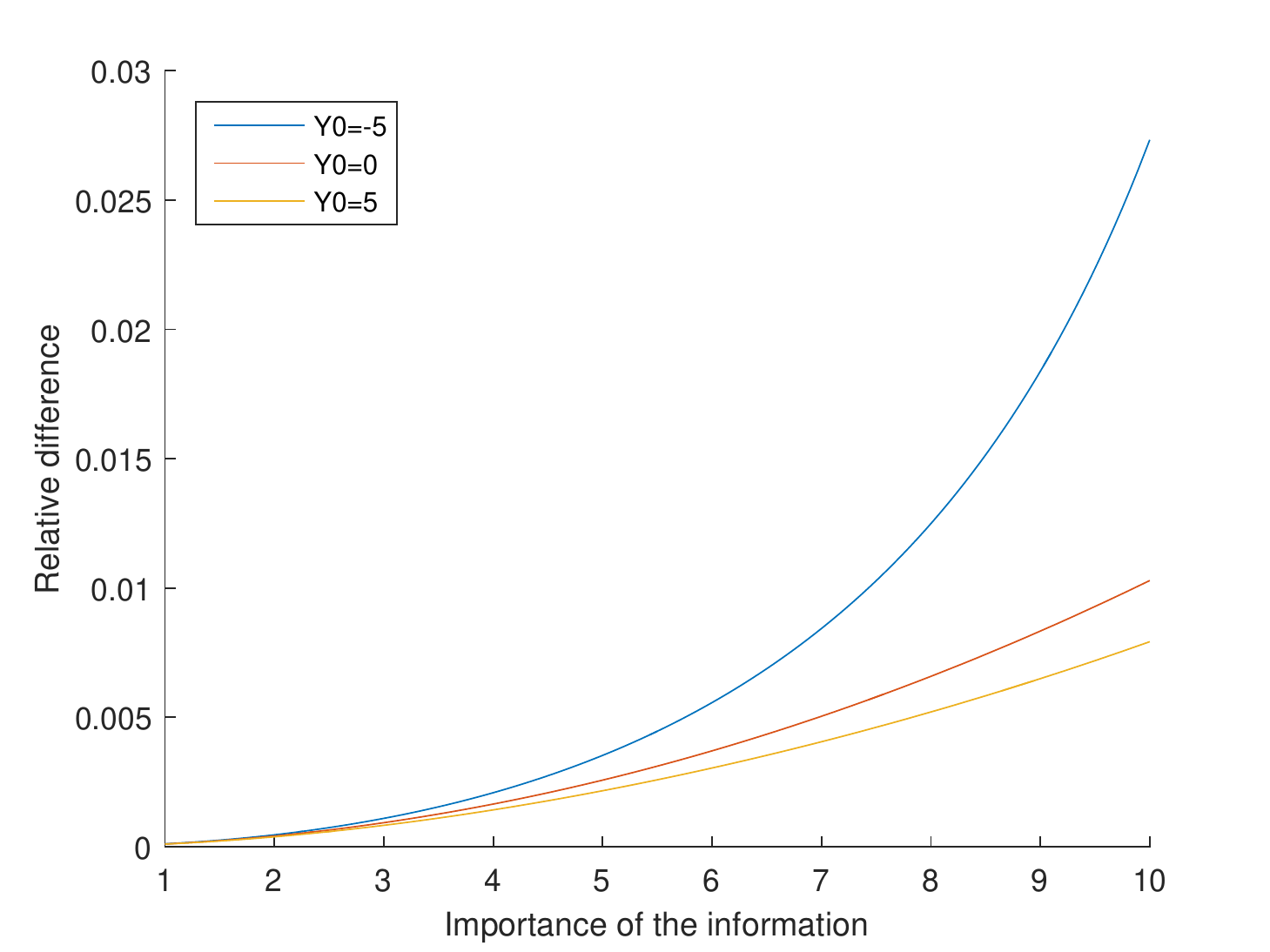}
	\end{minipage}
\caption{Influence of $\gamma$ on the expected costs and relative difference}
\label{tempInfG}
\end{figure}
Figure \ref{tempInfG} shows the evolution of the expected costs and the relative difference when $\gamma$ varies from $1$ to $10$.
The relative difference grows with $\gamma$ the same way as in the case of a permanent impact.

Finally, we consider the influence of $\sigma_Y$.
\begin{figure}[H]
	\begin{minipage}[c]{.54\linewidth}
		\includegraphics[scale=0.6]{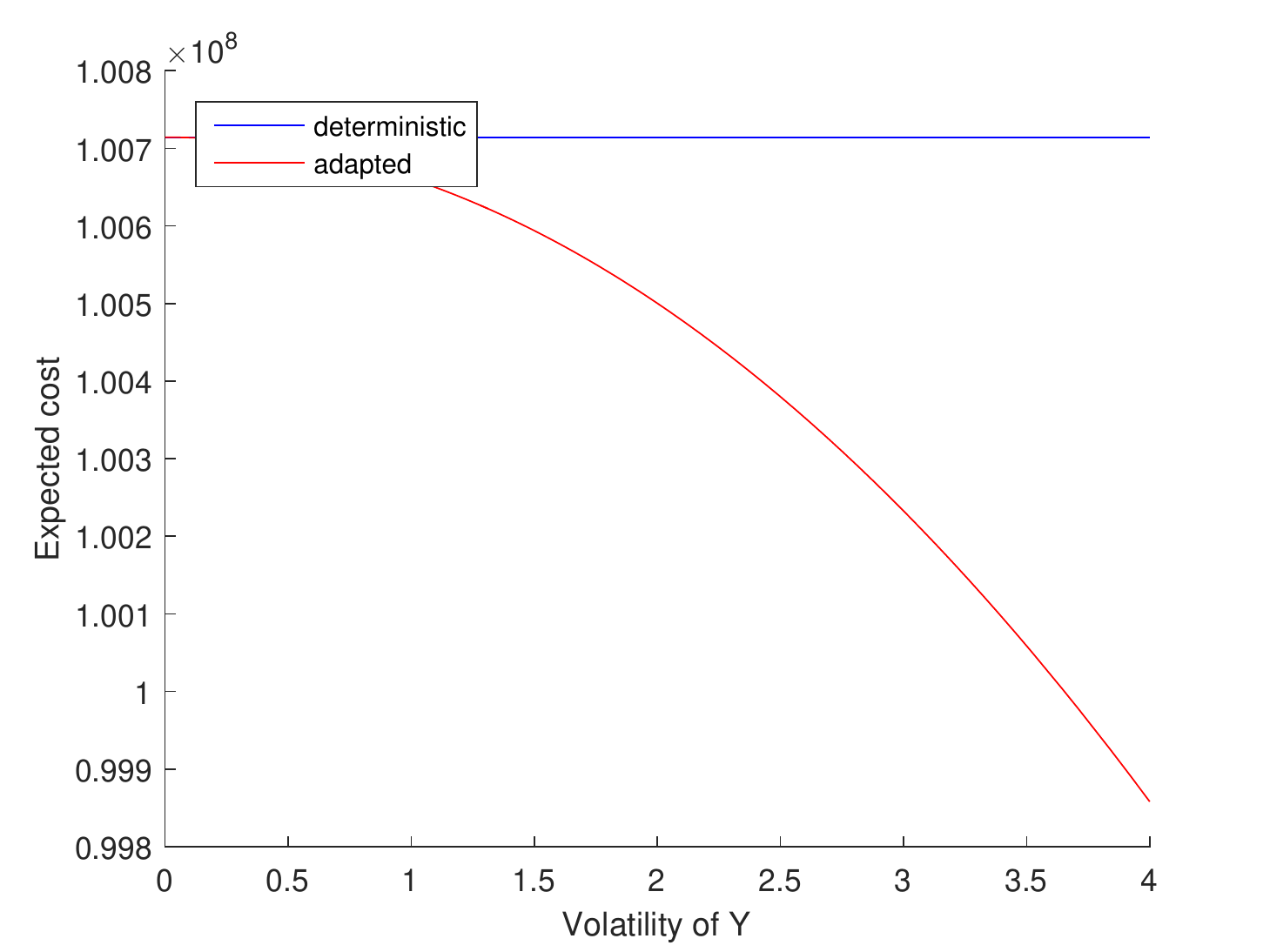}
	\end{minipage}
	\begin{minipage}[c]{.46\linewidth}
		\includegraphics[scale=0.6]{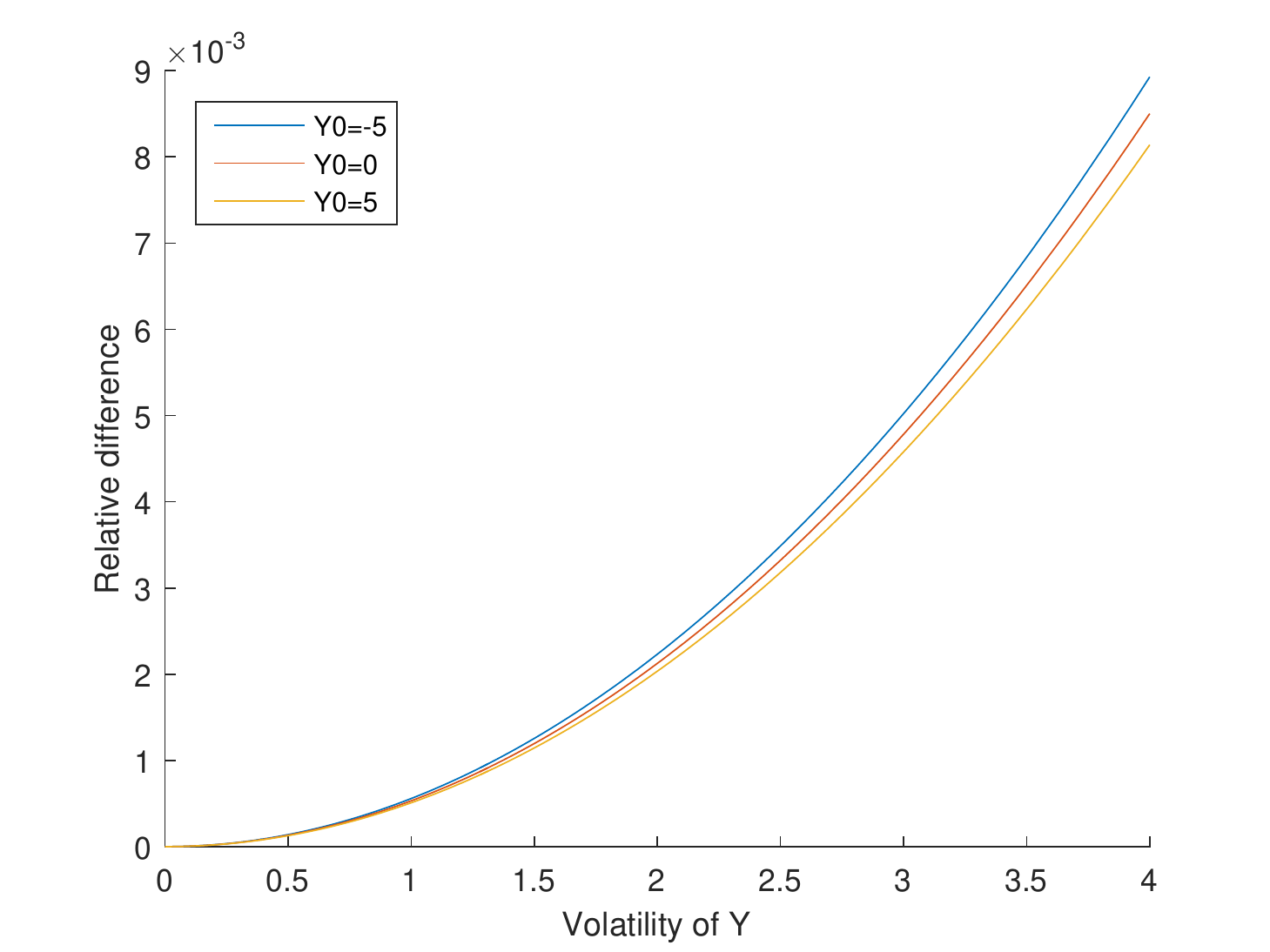}
	\end{minipage}
\caption{Influence of $\sigma_Y$ on the expected costs and relative difference}
\label{tempInfS}
\end{figure}
Figure \ref{tempInfS} shows the evolution of the expected costs and the relative difference when $\sigma_Y$ varies from $0$ to $4$.
The relative difference grows with $\sigma_Y$ the same way as in the case of a permanent impact.

\medskip

This concludes our analysis of the discrete time case. We now move to the continous time case.
%
%
%
%
%
%
%
%

\section{Continuous time trading with risk function}\label{sec:cont}

\subsection{Model formulation with cost and risk based criterion}

In this section we will recall the framework used by Gatheral and Shied \cite{gatheral2011}, with slightly modified notations. 
Let $x_t$ be the stochastic process for the number of units left to be executed at time $t$, such that $x_0 = X$ and $x_T = 0$. In the static case $x$ will be a deterministic function of time. We assume $t \mapsto x_t$ to have absolutely continuous paths and to be adapted. The unaffected price $\widetilde{S}$, namely the unaffected price one would observe in the market without our trades, is assumed to follow a geometric Brownian motion (GBM).
Hence the unaffected and impacted/affected asset mid-prices are respectively given by 
\begin{align}
d \widetilde{S}_t &= \sigma \widetilde{S}_t d W_{t},\quad \widetilde{S}_0 = S_0,\\
S_t &= \widetilde{S}_t + \eta  \dot{x}_{t} + \gamma (x_t - x_0),
\end{align}
where the volatility $\sigma$, the temporary impact parameter $\eta$ and the permanent impact parameter $\gamma$ are positive constants and $W$ is a standard Brownian motion.

The term $\eta \dot{x}_t$ is the temporary impact. As in the discrete time case, it only affects the current execution. The term $\gamma(x_t - x_0)$ is the permanent impact. As in the discrete time case, it has a permanent effect on the price. Indeed, the effect is proportional to the total amount of shares executed up to the current time.

\begin{remark}
Since the unaffected price is a GBM, it can not become negative. This is an improvement compared to the ABM of Bertsimas and Lo. However, we have seen in the examples given in \cite{brigodigraziano}, where a displaced diffusion is also considered, that this may not make a big difference in practice.  
\end{remark}

In this setting we will consider a sell order, which means that $x_t$ is the amount of shares left to be sold at time $t$. At time $t$, we instantly sell a quantity $-\dot{x}_tdt$ at price $S_t$. Hence the total execution cost associated with the strategy $x_t$ is  \\
\begin{eqnarray*} C(x) :=  \int_0^T S_t  \dot{x}_t dt =    \int_0^T \left[\widetilde{S}_t + \eta\ \dot{x}_t + \gamma (x_t - x_0)\right]\dot{x}_t dt \\
= -X S_0 - \int_0^T x_t d\widetilde{S}_t + \eta \int_0^T \dot{x}_t^2 dt + \frac{\gamma}{2} X^2.
\end{eqnarray*}

The problem is to minimize an objective function that consists in both the expected cost and a risk criterion.

The risk term chosen by Gatheral and Shied is 
\begin{equation*}
\mathbb{E}_0 \left[ \widetilde{\lambda} \int_{0}^{T} x_t \widehat{S}_t dt \right], 
\end{equation*}
where $\widehat{S}_t = \widetilde{S}_t + \gamma x_t$ and the risk aversion parameter $\widetilde{\lambda}$ is a positive constant. We choose to use $\widehat{S}$ instead of $\widetilde{S}$ because we want to take into account the effect of the permanent impact on the mid-price. Gatheral and Schied also consider the simpler case where ${\widetilde{S}}_t$ enters the risk criterion, instead of $\widehat{S}_t$, see also \cite{brigodigraziano} for the displaced diffusion case. 

The objective function to minimize is then 
\begin{equation}\label{objCont}
\mathbb{E}_0[C(x)] + \widetilde{\lambda}\mathbb{E}_0 \left[ \int_{0}^{T} x_t \hat{S}_t d t \right]= -S_0 X + \frac{\gamma}{2} X^2 +  \mathbb{E}_0 \left[ \eta  \int_{0}^{T} \dot{x}_{t}^2 d t + \widetilde{\lambda}\int_{0}^{T} x_t\hat{S}_t d t\right].
\end{equation}

We can simplify the problem easily by taking out the constants. Setting $\lambda =  \widetilde{\lambda} / \eta$, we now consider the problem \\
\begin{equation}\label{probCont}
\min_{x} \mathbb{E}_0 \left[\int_{0}^{T} (\dot{x}_{t}^2 + \lambda x_t\hat{S}_t) d t\right].
\end{equation} 

%

\subsection{Optimal adapted solution under temporary and permanent impact}

We will briefly recall the general (adapted) solutions of problem \eqref{probCont} since they have already been obtained by Gatheral and Shied \cite[Theorem 3.2, page 9]{gatheral2011}. Let $\kappa:=\sqrt{\lambda \gamma}$.
\begin{theorem}[Optimal execution strategy]
The unique optimal strategy is \\
\begin{equation}\label{stratAd}
x_t^* = \sinh(\kappa(T-t))\left(\frac{X}{\sinh(\kappa T)}-\frac{\lambda}{2\kappa}\int_0^t \frac{\widetilde{S}_s}{1+\cosh(\kappa(T-s))}ds\right).
\end{equation}
\end{theorem}

\begin{proposition}[Value of the minimization problem]
The value of the minimization problem is 
\begin{equation}\label{minValAd}
\mathbb{E}_0 \left[\int_{0}^{T} ((\dot{x}_{t}^*)^2 + \lambda x_t^*\hat{S}_t^*) d t\right] = \kappa X^2 \coth(\kappa T)+\frac{\lambda X S_0}{\kappa}\tanh\left(\frac{\kappa T}{2}\right)-\frac{\lambda^2S_0^2e^{\sigma^2 T}}{4\kappa^2}\int_0^T\tanh^2\left(\frac{\kappa t}{2}\right)e^{-\sigma^2 t} d t.
\end{equation}
\end{proposition}

\subsection{Optimal static solution under temporary and permanent impact}

We will now solve problem \eqref{probCont} restricted to the set of deterministic strategies. 
\begin{theorem}[Optimal deterministic execution strategy]
The optimal deterministic strategy is 
\begin{equation}\label{stratDet}
x_t^* = \frac{\sinh(\kappa (T-t))}{\sinh(\kappa T)}X + \frac{\sinh(\kappa (T-t)) + \sinh(\kappa t) - \sinh(\kappa T)}{\sinh(\kappa T)} \frac{S_0}{2 \gamma}.
\end{equation} 
\end{theorem}
\begin{proof}
To solve problem \eqref{probCont}, we will assume that the strategy $x$ is fully known at time $0$.
The function we want to minimize is 
\begin{align*}
\mathbb{E}_0 \left[\int_{0}^{T} (\dot{x}_{t}^2 + \lambda x_t\hat{S}_t) d t\right] &= \int_{0}^{T} (\dot{x}_{t}^2 + \lambda x_t\mathbb{E}_0[\hat{S}_t]) d t \quad \text{since $x_t$ is deterministic} \\
&= \int_{0}^{T} (\dot{x}_{t}^2 + \lambda x_t(S_0 + \gamma x_t)) d t.
\end{align*}
To find the optimal strategy $x^*$ that minimizes this function, we consider the standard perturbations of the processes $x$ and $\dot x$ (see for example \cite{digraziano2016}): 
\begin{eqnarray*}
x^{\epsilon}_t &=& x(t) + \epsilon h_t,\\
\dot x^{\epsilon}_t &=& \dot x_t + \epsilon \dot h_t,
\end{eqnarray*}
where the perturbation process $h$ is an arbitrary function satisfying $h_0=h_T=0$ and $\epsilon$ is a constant. 
Substituting the perturbed path into the previous formula we obtain \\
\begin{equation*}
H(\epsilon) = \int_{0}^{T} (\dot{x}_{t} + \epsilon \dot{h}_{t} )^2 + \lambda (x_{t} + \epsilon h_{t}) \left( S_0 + \gamma (x_{t} + \epsilon h_{t} )\right) d t.
\end{equation*}
The first derivative of $H$ with respect to $\epsilon$ is \\
\begin{equation*}
H'(\epsilon) = \int_{0}^{T} 2 \dot{h}_{t} (\dot{x}_{t} + \epsilon \dot{h}_{t} ) 
+ \lambda (x_{t} + \epsilon h_{t}) \left(\gamma h_{t}\right)
+ \lambda h_{t} \left( S_0  + \gamma (x_{t} + \epsilon h_{t})\right) d t.
\end{equation*}
Evaluating the previous expression at $\epsilon=0$ gives 
\begin{align*}
H'(0) &= \int_{0}^{T} 2 \dot{h}_{t} \dot{x}_{t}
+ \lambda x_{t} \gamma h_{t}
+ \lambda h_{t} \left( S_0  + \gamma x_t \right) d t \\
&= 2\left(h_T \dot{x}_T - h_0 \dot{x}_0 \right)-\int_{0}^{T} 2 h_{t} \ddot{x}_{t} d t
+ \int_{0}^{T} \lambda h_{t} \left( 2\gamma x_t + S_0  \right) d t \\
&= \int_{0}^{T} h_{t} \left(-2\ddot{x}_{t} + 2\lambda\gamma x_t + \lambda S_0 \right) d t.
\end{align*} 
The optimal path is obtained by setting $H'(0)=0$. Since $h$ is an arbitrary function, the following differential equation must be satisfied for all $t \in [0,T]$: 
\begin{equation}\label{diffCont}
\ddot{x}_t - \kappa^2 x_t = \frac{\lambda S_0 }{2}, 
\end{equation}
where we set $\kappa:=\sqrt{\lambda \gamma}$ as in the adapted case.

Since $\lambda$ is positive (the rational trader is risk-averse) and $\gamma$ is positive (the market reacts against our execution), the roots of the characteristic equation are real.
Hence the solution of this differential equation  is of the form $A\cosh(\kappa t) + B\sinh(\kappa t) + C$ for some constants $A$, $B$ and $C$.
Substitute in \eqref{diffCont}: 
\[
\kappa^2 A\cosh(\kappa t) + \kappa^2 B\sinh(\kappa t) - \kappa^2 \left(A\cosh(\kappa t) + B\sinh(\kappa t) + C \right) = \frac{\lambda S_0}{2}, \ \ 
C = -\frac{S_0}{2 \gamma}. \]
From the boundary conditions we have: 
\[
x_0 = A + C = X, \ 
A = X + \frac{S_0}{2 \gamma} \]
and 
\begin{eqnarray*}
x_T = A\cosh(\kappa T) + B\sinh(\kappa T) + C = 0, \ \
B = \frac{-X\cosh(\kappa T)}{\sinh(\kappa T)}+\frac{S_0}{2\gamma}\frac{(1-\cosh(\kappa T))}{\sinh(\kappa T)}.
\end{eqnarray*}
The solution of \eqref{diffCont} is 
\begin{align*}
x_t^* &= \left( X - C \right)\cosh(\kappa t) - \frac{\left( X - C \right)\cosh(\kappa T) + C}{\sinh(\kappa T)}\sinh(\kappa t) + C \\
&= \left( X - C \right)\left( \frac{\cosh(\kappa t)\sinh(\kappa T) - \cosh(\kappa T)\sinh(\kappa t)}{\sinh(\kappa T)}\right) + C\left( 1 - \frac{\sinh(\kappa t)}{\sinh(\kappa T)} \right).
\end{align*} 
\end{proof}

\begin{remark}
When $\lambda \downarrow 0$ (no risk in criterion), the deterministic strategy tends to a VWAP.
\end{remark}
\begin{theorem}[Value of the minimization problem with the deterministic strategy]
The value of the minimization problem in the deterministic framework is 
\begin{equation}\label{minValDet}
\mathbb{E}_0 \left[\int_{0}^{T} ((\dot{x}_{t}^*)^2 + \lambda x_t^*\hat{S}_t^*) d t\right]=\kappa X^2 \coth(\kappa T)+\frac{\kappa S_0}{\gamma}\left(X+\frac{S_0}{2\gamma}\right)\tanh\left(\frac{\kappa T}{2}\right) - \frac{\lambda T S_0^2}{4\gamma}.
\end{equation}
\end{theorem}
\begin{proof}
The value of the minimization problem obtained when following the deterministic strategy of equation \ref{stratDet} is \\
\begin{align*}
\mathbb{E}_0 &\left[\int_{0}^{T} ((\dot{x}_{t}^*)^2 + \lambda x_t^*\hat{S}_t^*) d t\right] \\
&= \int_0^T \left(\frac{-\kappa\cosh(\kappa (T-t))}{\sinh(\kappa T)}X + \frac{\kappa\cosh(\kappa t)-\kappa\cosh(\kappa (T-t))}{\sinh(\kappa T)} \frac{S_0}{2 \gamma}\right)^2 d t\\
&+\lambda S_0 \int_0^T \left(\frac{\sinh(\kappa (T-t))}{\sinh(\kappa T)}X + \frac{\sinh(\kappa (T-t)) + \sinh(\kappa t) - \sinh(\kappa T)}{\sinh(\kappa T)} \frac{S_0}{2 \gamma}\right) d t\\
&+\kappa^2 \int_0^T \left(\frac{\sinh(\kappa (T-t))}{\sinh(\kappa T)}X + \frac{\sinh(\kappa (T-t)) + \sinh(\kappa t) - \sinh(\kappa T)}{\sinh(\kappa T)} \frac{S_0}{2 \gamma}\right)^2 d t \\
&= \kappa^2\int_0^T \left(\frac{\cosh^2(\kappa (T-t))}{\sinh^2(\kappa T)}X^2 + \frac{\cosh^2(\kappa t)+\cosh^2(\kappa (T-t))-2\cosh(\kappa t)\cosh(\kappa (T-t))}{\sinh^2(\kappa T)} \frac{S_0^2}{4 \gamma^2}\right) d t \\
&+2\kappa^2\int_0^T\left(\frac{\cosh^2(\kappa (T-t))-\cosh(\kappa (T-t))\cosh(\kappa t)}{\sinh^2(\kappa T)} \frac{S_0X}{2 \gamma}\right) d t\\
&+\kappa^2 \int_0^T \left(\frac{\sinh(\kappa (T-t))}{\sinh(\kappa T)}\frac{S_0X}{\gamma} + \frac{\sinh(\kappa (T-t)) + \sinh(\kappa t) - \sinh(\kappa T)}{\sinh(\kappa T)} \frac{S_0^2}{2 \gamma^2}\right) d t\\
&+\kappa^2 \int_0^T \left(\frac{\sinh^2(\kappa (T-t))}{\sinh^2(\kappa T)}X^2 + \frac{(\sinh(\kappa (T-t)) + \sinh(\kappa t) - \sinh(\kappa T))^2}{\sinh^2(\kappa T)} \frac{S_0^2}{4 \gamma^2}\right) d t\\
&+\kappa^2 \int_0^T\left(2\frac{\sinh(\kappa (T-t))}{\sinh(\kappa T)}\frac{\sinh(\kappa (T-t)) + \sinh(\kappa t) - \sinh(\kappa T)}{\sinh(\kappa T)} \frac{S_0X}{2 \gamma}\right) d t \\
&= \kappa^2X^2\int_0^T \frac{\cosh^2(\kappa (T-t))+\sinh^2(\kappa (T-t))}{\sinh^2(\kappa T)} d t \\
&+ \kappa^2\frac{S_0^2}{4 \gamma^2}\int_0^T\frac{\cosh^2(\kappa t)+\cosh^2(\kappa (T-t))-2\cosh(\kappa t)\cosh(\kappa (T-t))}{\sinh^2(\kappa T)}  d t \\
&+ \kappa^2\frac{S_0^2}{4 \gamma^2}\int_0^T\frac{\sinh^2(\kappa (T-t))+\sinh^2(\kappa t)- \sinh^2(\kappa T) + 2\sinh(\kappa(T-t))\sinh(\kappa t) }{\sinh^2(\kappa T)}  d t \\
&+\kappa^2\frac{S_0X}{\gamma}\int_0^T\frac{\cosh^2(\kappa (T-t))-\cosh(\kappa (T-t))\cosh(\kappa t)+\sinh^2(\kappa (T-t))+\sinh(\kappa (T-t))\sinh(\kappa t)}{\sinh^2(\kappa T)}  d t\\
&= \kappa^2X^2\int_0^T \frac{\cosh(2\kappa (T-t))}{\sinh^2(\kappa T)} d t
+ \kappa^2\frac{S_0^2}{4 \gamma^2}\int_0^T\frac{\cosh(2\kappa t)+\cosh(2\kappa (T-t))-2\cosh(\kappa (T-2t))}{\sinh^2(\kappa T)}-1 d t \\
&+\kappa^2\frac{S_0X}{\gamma}\int_0^T\frac{\cosh(2\kappa (T-t))-\cosh(\kappa (T-2t))}{\sinh^2(\kappa T)}  d t\\
&= \kappa^2X^2\frac{\sinh(2\kappa T)}{2\kappa\sinh^2(\kappa T)}
+ \kappa^2\frac{S_0^2}{4 \gamma^2}\frac{2\sinh(2\kappa T)-4\sinh(\kappa T) }{2\kappa\sinh^2(\kappa T)}  -\kappa^2\frac{S_0^2T}{4 \gamma^2}
+\kappa^2\frac{S_0X}{\gamma}\frac{\sinh(2\kappa T)-2\sinh(\kappa T)}{2\kappa\sinh^2(\kappa T)}\\
&= \kappa X^2\frac{\cosh(\kappa T)}{\sinh(\kappa T)}
+ \kappa\frac{S_0^2}{4 \gamma^2}\frac{2\cosh(\kappa T)-2}{\sinh(\kappa T)} -\kappa^2\frac{S_0^2T}{4 \gamma^2}+\kappa\frac{S_0X}{2 \gamma}\frac{2\cosh(\kappa T)-2}{\sinh(\kappa T)}.
\end{align*}
\\
\end{proof}

\subsection{Comparison of optimal static and adapted solutions}

We will now numerically attempt to quantify the differences in the minimum objective function obtained by the deterministic and by the adapted strategies.

Since we operated a linear transformation from \eqref{objCont} to \eqref{probCont}, we will multiply the value of the minimization problems \eqref{minValAd} and \eqref{minValDet} by $\eta$ and add back the term $-S_0 X + \frac{\gamma}{2} X^2$ to obtain the value of the objective functions along the optimal solution. We will denote them respectively $J_{ad}^*$ for the fully adapted case and $J_{det}^*$ for the deterministic/static case.
\begin{corollary}[Minimum of the objective function]
The minimum value of the objective function is
\begin{equation*}
J^*_{ad}(X_0,S_0) = -S_0 X + \frac{\gamma}{2} X^2 + \eta\left(\kappa X^2 \coth(\kappa T)+\frac{\lambda X S_0}{\kappa}\tanh\left(\frac{\kappa T}{2}\right)-\frac{\lambda^2S_0^2e^{\sigma^2 T}}{4\kappa^2}\int_0^T\tanh^2\left(\frac{\kappa t}{2}\right)e^{-\sigma^2 t} d t\right),
\end{equation*}
and the value of the objective function obtained when using the optimal deterministic strategy is
\begin{equation*}
J^*_{det}(X_0,S_0) = -S_0 X + \frac{\gamma}{2} X^2 + \eta\left(\kappa X^2 \coth(\kappa T)+\frac{\kappa S_0}{\gamma}\left(X+\frac{S_0}{2\gamma}\right)\tanh\left(\frac{\kappa T}{2}\right) - \frac{\lambda T S_0^2}{4\gamma}\right).\\
\end{equation*}
\end{corollary}

Similarly to the cases with no risk criterion, we define the absolute and relative differences.
\begin{definition}[Absolute difference]
\begin{align*}
\epsilon_{abs}&:= J_{det}^*(X_0,S_0)-J_{ad}^*(X_0,S_0) \\
&=\frac{\kappa S_0^2}{2\gamma^2}\tanh\left(\frac{\kappa T}{2}\right) - \frac{\lambda T S_0^2}{4\gamma}
+ \frac{\lambda^2S_0^2e^{\sigma^2 T}}{4\kappa^2}\int_0^T\tanh^2\left(\frac{\kappa t}{2}\right)e^{-\sigma^2 t} d t.
\end{align*}
\end{definition}
\begin{proposition}\label{sig0}
Both strategies have the same expected cost when there is no randomness.
Hence deciding the strategy entirely before the execution is equivalent to assuming that there is no randomness in the price movements, as in the discrete setting studied in the previous section.
\end{proposition}
\begin{proof}
For $\sigma=0$, $\epsilon_{abs}$ becomes \\
\begin{align*}
\epsilon_{abs}&=\frac{\kappa S_0^2}{2\gamma^2}\tanh\left(\frac{\kappa T}{2}\right) - \frac{\lambda T S_0^2}{4\gamma}
+ \frac{\lambda^2S_0^2}{4\kappa^2}\int_0^T\tanh^2\left(\frac{\kappa t}{2}\right) d t \\
&= \frac{\kappa S_0^2}{2\gamma^2}\tanh\left(\frac{\kappa T}{2}\right) - \frac{\lambda T S_0^2}{4\gamma}+\frac{\lambda^2S_0^2}{4\kappa^2}\left(T-\frac{2}{\kappa}\tanh\left(\frac{\kappa T}{2}\right) \right) \\
&= 0.
\end{align*}
\end{proof}
\begin{proposition}[Sign of the absolute difference]
As expected, the adapted strategy is always better than the deterministic one, in that it results in a criterion that is smaller or equal to the deterministic one.
\end{proposition}
\begin{proof}
Consider the absolute difference as a function of $\sigma$.
\begin{equation*}
\epsilon_{abs}(\sigma) = \frac{\kappa S_0^2}{2\gamma^2}\tanh\left(\frac{\kappa T}{2}\right) - \frac{\lambda T S_0^2}{4\gamma}
+ \frac{\lambda^2S_0^2e^{\sigma^2 T}}{4\kappa^2}\int_0^T\tanh^2\left(\frac{\kappa t}{2}\right)e^{-\sigma^2 t} d t.
\end{equation*}
Let us compute the derivative of $\epsilon_{abs}$ with respect to $\sigma$.
\begin{equation*}
\epsilon'_{abs}(\sigma) = \frac{\lambda^2S_0^2}{4\kappa^2}\int_0^T 2(T-t)\sigma\tanh^2\left(\frac{\kappa t}{2}\right)e^{\sigma^2(T-t)} d t.
\end{equation*}
Since every term in the expression above is positive for $\sigma>0$, by integration and multiplication $\epsilon'_{abs}$ is always positive so $\epsilon_{abs}$ is an increasing function of $\sigma$ on $[0,\infty)$. From Proposition \ref{sig0} we know that $\epsilon_{abs}(0)=0$. Hence $\epsilon_{abs}$ is never negative.
\end{proof}
\begin{definition}[Relative difference]
\begin{equation*}
\epsilon_{rel} :=\frac{\epsilon_{abs}}{|J_{det}^*(X_0,S_0)|}.
\end{equation*}
\end{definition}

For the numerical applications we will consider a single stock with current price $S_0=100$, making the use of percentage volatility intuitive. We want to sell $X=10^6$ shares in $T=1$ day. The stock has a percentage daily volatility $\sigma=1.89\%$, as in the discrete-time cases. $\gamma=2 \times 10^{-6}$ is chosen such that the permanent impact is around $10\%$, assuming there is no risk aversion. The temporary market impact parameter $\eta=2 \times 10^{-6}$ is chosen such that the impact of an instantaneous execution is $2\$$ per share. The risk aversion factor $\widetilde{\lambda}=0.05$ is taken so that the risk term in the objective function is of the same order as the market impacts.

The values described above are summarized in Table \ref{benchCont}.
\begin{table}[!h]
\centering
    \begin{tabular}{| l | l |}
    \hline
   $X$ & $10^6$ \\ \hline
   $S_0$ & $100$ \\ \hline
   $T$ & $1$ \\ \hline
   $\sigma$& $1.89\%$   \\ \hline
   $\gamma$ & $2 \times 10^{-6}$  \\ \hline
   $\eta$ & $2 \times 10^{-6}$   \\ \hline
   $\widetilde{\lambda}$ & $0.05$   \\ \hline
    \end{tabular}
\caption{Benchmark parameter values}
\label{benchCont}
\end{table}

\begin{remark}
Since this is a sell order, the expected costs should be negative (assuming the trader has no incentive to sell at a loss).
\end{remark}

To get an idea of the influence of the risk aversion factor on the strategies, we give a few examples of paths obtained with different values of $\widetilde{\lambda}$ in Figures \ref{contSim}, \ref{contSimL10} and \ref{contSimL0}.
\begin{figure}[H]
\centering
\includegraphics[scale=0.8]{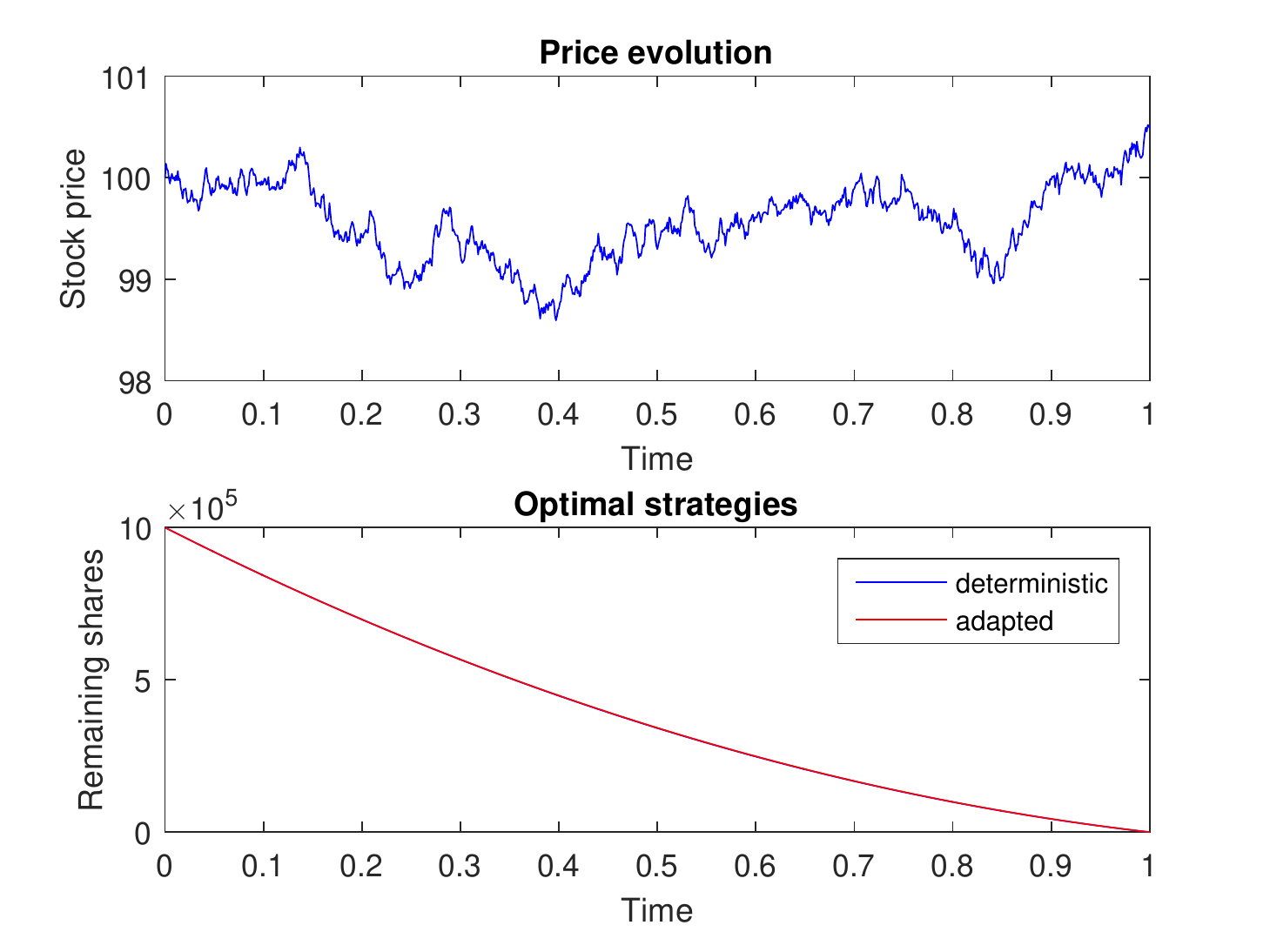}
\caption{One path of a simulated strategy with benchmark parameters ($\widetilde{\lambda}=0.05$)}
\label{contSim}
\end{figure}
\begin{figure}[H]
\centering
\includegraphics[scale=0.8]{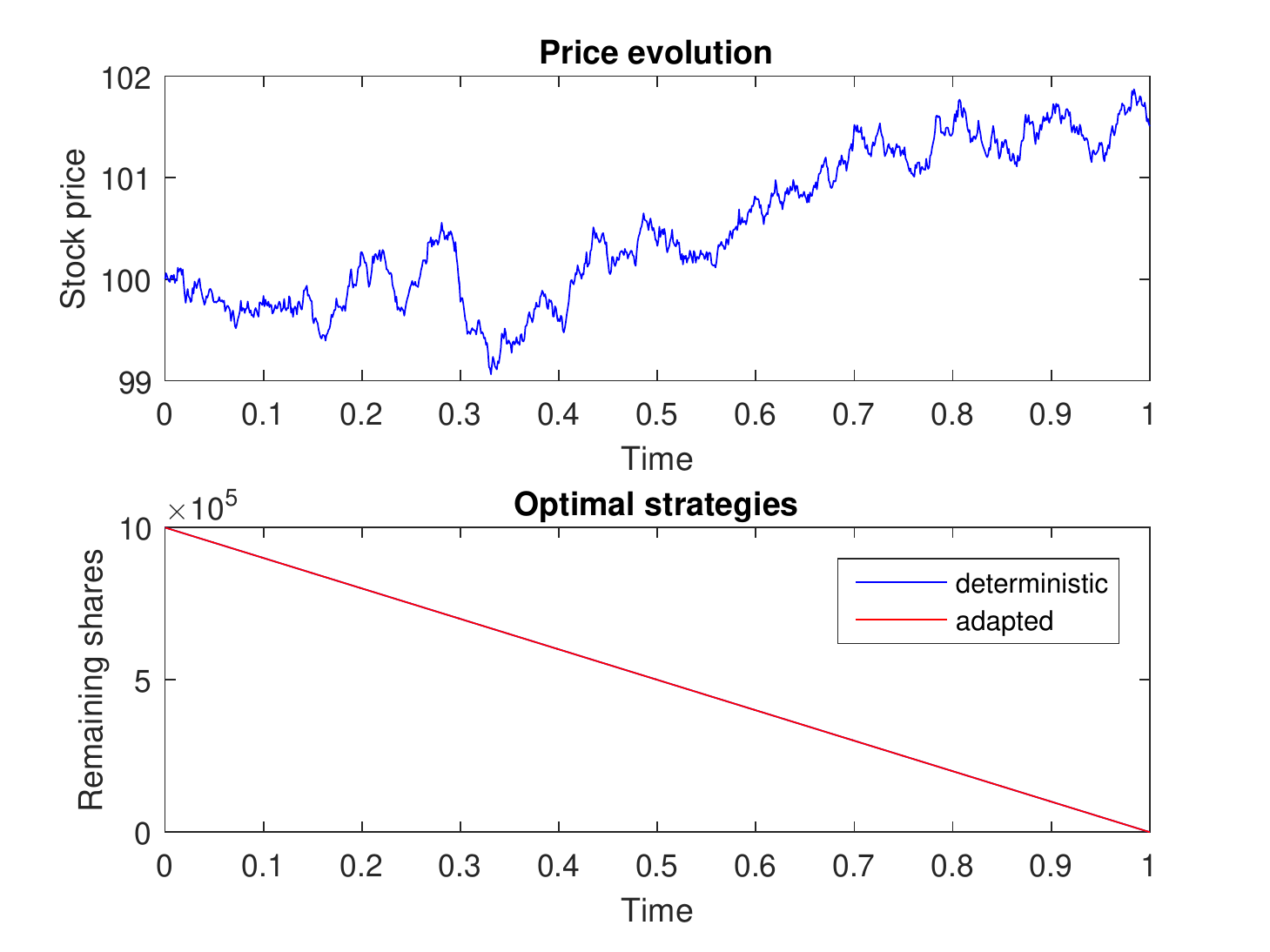}
\caption{One path of a simulated strategy with small risk aversion ($\widetilde{\lambda}=10^{-10}$)}
\label{contSimL0}
\end{figure}
With the benchmark parameters, we find that $J_{det}^*=-9.4736 \times 10^7$, $J_{ad}^*=-9.4736 \times 10^7$ and $\epsilon_{rel}=2.45 \times 10^{-7}$.

With $\widetilde{\lambda}=10^{-10}$, we find that $J_{det}^*=-9.7000 \times 10^7$, $J_{ad}^*=-9.7000 \times 10^7$ and $\epsilon_{rel}=0$. Both strategies are straight lines, which means that they practically follow a VWAP. This is consistent with the fact that with very small $\lambda$ we are close to not having risk in the criterion, leading to the VWAP solution. 
\begin{figure}[h]
\centering
\includegraphics[scale=0.9]{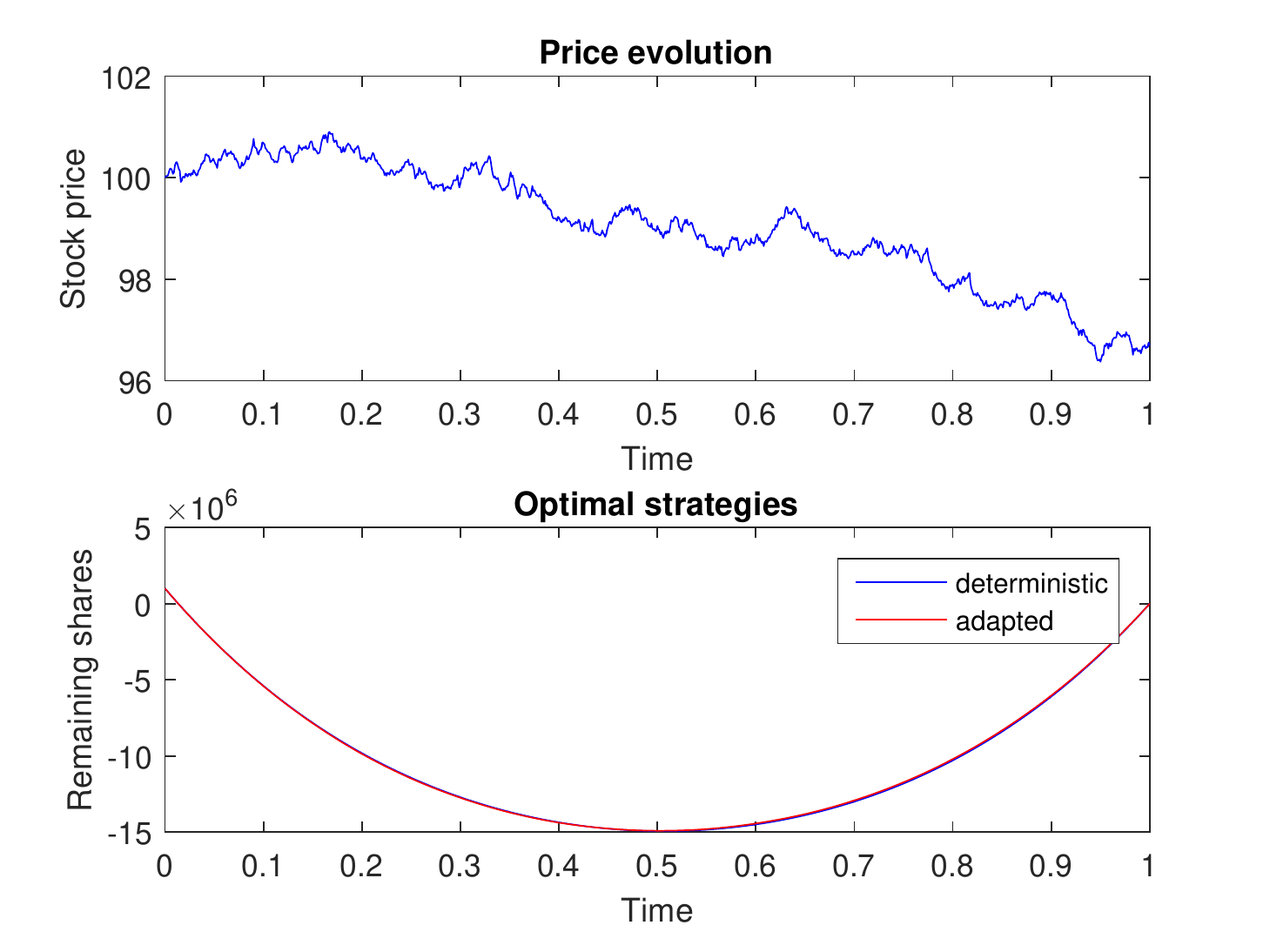}
\caption{One path of a simulated strategy with large risk aversion ($\widetilde{\lambda}=10$)}
\label{contSimL10}
\end{figure}
With $\widetilde{\lambda}=10$, we find that $J_{det}^*=-5.0391 \times 10^9$, $J_{ad}^*=-5.0385 \times 10^9$ and $\epsilon_{rel}=1.14 \times 10^{-4}$.

With $\widetilde{\lambda}=10^3$, we find that $J_{det}^*=-1.1678 \times 10^{12}$, $J_{ad}^*=-1.1680 \times 10^{12}$ and $\epsilon_{rel}=1.68 \times 10^{-4}$.

The last plots are interesting because they illustrate
 the fact that when the risk aversion factor is big, as in Figures \ref{contSimL10} and \ref{contSimL1000}, we tend to execute everything very fast, even exceeding the amounts we are supposed to execute. At the end of the period we buy back what we need to get back to our objective. The larger the risk factor, the steeper the execution. When $\lambda$ is very small, the strategies tend to a VWAP. A reasonable value for $\widetilde{\lambda}$ would be something in-between, as in the slightly curved line of Figure \ref{contSim}. Note however that the risk aversion factor is completely arbitrary, and depends only on the trader so any value of $\widetilde{\lambda}$ is possible.
\begin{figure}[H]
\centering
\includegraphics[scale=0.9]{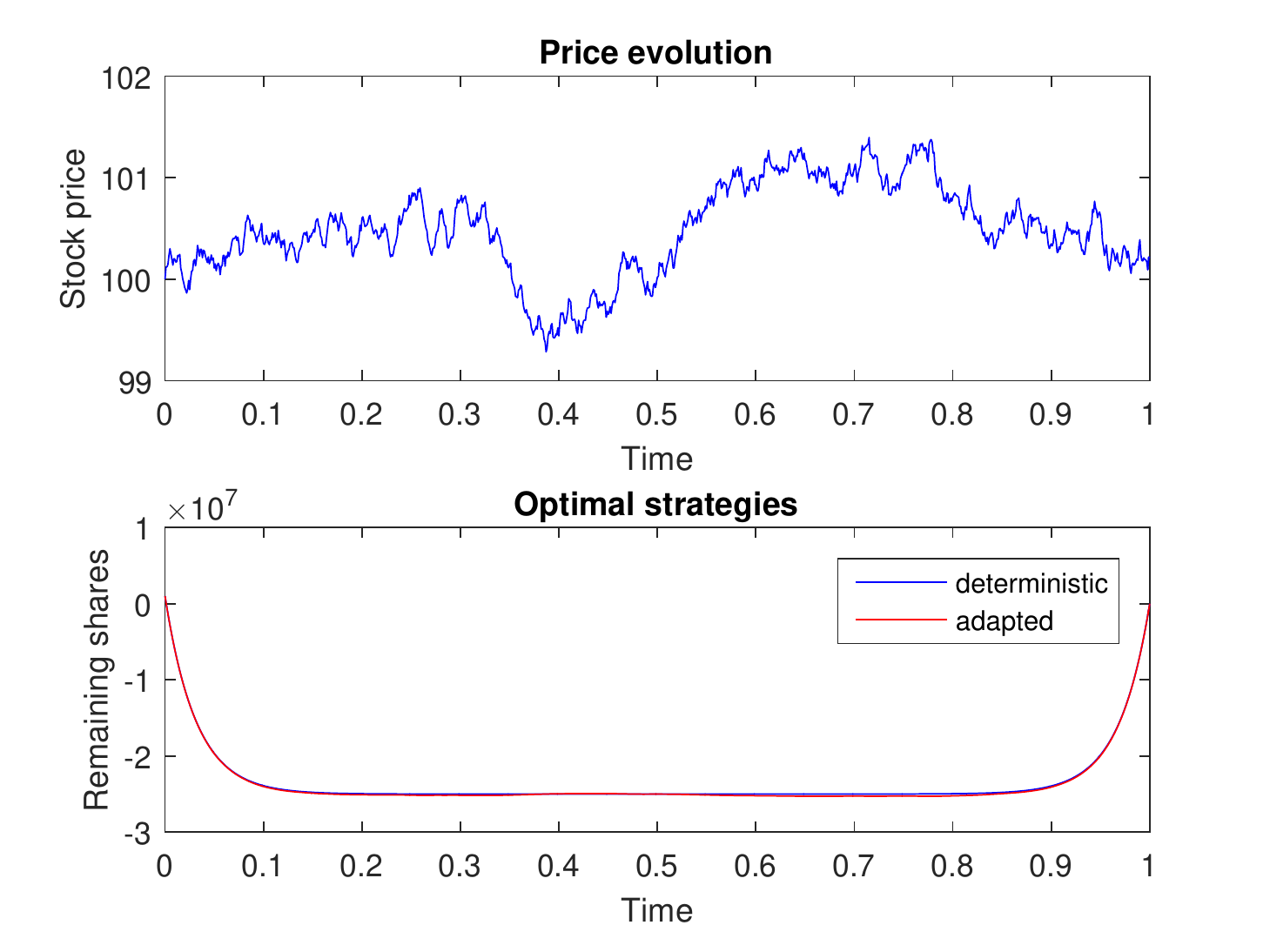}
\caption{One path of a simulated strategy with huge risk aversion ($\widetilde{\lambda}=10^3$)}
\label{contSimL1000}
\end{figure}


To get a more precise idea of the difference between the fully adapted and static optimal strategies, we study the influence of each parameter on the minimized objective functions and their relative difference.
In each numerical example, the parameters will be those of Table \ref{benchCont} except for the one whose influence we study. We will consider parameters and inputs
\[ X, T, \sigma, \gamma, \eta, \widetilde{\lambda} .\] 

We begin with the influence of $X$.
\begin{figure}[h]
	\begin{minipage}[c]{.54\linewidth}
		\includegraphics[scale=0.6]{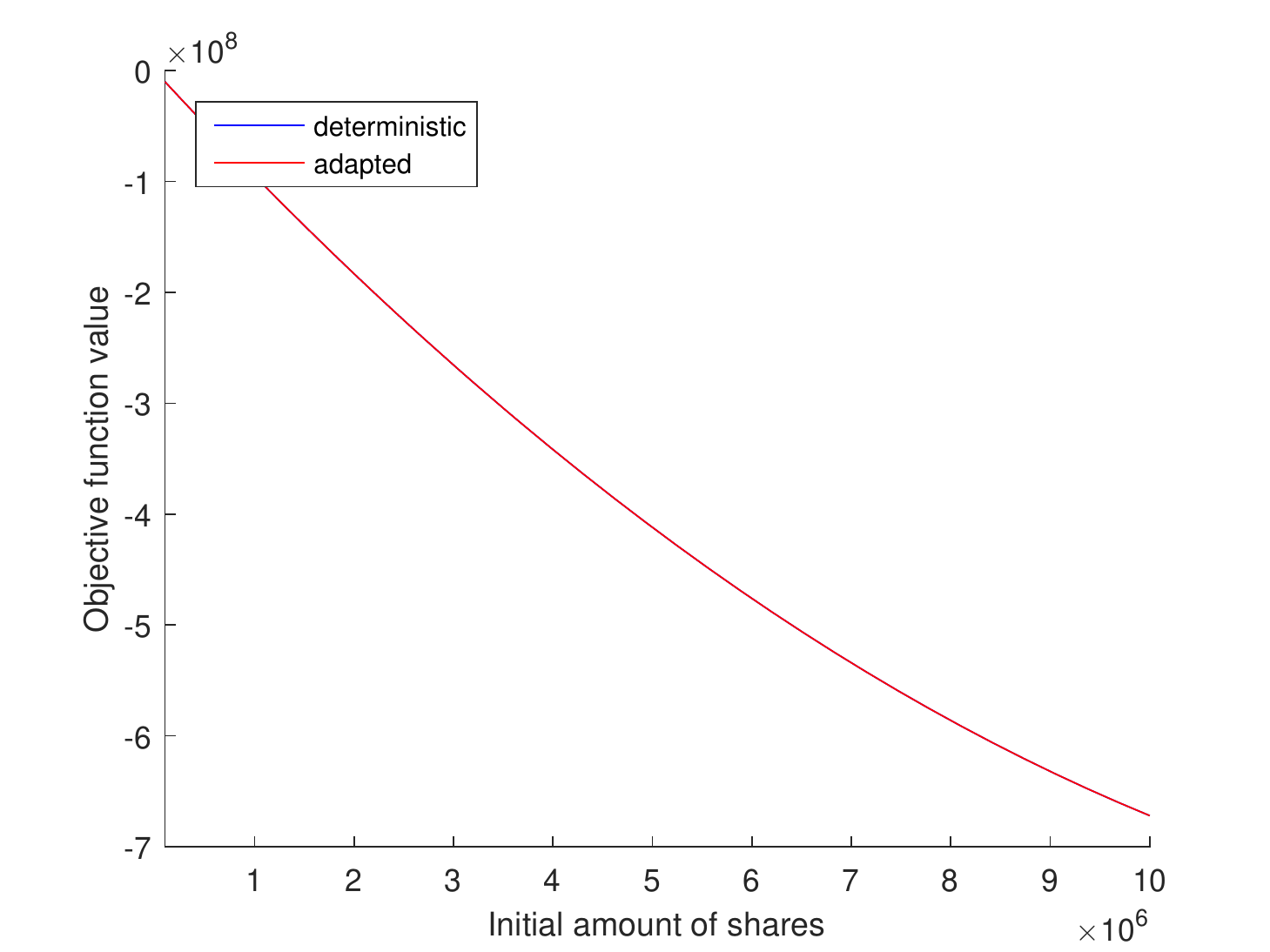}
	\end{minipage}
	\begin{minipage}[c]{.46\linewidth}
		\includegraphics[scale=0.6]{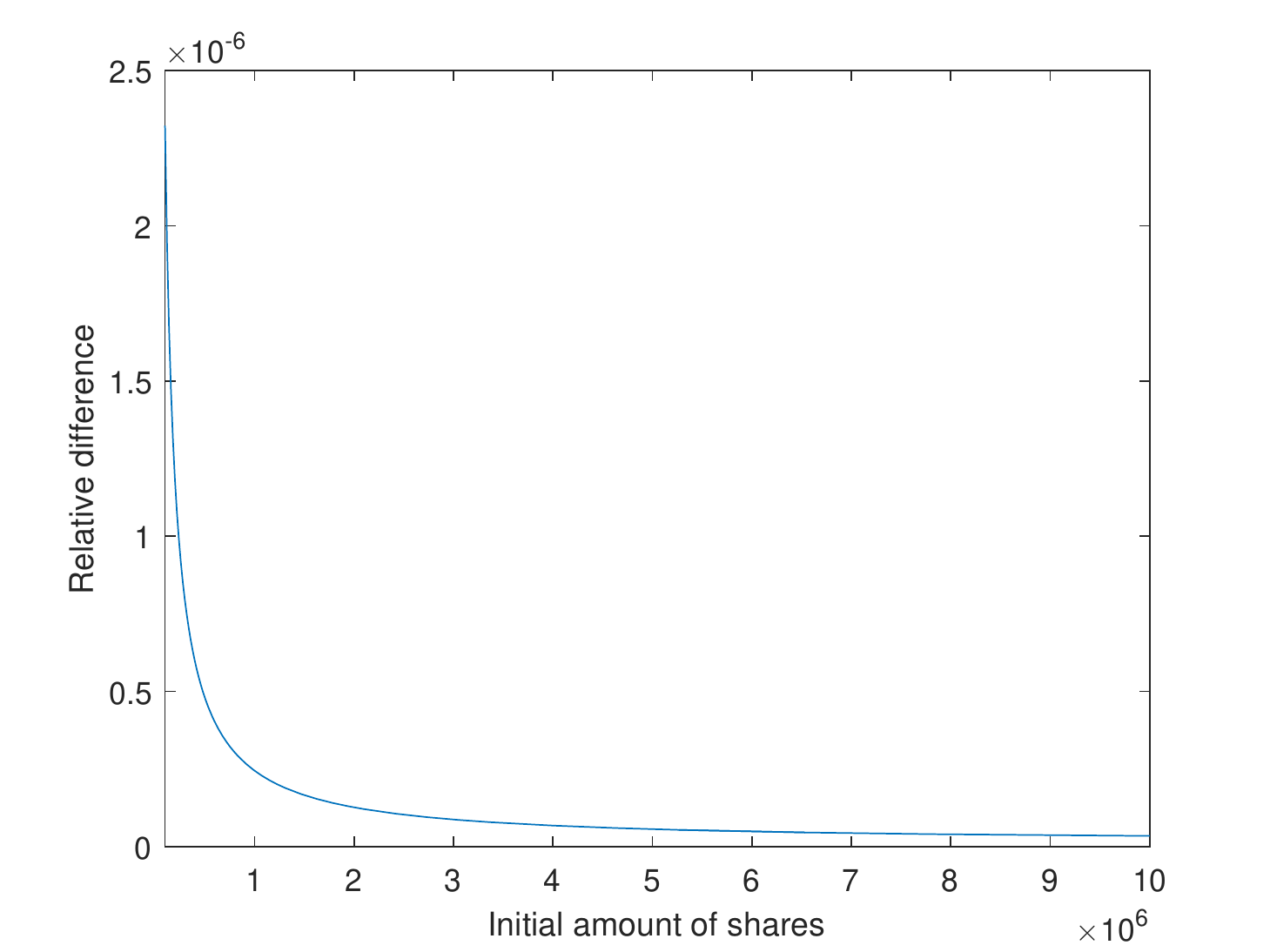}
	\end{minipage}
\caption{Influence of $X$ on the expected costs and relative difference}
\label{contInfX}
\end{figure}
Figure \ref{contInfX} shows the evolution of the expected costs and the relative difference when $X$ varies from $10^5$ to $10^7$.
As in the discrete time case, the relative difference between the two strategies decreases exponentially when the initial amount of shares to execute increases because the expected cost increases with $X$, but not the absolute error. Once again, the market impact parameters $\gamma$ and $\eta$ have been calibrated for a certain $X$, and their influence becomes overwhelming when $X$ is too big. Although the expected costs seem to decrease drastically, one should keep in mind that we are looking at a sell order, so the profit should indeed increase when we sell more shares. In practice, the percentage loss on our profit is bigger as $X$ increases.

We now look at the influence of $T$.
\begin{figure}[H]
	\begin{minipage}[c]{.54\linewidth}
		\includegraphics[scale=0.6]{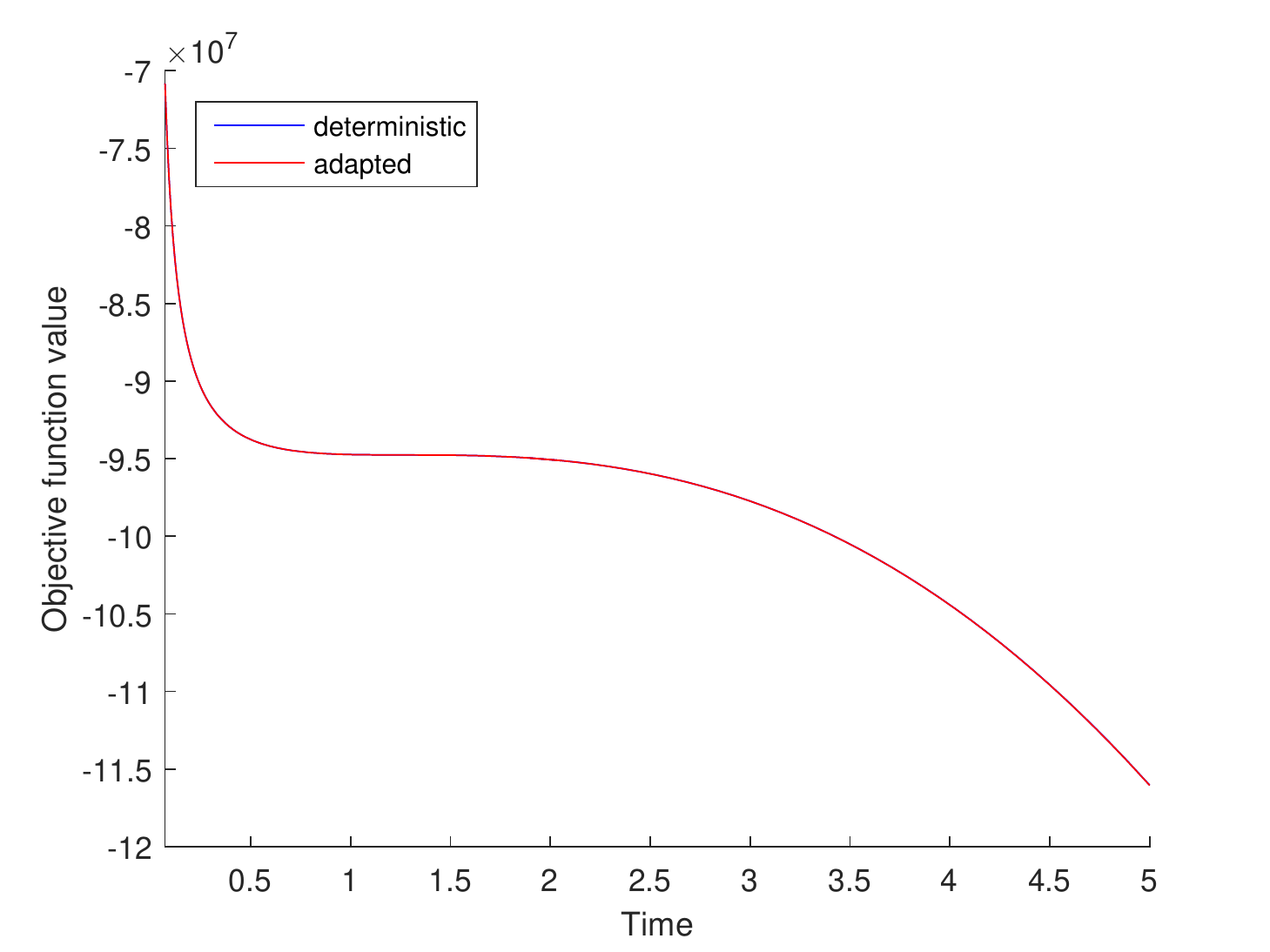}
	\end{minipage}
	\begin{minipage}[c]{.46\linewidth}
		\includegraphics[scale=0.6]{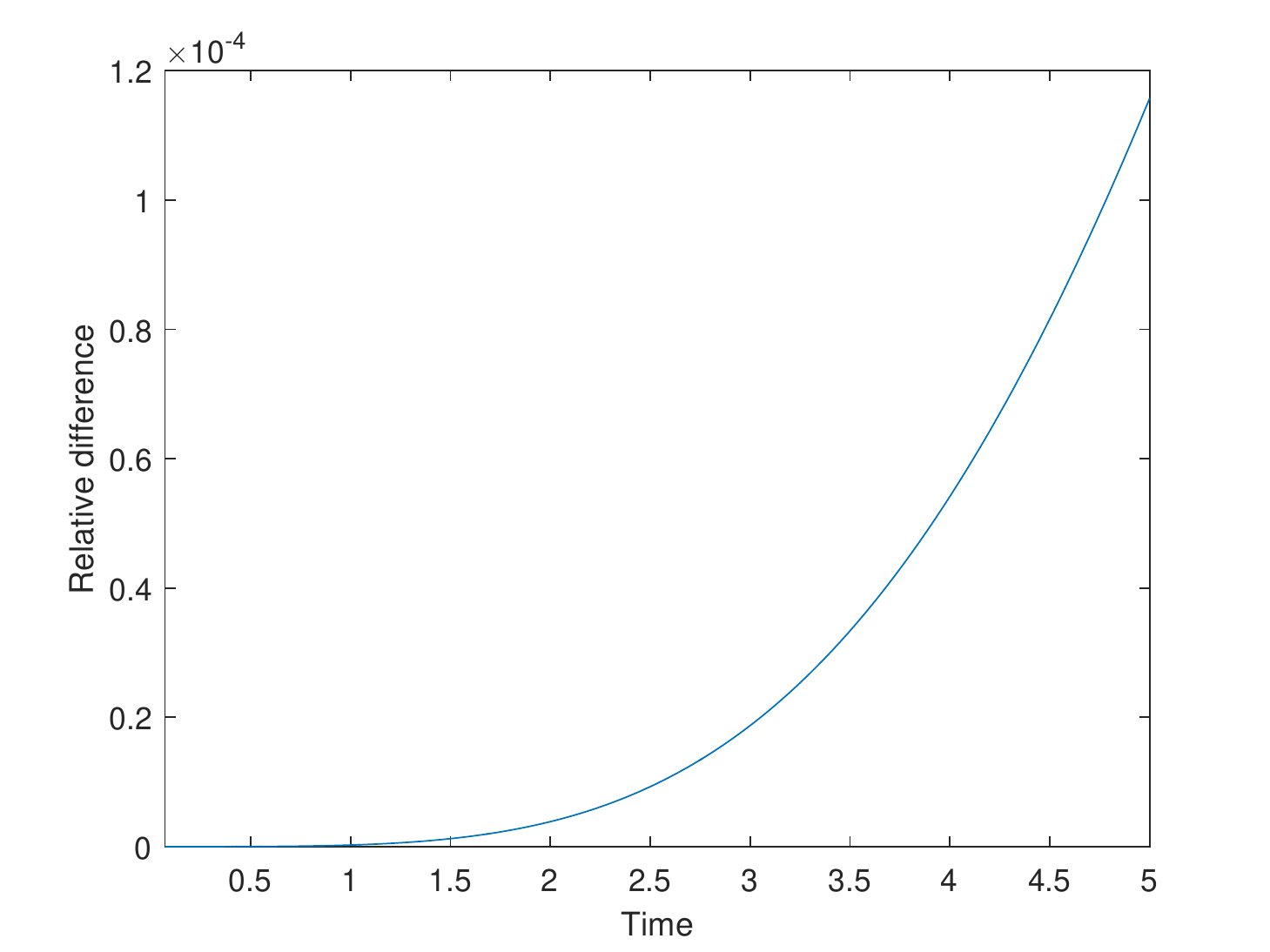}
	\end{minipage}
\caption{Influence of $T$ on the expected costs and relative difference}
\label{contInfT}
\end{figure}
Figure \ref{contInfT} shows the evolution of the expected costs and the relative difference when $T$ varies from $1/14$ (half an hour) to $5$ days.
The relative difference between the two strategies increases with the time horizon since the adapted strategy benefits more having more time to adapt. With a time horizon of a trading week, the relative difference is $1.2 \times 10^{-4}$.

As regards the influence of $\sigma$, Figure \ref{contInfS} shows the evolution of the expected costs and the relative difference when $\sigma$ varies from $0$ to $100\%$.
When $\sigma$ increases, the importance of using up to speed price information during the strategy  increases, since there is more uncertainty on what the new information will be. The adapted strategy takes incoming price information into account, unlike the deterministic one. Hence the relative difference increases as $\sigma$ increases. However, even when $\sigma=1$, which is equivalent to a gigantic annual volatility of $1588\%$, the relative difference between the two strategies is not even $0.1\%$. This seems to suggest that with this particular model the optimality does not change much when reducing the strategy class from adapted to deterministic. 
\begin{figure}[H]
	\begin{minipage}[c]{.54\linewidth}
		\includegraphics[scale=0.6]{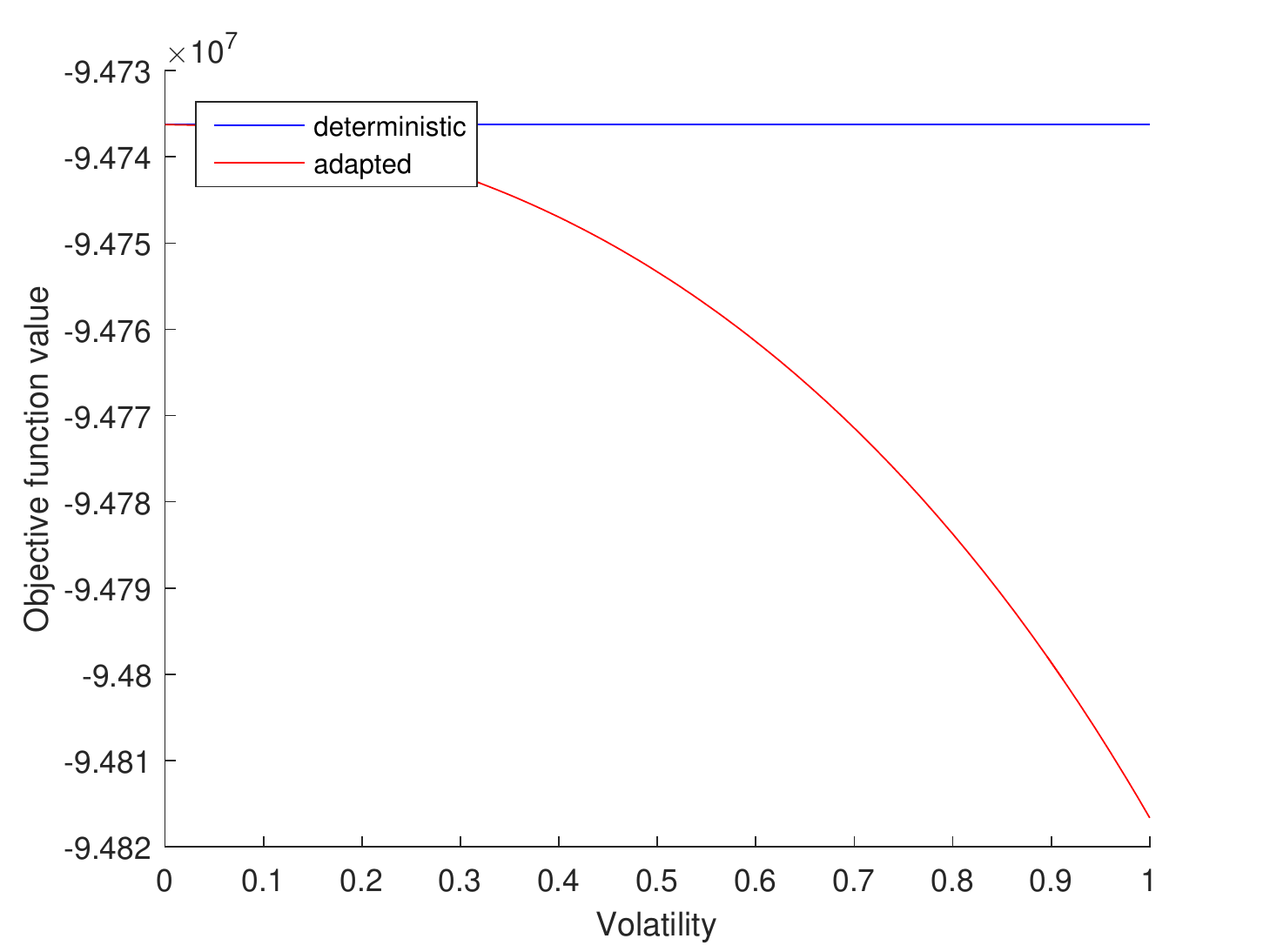}
	\end{minipage}
	\begin{minipage}[c]{.46\linewidth}
		\includegraphics[scale=0.6]{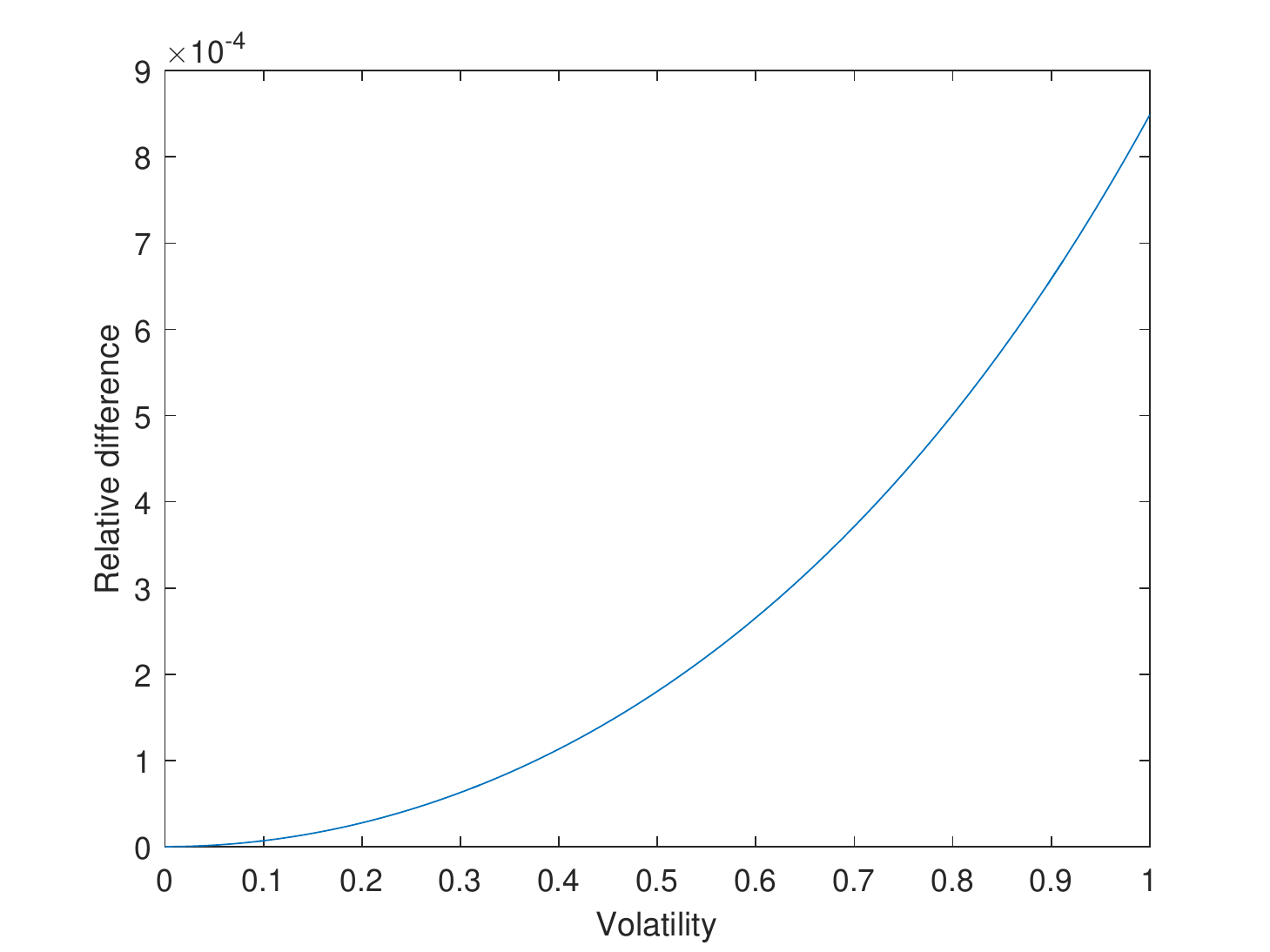}
	\end{minipage}
\caption{Influence of $\sigma$ on the expected costs and relative difference}
\label{contInfS}
\end{figure}

For the influence of  of $\gamma$, 
\begin{figure}[H]
	\begin{minipage}[c]{.54\linewidth}
		\includegraphics[scale=0.6]{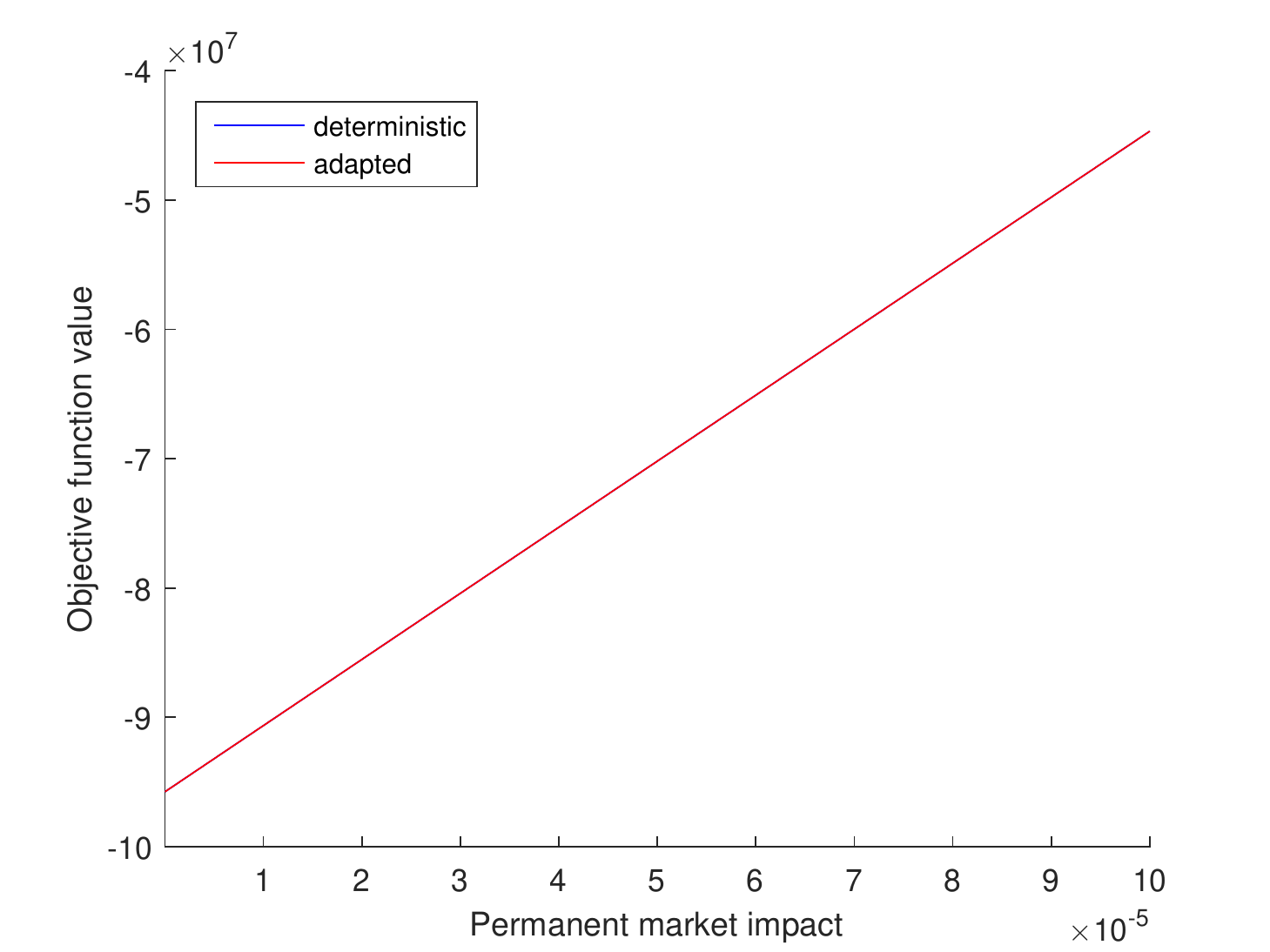}
	\end{minipage}
	\begin{minipage}[c]{.46\linewidth}
		\includegraphics[scale=0.6]{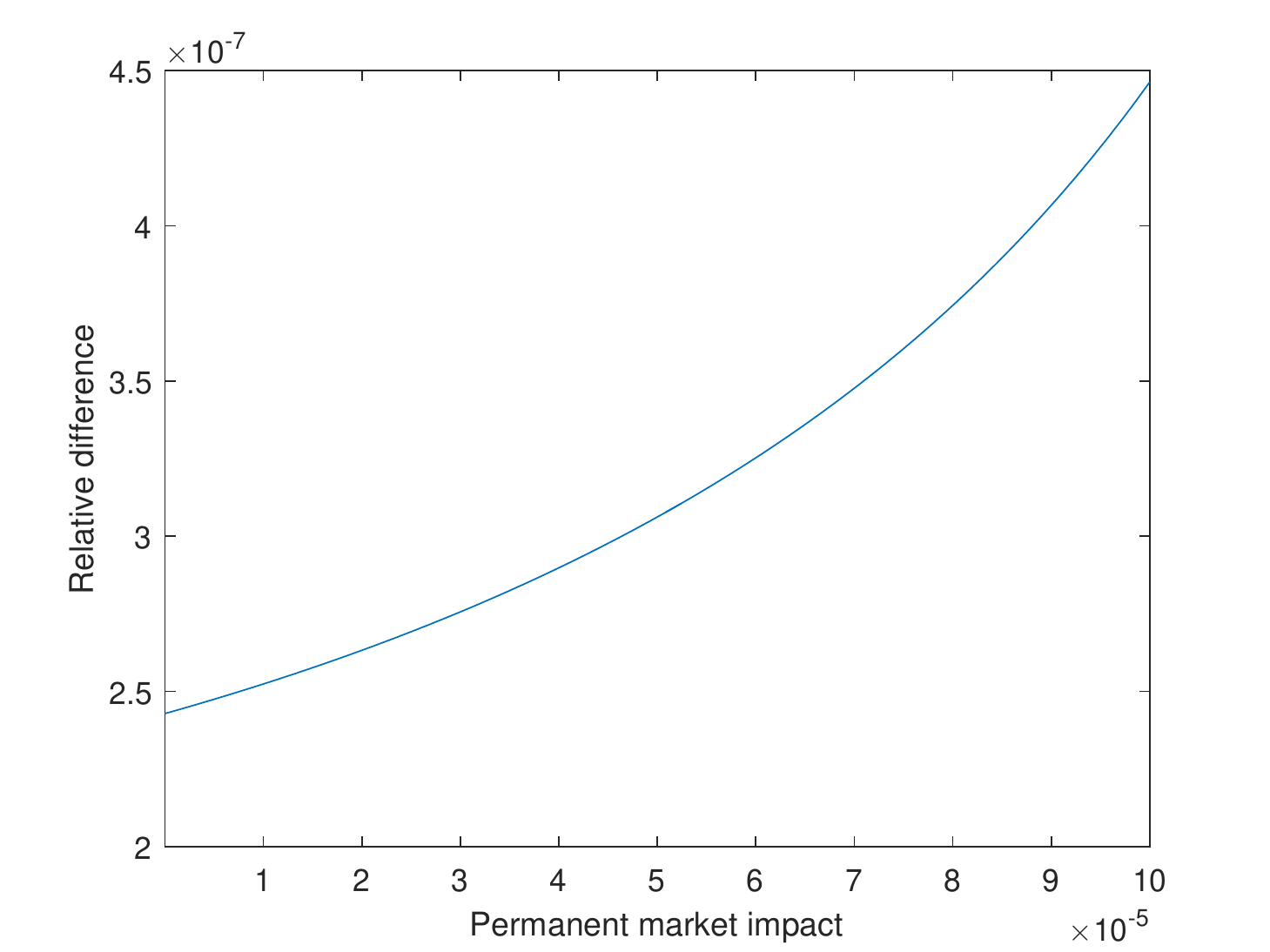}
	\end{minipage}
\caption{Influence of $\gamma$ on the expected costs and relative difference}
\label{contInfG}
\end{figure}
Figure \ref{contInfG} shows the evolution of the expected costs and the relative difference when $\gamma$ varies from $10^{-8}$ to $10^{-4}$.
The relative difference increases with the permanent impact parameter, unlike in the discrete time case. However it is always very small.

Consider now the influence of $\eta$.
\begin{figure}[H]
	\begin{minipage}[c]{.54\linewidth}
		\includegraphics[scale=0.6]{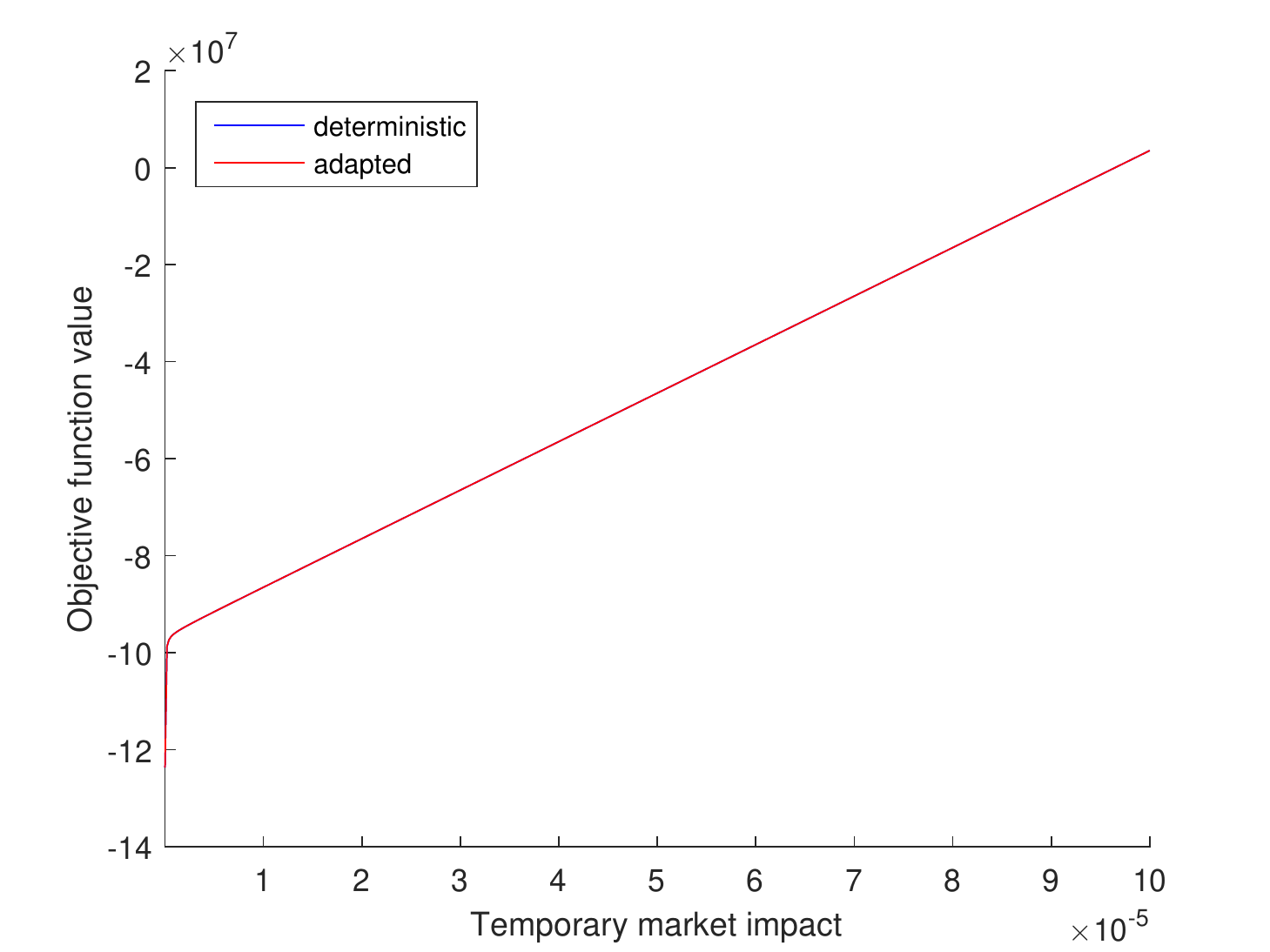}
	\end{minipage}
	\begin{minipage}[c]{.46\linewidth}
		\includegraphics[scale=0.6]{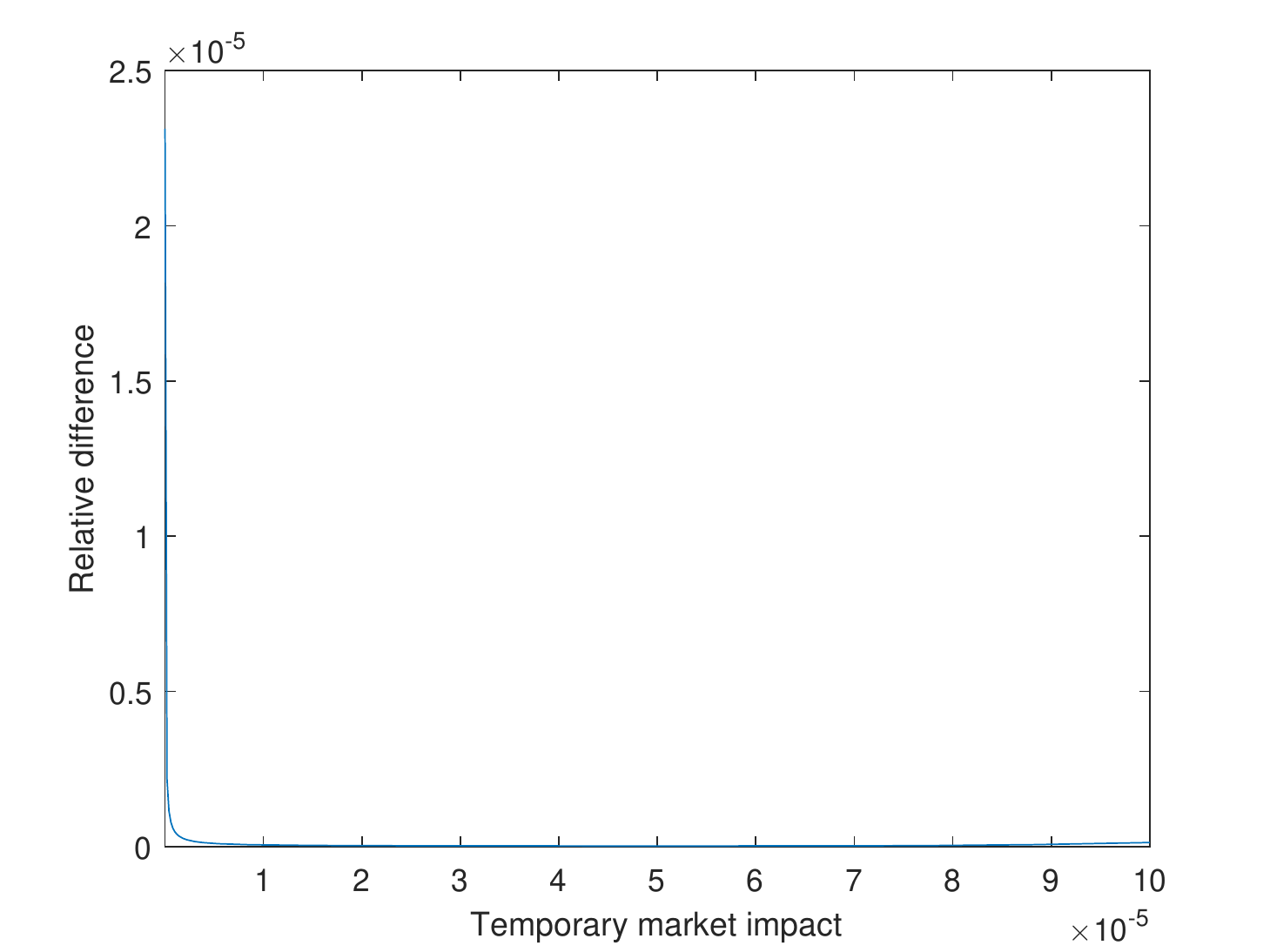}
	\end{minipage}
\caption{Influence of $\eta$ on the expected costs and relative difference}
\label{contInfET}
\end{figure}
Figure \ref{contInfET} shows the evolution of the expected costs and the relative difference when $\eta$ varies from $10^{-8}$ to $10^{-4}$.
The relative difference decreases with $\eta$ since both strategies become overwhelmed by a big temporary market parameter and have trouble reducing the cost by a noticeable margin. Even for $\eta=2 \times 10^{-7}$, which represents an increase of just $0.2\$$ per share over an instantaneous execution, the relative difference is just $2.2778 \times 10^{-6}$. Again, it looks like for this particular model optimality is practically attained already in the narrow class of static strategies. 

\begin{remark}
As in the setting of Bertsimas and Lo, the expected costs tend to $-\infty$ when $\eta$ or $\gamma$ tend to $0$. 
\end{remark}

Finally, we look at the influence of the risk aversion parameter $\widetilde{\lambda}$.

Figure \ref{contInfL} shows the evolution of the expected costs and the relative difference when $\widetilde{\lambda}$ varies from $10^{-5}$ to $10$.
The relative difference increases logarithmically with the risk aversion factor. When $\widetilde{\lambda}=10$, which is big as we have seen in Figure \ref{contSimL10}, the relative difference is $1.1 \times 10^{-4}$.
\begin{figure}[H]
	\begin{minipage}[c]{.54\linewidth}
		\includegraphics[scale=0.6]{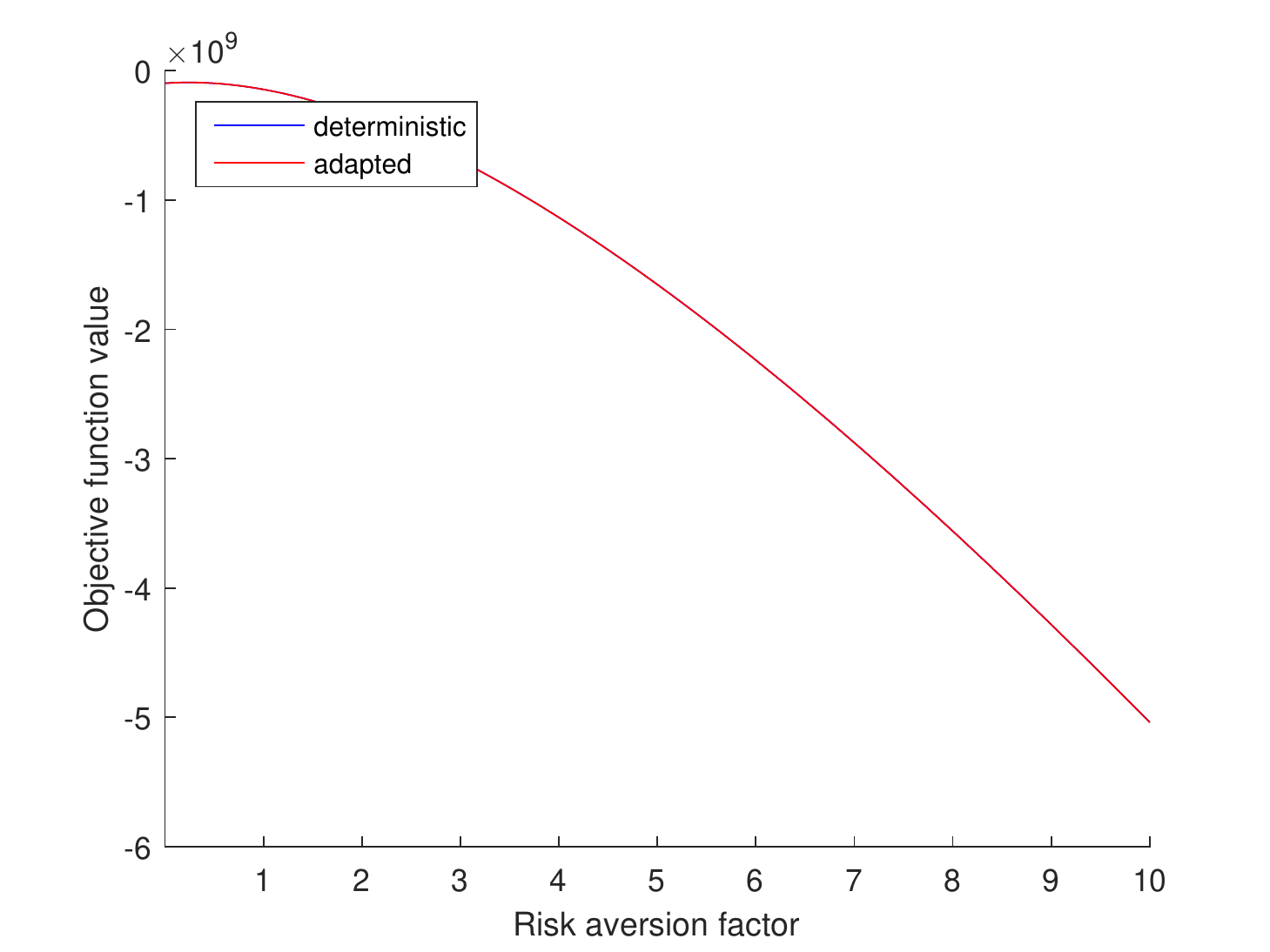}
	\end{minipage}
	\begin{minipage}[c]{.46\linewidth}
		\includegraphics[scale=0.6]{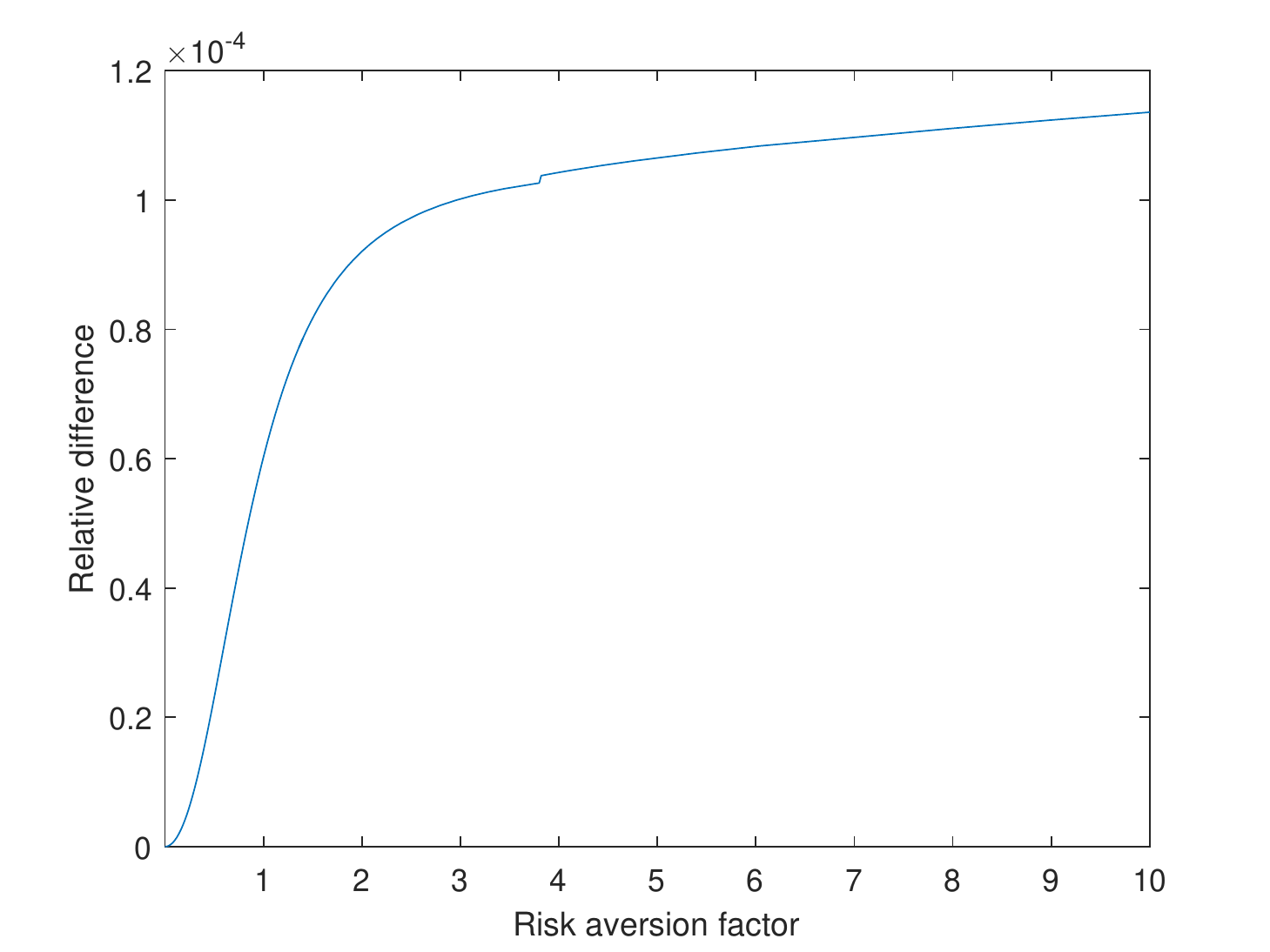}
	\end{minipage}
\caption{Influence of $\widetilde{\lambda}$ on the expected costs and relative difference}
\label{contInfL}
\end{figure}


\section{Conclusions and further research}\label{sec:conc}
We derived the optimal solutions to the trade execution problem in the two different classes of fully adapted trading strategies and deterministic ones, trying to assess how much optimality was lost when moving from the larger adapted class to the narrow static class. We did this in two different frameworks. The first was the discrete time framework of Bertsimas and Lo with an information flow process, dealing with both cases of permanent and temporary impact. The second framework was the continuous time framework of Gatheral and Schied, where the objective function is the sum of the expected cost and a value at risk (or expected shortfall) risk criterion. Optimal adapted solutions were known in both frameworks from the original works of these authors, \cite{bertsimas1998} and \cite{gatheral2011}. We derived the optimal static solutions for both approaches. We used those to study quantitatively the advantage gained by adapting our strategy instead of setting it entirely at time $0$.
Our conclusion is that there is no sensible difference, except for extreme cases that do not seem realistic. This seems to say that as long as we use simple models such as the benchmark models proposed here, it does not make much difference to search the solution in the larger adapted class, compared with the narrow static / deterministic class. This indirectly confirms that in the similar framework of Almgren and Chriss \cite{almgren2000} one is ok starting from a static solution, which happens to be more tractable, as is indeed done in that paper. 

In terms of further research, we might consider more recent models incorporating jumps,  as in \cite{alfonsi2014}, or considering daily cycles as in \cite{almgren2006}. It may happen that in those cases the difference between the optimal fully adapted solution and the static one is more sizeable.

\end{document}